\documentclass[twoside,leqno,twocolumn]{article}

\usepackage[letterpaper]{geometry}

\usepackage{ltexpprt}

\usepackage[numbers,sort&compress]{natbib}
\usepackage{grffile}
\usepackage{appendix}

\usepackage[utf8]{inputenc} 
\usepackage[T1]{fontenc}    
\usepackage{hyperref}       
\usepackage{url}            
\usepackage{subcaption}
\usepackage{booktabs}       
\usepackage{amsfonts}       
\usepackage{nicefrac}       
\usepackage{microtype}      

\usepackage{amsfonts} 
\usepackage{amsmath}
\usepackage{color}
\usepackage{subcaption}
\usepackage{array}
\usepackage{pgfplots}
\usepackage{tikz}

\usepackage{algorithmic}  
\usepackage[ruled,linesnumbered,algo2e,noend]{algorithm2e} 

\usepackage{float}
\usepackage{multirow}
\usepackage{multicol}
\usepackage{thmtools}
\usepackage[font=small]{caption}
\graphicspath{ {./imgs/} }
\newcommand{\hau}[1]{\textcolor{blue}{Hau: #1}}

\newtheorem{definition}{Definition}[section]

\DeclareMathOperator*{\argmax}{arg\,max}

\title{Maximizing Approximately $k$-Submodular Functions\thanks{Partially supported by NSFC Project 11771365.}}
\author{%
  Leqian Zheng\thanks{City University of Hong Kong, Hong Kong SAR  
   \texttt{leqizheng2-c@my.cityu.edu.hk}}  \hspace{+3mm} Hau Chan\thanks{University of Nebraska-Lincoln, USA,  \texttt{hchan3@unl.edu}}  \hspace{+3mm} Grigorios Loukides\thanks{King's College London, UK,  \texttt{grigorios.loukides@kcl.ac.uk}}  \hspace{+3mm} Minming Li\thanks{City University of Hong Kong, Hong Kong SAR, and City University of Hong Kong Shenzhen Research Institute, Shenzhen, P. R. China,  \texttt{minming.li@cityu.edu.hk}}
}
\date{}

\begin{document}

\maketitle

\begin{abstract}
We introduce the problem of maximizing \emph{approximately} $k$-submodular functions subject to size constraints. In this problem, one seeks to select $k$-disjoint subsets of a ground set with bounded total size or individual sizes, and maximum utility, given by a function that is ``close'' to being $k$-submodular. The problem finds applications in  tasks such as sensor placement, where one wishes to install $k$ types of sensors whose measurements are noisy, and influence maximization, where one seeks to advertise $k$ topics to users of a social network whose level of influence is uncertain. To deal with the problem, we first provide two natural definitions for approximately $k$-submodular functions and establish a hierarchical relationship between them. Next, we show that simple greedy algorithms offer approximation guarantees for different types of size constraints. Last, we demonstrate  experimentally that the greedy algorithms are effective in sensor placement and influence maximization problems. 
\end{abstract}

\section{Introduction}
A $k$-submodular function is a natural generalization of a submodular function to $k$ arguments or dimensions~\cite{Huber:2012aa}.  Recently, there has been an increased theoretical and algorithmic 
 interest in studying the problem of maximizing (monotone) 
\emph{$k$-submodular functions},  with~\cite{Ohsaka:2015aa} or  without \cite{Huber:2012aa,TALG16,iwata} size constraints. 
The problem finds applications in sensor placement~\cite{Ohsaka:2015aa}, influence maximization  \cite{Ohsaka:2015aa}, coupled feature selection~\cite{singh}, and network cut capacity optimization~\cite{iwata}. In these applications, one 
wishes to select $k$ \emph{disjoint} subsets from a ground set 
that maximize a given function with $k$ arguments subject to an upper bound on the total size of the selected subsets or the size of each individual selected subset~\cite{Ohsaka:2015aa}. Such a function of interest often has a \emph{diminishing returns property}  
with respect to each subset when fixing the other $k-1$ subsets \cite{Huber:2012aa}. 

More formally, let $V = \{1, ..., n\} = [n]$ be a finite ground set of $n$ elements 
and $(k+1)^V =\{(X_1, X_2, \dots, X_k)\mid X_i\subseteq V, \forall i\in [k], X_{i}\cap X_{j}=\emptyset\}$ 
be the set of $k$-disjoint subsets. 
A function $f:(k+1)^V \to \mathbb{R}^+$ is $k$-submodular if and only if, for any $\pmb{x},\pmb{y}\in(k+1)^{V}$, 
\begin{align*}
f(\pmb{x})+f(\pmb{y})\geq{}f(\pmb{x}\sqcap\pmb{y})+f(\pmb{x}\sqcup\pmb{y}), 
\end{align*}
where $\pmb{x}\sqcap\pmb{y}=(X_1\cap{}Y_1, X_2\cap{}Y_2,\hdots,X_k\cap{}Y_k)$ and 
{\small
\begin{align*}
\pmb{x}\sqcup\pmb{y}=(X_1\cup{}Y_1\setminus(\bigcup_{i\neq{1}}X_i\cup{}Y_i),\hdots,X_k\cup{}Y_k\setminus(\bigcup_{i\neq{k}}X_i\cup{}Y_i)). 
\end{align*}
\vspace{-2mm}
}

A $k$-submodular function $f$ is \emph{monotone} if and only if, for any $\pmb{x},\pmb{y}\in(k+1)^{V}$ such that 
${\pmb x}\preceq {\pmb y}$ (i.e., $X_i\subseteq Y_i,$ $\forall i\in [k]$), $f(\pmb{x}) \le f(\pmb{y})$. 
The problem of maximizing a monotone $k$-submodular function $f$ subject to the \emph{total size} (TS) and \emph{individual size} (IS) constraints \cite{Ohsaka:2015aa} is 
{\small
\begin{align*}
\max_{\pmb x \in (k+1)^V : |\cup_{i \in [k]} X_i|  \le B} f(\pmb x) ~~ \text{ and } ~~\max_{\pmb x \in (k+1)^V : |X_i|  \le B_i \; \forall i \in [k]} f(\pmb x), 
\end{align*}}
respectively, for some positive integers $B \in \mathbb{Z}^+$ and $B_i \in \mathbb{Z}^+$ for all $i \in [k]$. 
While recent results \cite{Ohsaka:2015aa} show that the above two problems can be well-approximated using 
greedy algorithms (with an approximation ratio of $\frac{1}{2}$ and $\frac{1}{3}$ for total size and individual size constraints, respectively) under the value oracle model, not much is known when the function is not entirely $k$-submodular, which can result from the oracle being inaccurate 
or the function itself not being $k$-submodular by  default (see examples and applications below). 
In fact, often there is no access to the exact value of the function but only to noisy values of it.   
In this paper, we initiate the study of the above maximization problems for non-$k$-submodular functions and pose the following questions: \begin{quote}
{\bf Q1.} \emph{How to define an approximately $k$-submodular function?}

{\bf Q2.} \emph{What approximation guarantees can be obtained when maximizing such a function under total size  
or individual size constraints?}
\end{quote}

\begin{table*} \centering
\resizebox{0.75\textwidth}{!}{%
\begin{tabular}{|c|c|c|c|c|}
\hline
\multicolumn{2}{|c|}{\multirow{2}{*}{}} & \multicolumn{2}{c|}{$\varepsilon$-AS} & \multicolumn{1}{c|}{$\varepsilon$-ADR} \\
\cline{3-4}
\multicolumn{2}{|c|}{$F$'s Solution}          & $F$'s Solution & $f$'s Solution &  \multicolumn{1}{c|}{$F$'s Solution} \\
\hline
\multicolumn{2}{|c|}{$k=1$}    & $\frac{1}{1+\frac{4B\varepsilon}{(1-\varepsilon)^2}}\left(1-\left(\frac{1-\varepsilon}{1+\varepsilon}\right)^{2B}\left(1-\frac{1}{B}\right)^B\right)$ \cite{Horel:2016aa} & 
$\frac{1-\varepsilon}{1+\varepsilon}\left(1-\frac{1}{e}\right)$ & 
$\left(1-e^{-\frac{(1- \varepsilon)}{(1+\varepsilon)}}\right)$
\\
\hline
\multirow{2}{*}{$k\geq 2$} & TS & $\frac{(1-\varepsilon)^2}{2(1-\varepsilon+\varepsilon B)(1+\varepsilon)}$ & $\frac{1-\varepsilon}{2(1+\varepsilon)}$ & $\frac{1-\varepsilon}{2}$ \\
\cline{2-5} 
                         &  IS & $\frac{(1-\varepsilon)^2}{(3-3\varepsilon+2\varepsilon B)(1+\varepsilon)}$ & $\frac{1-\varepsilon}{3(1+\varepsilon)}$ & $\frac{1-\varepsilon}{3+\varepsilon}$ \\
\hline
\end{tabular}
}
\vspace{+1mm}
\caption{Approximation ratios of the greedy algorithms of  \cite{Ohsaka:2015aa} for maximizing $\varepsilon$-approximately $k$-submodular ($\varepsilon$-AS) 
and $\varepsilon$-approximately diminishing returns ($\varepsilon$-ADR) function $F$ under a total size (TS) constraint $B$, or individual size (IS) constraints $B_1, ..., B_k$ with $B=\sum_{i \in[k]} B_i$, when they are applied to $F$ or to $f$. 
An algorithm has approximation ratio $\alpha \le 1$ if and only if it returns a (feasible) solution $\pmb{x}\in{(k+1)^V}$ such that $f(\pmb{x}) \ge \alpha f(\pmb{x}^*)$, where $\pmb{x}^*$ is an optimal solution.}
\label{table:results}
\vspace{-4mm}
\end{table*}

\noindent \textbf{Applications.} Answering Q1 and Q2 could be important in optimization, machine learning, and beyond. Let us first consider the case of $k=1$. In many applications, a function $F$ is \emph{approximately submodular} rather than submodular (i.e., $(1-\varepsilon)f(x)\leq F(x)\leq(1+\varepsilon)f(x)$, for an $\varepsilon>0$, where $f$ is a monotone submodular function)~\cite{Horel:2016aa,Hassidim:2017aa,Singer:2018aa}.  These applications include subset selection which is fundamental in areas such as (sequential) document summarization, sensor placement, and influence maximization~\cite{Ohsaka:2015aa,Elhamifar:2019aa}. For example, in sensor placement, the objective is to select a subset of good sensors and the approximation comes from  sensors producing noisy values due to hardware issues, environmental effects, and imprecision in measurement \cite{cleaningsensors,weightedentropy}. In influence  maximization, the objective is to select a subset of good users to start a viral marketing campaign over a social network and the approximation comes from our uncertainty about the level of influence of specific users in the social network~\cite{Li:2017aa,noveltydecay,hypergraph}. Other applications are PMAC learning \cite{Balcan:2011aa,Balcan:2012aa,Feldman:2013aa}, where the objective is to learn a submodular function, and sketching \cite{Badanidiyuru:2012aa}, where the objective is to find a good representation of a submodular function of polynomial size. 

In many of these applications, it is often needed to select $k$ disjoint subsets of given maximum total size or sizes, instead of a single subset, which gives rise to size-constrained $k$-submodular function  maximization~\cite{Ohsaka:2015aa}. This is the case for sensor placement where $k$ types of sensors need to be placed in disjoint subsets of locations, or influence maximization where $k$ viral marketing campaigns, each  starting from a different subset of users, need to be performed simultaneously over the same social network. Again, a function may not be exactly $k$-submodular, due to noise in sensor measurements and uncertainty in user influence levels.   
Thus, by answering question Q1, we could more accurately model the quality of solutions in these applications and, by answering Q2, we could obtain solutions of guaranteed quality. 

\vspace{+2mm}
\noindent \textbf{Our Contributions.} 
We address Q1 by introducing two natural definitions 
of an approximately $k$-submodular function.  To the best of our knowledge, approximately $k$-submodular functions have not been defined or studied before for general $k>1$. 
Namely, we define a function $F: (k+1)^V \to \mathbb{R}^+$ as $\varepsilon$-\emph{approximately $k$-submodular} ($\varepsilon$-AS) 
or \emph{$\varepsilon$-approximately diminishing returns} ($\varepsilon$-ADR) for some small  $\varepsilon > 0$, 
if and only if 
there exists a monotone \textit{k-submodular} function $f$ such that for any $\pmb{x}\in{(k+1)^V}$, $u \not\in \bigcup_{l \in [k]} X_l$, $u\in V$ and $i\in[k]$ 

\vspace{-5mm}
\begin{align*}
\varepsilon\text{-AS}: &  ~(1-\varepsilon)f(\pmb{x})\leq F(\pmb{x})\leq(1+\varepsilon)f(\pmb{x}) \text{ or } \\
\varepsilon\text{-ADR}: &  ~(1-\varepsilon)\Delta_{u,i}f(\pmb{x})\leq{}\Delta_{u,i}F(\pmb{x})\leq{}(1+\varepsilon)\Delta_{u,i}f(\pmb{x}), 
\end{align*}
where $\Delta_{u,i}f({\pmb x})=f(X_1, \ldots, X_{i-1}, X_i\cup \{u\}, X_{i+1}, \ldots, X_k) -f(X_1, \ldots, X_k)$ and $\Delta_{u,i}F(\pmb x)$ is defined similarly. Our $\varepsilon$-AS definition generalizes 
the approximately submodular definition of \cite{Horel:2016aa} from $k=1$ to $k
\geq 1$ dimensions. 
The $\varepsilon$-ADR definition is related to the \emph{marginal gain} of the functions $F$ and $f$ 
and implies $\alpha$-submodularity when $k=1$ \cite{Halabi:2018aa}. As we will show, an $\varepsilon$-ADR function $F$ is also $\varepsilon$-AS. However, the converse is not true. 
Thus, we establish a novel approximately $k$-submodular hierarchy for $k \ge 1$. 

We address Q2 by considering the maximization problems on a function $F$ that is $\varepsilon$-AS or $\varepsilon$-ADR, 
subject to the total size (TS) or individual size (IS) constraint. We show that, when applying 
the simple greedy algorithms of \cite{Ohsaka:2015aa}
for TS and IS, we obtain approximation guarantees that depend on $\varepsilon$. 
Table \ref{table:results} provides an overview of our results. 
Note that there are two cases with respect to an $\varepsilon$-AS $F$ and $f$ 
where we can apply greedy algorithms to $F$ or $f$ (when $f$ is known). 
A surprising observation is that we can derive better approximation ratios 
by applying the greedy algorithms to $f$ and using the solutions to approximate $F$.  
When only $F$ is known (e.g., through a noisy oracle, an approximately learned submodular function, or a sketch), we provide approximation ratios of the greedy algorithms being applied to $F$. When the function $F$ is $\varepsilon$-ADR, we provide better approximation guarantees, compared to when $F$ is $\varepsilon$-AS. 

We conduct experiments on two real datasets to evaluate the effectiveness of the greedy algorithms 
in a $k$-type sensor placement problem and a $k$-topic influence maximization problem for our setting. 
We provide three approximately $k$-submodular function generation/noise methods for generating function  values. 
Our experimental results are consistent with the theoretical findings of the derived approximation ratios and showcase the impact of the type of noise on the quality of the solutions. 

\vspace{+1mm}
\noindent \textbf{Organization.} 
Section~\ref{sec:related} discusses related work. 
Section \ref{section:prelim} provides some 
preliminaries and establishes a strong relationship between 
$\varepsilon$-AS and $\varepsilon$-ADR. Section  \ref{section:k1} and Section \ref{section:kg1} considers the approximation ratios of the greedy algorithms for maximizing an $\varepsilon$-AS and $\varepsilon$-ADR function for $k=1$ and $k > 1$, respectively.  
Section \ref{sec:paper:improvedgreedy} discusses the approximation ratios of $f$'s greedy solutions to $F$. 
Section \ref{section:experiment} presents experimental results. In Section~\ref{section:discussion}, we conclude and present extensions.

\section{Related Work.}\label{sec:related}
The works of \cite{Hassidim:2017aa,Singer:2018aa} and \cite{Horel:2016aa} consider the problem of maximizing an approximately submodular function  (e.g., $\varepsilon$-AS definition with $k = 1$)  
subject to constraints under stochastic errors (e.g., $\varepsilon$) drawn i.i.d. from some distributions 
and sub-constant (bounded) errors, respectively. 
Our paper focuses on the latter setting with $k\geq 1$ and considers the performance of greedy algorithms with respect to errors (e.g., a fixed $\varepsilon$). 
In particular, \cite{Horel:2016aa} shows that, for some $\varepsilon = \frac{1}{n^{1/2-\beta}}$ for $\beta > 0$, 
no algorithm can obtain any constant approximation ratio using polynomially  many queries of the value oracle.  
The same lower bound can be applied to our setting. 
However, when $\varepsilon$ is sufficiently small with respect to the size constraint (e.g., $\varepsilon = \frac{\delta}{B}$), 
the standard greedy algorithm~\cite{Nemhauser:1978aa} provides an (almost) constant approximation ratio of $(1 - 1/e - O(\delta))$. 

Under the stochastic noise assumptions, \cite{Hassidim:2017aa} shows that, 
when the size constraint $B \in \Omega(\log\log n)$ is sufficiently large, there is an algorithm that achieves 
$(1- 1/e)$ approximation ratio with high probability (w.h.p.) and also shows some impossibility results for the existence of randomized algorithms with a good approximation ratio using a polynomial number of queries w.h.p. 
The work of \cite{Singer:2018aa} shows that good approximation ratios are attainable  
(e.g., an approximation ratio of $(1-1/e)$) w.h.p. and/or  in expectation 
for arbitrary size constraint $B$ using greedy algorithms (with or without randomization) under different size constraints. 
Additionally, \cite{Singer:2018aa} considers general matroid constraints  and provides an approximation ratio that depends on the matroid constraints. 
Another work \cite{Qian:2017aa} considers the approximate function maximization problem (with $f$ that is not necessarily submodular) under the cardinality constraint and derives several approximation ratios that depend on the submodularity ratio \cite{Das:2011aa}, 
$\varepsilon$, and $B$ using the greedy algorithm, as well as an algorithm based on the Pareto Optimization for Subset Selection (POSS) strategy. 

A recently published work \cite{Nguyen:2020aa} explores the problem of maximizing approximate $k$-submodular functions in the streaming context,  where the elements of $V$ are scanned at most once, under the total size constraint. 
The approximate notion, introduced independently, is that of our $\varepsilon$-AS. However,~\cite{Nguyen:2020aa} aims to optimize $f$ instead of $F$, which is the main focus of our work, and it does not consider $\varepsilon$-ADR nor individual size constraints.  

\section{Preliminaries} \label{section:prelim}
We denote the vector of $k$ empty subsets by ${\bf 0}=(X_1=\{\}, \ldots, X_k=\{\})$.  
Without loss of generality, we assume that the functions $F$ and $f$ are normalized such that
$F(\textbf{0}) = f(\textbf{0}) = 0$. 
We 
define $k$-submodular functions as those that are monotone and have a diminishing returns property in each dimension \cite{TALG16}. 

\vspace{-1mm}
\begin{definition}[$k$-submodular function~\cite{TALG16}] A function $f:(k+1)^V\rightarrow \mathbb{R}$ is $k$-\emph{submodular} if and only if:
(a) $\Delta_{u,i}f({\pmb x})\geq \Delta_{u,i}f({\pmb y})$, for all ${\pmb x}, {\pmb y}\in (k+1)^V$ with ${\pmb x}\preceq{\pmb y}$, $u\notin \cup_{\ell \in [k]}Y_i$,
and $i\in [k]$, and
(b) $\Delta_{u,i}f({\pmb x})+\Delta_{u,j}f({\pmb x})\geq 0$, for any
${\pmb x}\in (k+1)^V$, $u\notin \cup_{\ell \in [k]}X_i$, and $i,j\in [k]$ with $i\neq j$.\label{def:ksubgain}
\end{definition}
\vspace{-1mm}

Part (a) of Definition \ref{def:ksubgain} is known as the \emph{diminishing returns} property
and part (b) as \emph{pairwise monotonicity}. We assume that the function is monotone, and this implies pairwise monotonicity directly. Moreover, $k$-submodular functions are \emph{orthant submodular}~\cite{TALG16}. 

For consistency, we use the following notations as used by \cite{Ohsaka:2015aa}. 
Namely, with a slight abuse of notation, we associate each $\pmb{x} = (X_1, \hdots, X_k)\in (k+1)^V$  with $\pmb{x}\in \{0,1,\hdots,k\}^V$ where  
$X_i=\{e\in{}V\mid \pmb{x}(e)=i\}$ for $i\in[k]$. We define the \emph{support} or \emph{size} of $\pmb{x}\in(k+1)^V$ 
as $supp(\pmb{x})=\{e\in{}V\mid \pmb{x}(e)\neq{0}\}$. Similarly, we define $supp_{i}(\pmb{x}) =\{e\in{}V\mid{}\pmb{x}(e)=i\}$.
    
In the following, we show that when the function $f$ is $\varepsilon$-ADR, the function is also $\varepsilon$-AS. The proof is in Appendix~\ref{appendA}.

\begin{theorem}\label{lemma1}
If $F$ is $\varepsilon$-approximately diminishing returns, 
then $F$ is 
$\varepsilon$-approximately $k$-submodular.    
\end{theorem}

The converse is not true (as shown below). 
One could attempt to upper bound 
$\Delta_{u,i}F(\pmb{x}) \le (1+\varepsilon)f(X_1, ..., X_i \cup\{u\}, ..., X_k) - (1-\varepsilon)f(\pmb{x}) 
= (1+\varepsilon)\Delta_{u,i}f(\pmb{x}) + 2\varepsilon f(\pmb{x}) $ 
using the definition of $\varepsilon$-AS for each term. However, the resultant function does not appear to be $\varepsilon$-ADR. 

\begin{theorem}
If $F$ is 
$\varepsilon$-approximately $k$-submodular, 
then $F$ is not necessarily 
$\varepsilon$-approximately diminishing returns.    
\end{theorem}
\begin{proof}
To prove the 
statement, we construct a function $F$ that is $\varepsilon$-AS but not $\varepsilon$-ADR for $k=1$. 
Let $V = \{e_1, e_2\}$. By the $\varepsilon$-AS definition, we have that $(1-\varepsilon)f(\pmb{x})\leq F(\pmb{x})\leq(1+\varepsilon)f(\pmb{x})$ for any $\pmb{x}\in{(k+1)^V}$. 
We define $F$ partially as follows. 
Let $F(\{e_1\})=(1+\varepsilon)f(\{e_1\})$ and $F(\{e_1, e_2\})=(1-\varepsilon)f(\{e_1, e_2\})$. 
Consider the marginal gain of adding $e_2$ to the set $\{e_1\}$. 
We have that 
$\Delta_{e_2}F(\{e_1\})=(1-\varepsilon)f(\{e_1, e_2\})-(1+\varepsilon)f(\{e_1\})<(1-\varepsilon)\left[f\left(\{e_1, e_2\}\right)-f\left(\{e_1\}\right)\right]=(1-\varepsilon)\Delta_{e_2}f(\{e_1\})$, 
where the first equality holds due to our construction, the inequality is due to $\varepsilon>0$, and the last equality is from the lower bound of the $\varepsilon$-ADR definition. 
\end{proof}

\section{Approximately $k$-Submodular Function Maximization: $k=1$} \label{section:k1}
When $k = 1$, $k$-submodular functions coincide with the standard submodular functions~\cite{Nemhauser:1978aa}.  
Thus, the total size and individual size constraints are the same. 
We are interested in the maximization problem with a size constraint with parameter $B$. 
For an $\varepsilon$-AS $F$, \cite{Horel:2016aa} proves the following approximation ratio 
when applying the greedy algorithm~\cite{Nemhauser:1978aa}, which iteratively adds into $F$ an element that achieves the highest marginal gain.

\begin{corollary} [Theorem 5 \cite{Horel:2016aa}] \label{coro1} 
Suppose $F$ is $\varepsilon$-approximately submodular. The greedy algorithm provides an approximation ratio of  
$\frac{1}{1+\frac{4B\varepsilon}{(1-\varepsilon)^2}}\left(1-\left(\frac{1-\varepsilon}{1+\varepsilon}\right)^{2B}\left(1-\frac{1}{B}\right)^B\right)$ for the size constrained maximization problem.
\end{corollary}
\vspace{-1mm}

When $F$ is $\varepsilon$-ADR, we obtain the following approximation ratio using its connection to the notions of a-submodularity, weakly submodularity, 
and submodularity ratio (see e.g., \cite{Halabi:2018aa,Das:2011aa,Bian:2017aa}). 

\vspace{-2mm}
\begin{theorem} \label{newtheo:subk1}
Suppose $F$ is $\varepsilon$-approximately diminishing returns. The greedy algorithm provides an approximation ratio of 
$(1-e^{-a})$ for the size constrained maximization problem where $a = (1-\varepsilon)/(1+\varepsilon)$. 
\end{theorem}
\vspace{-4mm}

\begin{proof}
We note that if $F$ is $\varepsilon\text{-ADR}$, then 
$F$ is $a$-submodular \cite{Halabi:2018aa} via the fact that 
$\varepsilon\text{-ADR}:  
(1-\varepsilon)\Delta_{u,i}f(\pmb{x})\leq{}\Delta_{u,i}F(\pmb{x})\leq{}(1+\varepsilon)\Delta_{u,i}f(\pmb{x}),$ 
which implies 
$\Delta_{u,i}F(\pmb{y})(1-\varepsilon)/(1+\varepsilon) \le
(1-\varepsilon)\Delta_{u,i}f(\pmb{y}) \le (1-\varepsilon)\Delta_{u,i}f(\pmb{x})\leq \Delta_{u,i}F(\pmb{x}) 
$ 
for any $\pmb{x} \le \pmb{y}, u \not\in \bigcup_{l \in [k]} X_l$ and $i\in[k]$ with 
$a = (1-\varepsilon)/(1+\varepsilon)$
(i.e., 
a function, say $g$, is $a$-submodular if and only if $ \Delta_{u,i}g(\pmb{x}) \ge a \Delta_{u,i}g(\pmb{y})$ for any $\pmb{x} \le \pmb{y}$ \cite{Halabi:2018aa}). 
As any $a$-submodular function is weakly submodular (Proposition 8 \cite{Halabi:2018aa}), the greedy algorithm provides an approximation of $(1-e^{-a})$ \cite{Das:2011aa,Bian:2017aa}.
\end{proof}

\section{Approximately $k$-Submodular Function Maximization: $k>1$} \label{section:kg1}
In this section, we consider the problems of maximizing approximately $k$-submodular functions 
under the $\varepsilon$-AS and $\varepsilon$-ADR definitions subject to the total size, $\max\limits_{\pmb{x}:|supp(\pmb{x})|\leq{B}}F(\pmb{x})$, 
or individual size constraints, $\max\limits_{\pmb{x}:|supp_i(\pmb{x})|\leq{B_i},\forall{}i\in[k]}F(\pmb{x})$, for some function $F$ 
and $B, B_1, ..., B_k \in \mathbb{Z}^+$. 
We show that the greedy algorithms \cite{Ohsaka:2015aa} $k$-Greedy-TS (see Algorithm \ref{algorithm1}) 
and $k$-Greedy-IS (see Algorithm \ref{algorithm2}) provide (asymptotically tight as $\varepsilon\to 0$ for TS) approximation ratios to function $F$. 
Algorithm \ref{algorithm1} and Algorithm \ref{algorithm2}  essentially add a single element with the highest marginal gain to one of the $k$ subsets at each iteration without violating the TS and IS constraints, respectively. 
Algorithm \ref{algorithm1} and Algorithm \ref{algorithm2} requires evaluating the function $O(knB)$ and $O(kn\sum_{i\in[k]}B_i)$ times, respectively.   

\begin{algorithm2e}\small
     \KwIn{a $\varepsilon$-approximately $k$-submodular function $F:(k+1)^V\mapsto\mathbb{R}^+$ 
     and $B\in\mathbb{Z}^+$.}
     \KwOut{a vector $\pmb{x}$ with $|supp(\pmb{x})|=B$.}

     $\pmb{x}\leftarrow\pmb{0}$\;
     
     \For{$j=1$ \KwTo $B$}{
         $(e, i)\leftarrow{}\argmax_{e\in{}V\setminus{supp(\pmb{x})},i\in[k]}\Delta_{e,i}F(\pmb{x})$\;
         
         $\pmb{x}(e)\leftarrow{}i$\;
     } \caption{$k$-Greedy-TS (Total Size)}
     \label{algorithm1}
 
 \end{algorithm2e}
 \vspace{-2mm}
 \begin{algorithm2e}\small
    \caption{$k$-Greedy-IS (Individual Size)}
    \label{algorithm2}
    \KwIn{a $\varepsilon$-approximately $k$-submodular function $F:(k+1)^V\mapsto\mathbb{R}^+$ 
    and $B_1, \cdots, B_k\in\mathbb{Z}^+$.}
    \KwOut{a vector $\pmb{x}$ with $|supp_i(\pmb{x})|=B_i$ $\forall i\in[k]$.}
    $\pmb{x}\leftarrow\pmb{0}$; $I\leftarrow[k]$\;
    
    \While{$I\neq\emptyset$}{
        $(e, i)\leftarrow{}\argmax_{e\in{}V\setminus{}supp(\pmb{x}),i\in I}\Delta_{e,i}F(\pmb{x})$\;
        
        $\pmb{x}(e)\leftarrow{}i$\;
        
        \If{$|supp_i(\pmb{x})=B_i|$}{
            $I\leftarrow{}I\setminus\{i\}$\;
        }
    }
\end{algorithm2e}

The proof techniques in this section use similar ideas from \cite{Ohsaka:2015aa}. 
However, the proofs in \cite{Ohsaka:2015aa} do not apply trivially and directly without our carefully designed lemmas and appropriate derivations. 

\subsection{Maximizing $\varepsilon$-AS and $\varepsilon$-ADR Functions with the TS Constraint} \label{section1} 
  
We consider the problem of maximizing $\varepsilon$-AS or $\varepsilon$-ADR function $F$ 
subject to the total size constraint $B$ using Algorithm \ref{algorithm1} on function $F$.  
We use the following notations as in \cite{Ohsaka:2015aa}. 
Let $\pmb{x}^{(j)}$ be the solution after the $j$-th iteration of Algorithm \ref{algorithm1}. 
For each $j$, we let $(e^{(j)}, i^{(j)})\in{}V\times[k]$ be the selected pair. 
Let $\pmb{o} \in \argmax\limits_{\pmb{x}:|supp(\pmb{x})|\leq{B}}F(\pmb{x})$ be an optimal solution.

Our goal is to compare the greedy solution $\pmb{x}$ and an optimal solution $\pmb{o}$. 
To begin, we define $\pmb{o}^{(0)}=\pmb{o}, \pmb{o}^{(\frac{1}{2})}, \pmb{o}^{(1)}, 
\cdots, \pmb{o}^{(B)}$ iteratively. 
Let $S^{(j)}=supp(\pmb{o}^{(j-1)})\setminus{}supp(\pmb{x}^{(j-1)})$ with $\pmb{x}^{(0)}=\pmb{0}$. 
We set $o^{(j)}$ to be an arbitrary element in $S^{(j)}$ if $e^{(j)}\notin{}S^{(j)}$, and set $o^{(j)}=e^{(j)}$ otherwise.
We construct $\pmb{o}^{(j-\frac{1}{2})}$ from $\pmb{o}^{(j-1)}$ by assigning $0$ to the $o^{(j)}$-th element. 
Next, we define $\pmb{o}^{(j)}$ from $\pmb{o}^{(j-\frac{1}{2})}$ by assigning $i^{(j)}$ to the $e^{(j)}$-th element. 
As such, we have $|supp(\pmb{o}^{(j)})|=B$ for all $j\in [B]$ and $\pmb{o}^{(B)}=\pmb{x}^{(B)}=\pmb{x}$. 
Finally, we note that $\pmb{x}^{(j-1)}\preceq\pmb{o}^{(j-\frac{1}{2})}$ for every $j\in[B]$.

\subsubsection{$\varepsilon$-AS Functions with the TS Constraint} 

\begin{theorem} \label{theorem:eASTS}
Suppose $F$ is $\varepsilon$-approximately $k$-submodular. The $k$-Greedy-TS algorithm 
provides an approximation ratio of 
$\frac{(1-\varepsilon)^2}{2(1-\varepsilon+\varepsilon{}B)(1+\varepsilon)}$ for the total size constrained maximization problem. \label{thm:greedyonF}
\end{theorem} 

To prove the above theorem, we first need to prove the following key lemma 
(see Appendix \ref{sec:algandproofs}).

\begin{lemma} \label{lemma:ineq1}
For any $j \in [B]$, $ \frac{1+\varepsilon}{1-\varepsilon}f(\pmb{x}^{(j)})-f(\pmb{x}^{(j-1)})\geq{}f(\pmb{o}^{(j-1)})-f(\pmb{o}^{(j)})$. 
\end{lemma}

\begin{proof}[Proof of Theorem \ref{theorem:eASTS}] We have
\begin{equation*}
{\small
\begin{aligned}
&f(\pmb{o})-f(\pmb{x})
=\sum_{j\in[B]}\left(f(\pmb{o}^{(j-1)})-f(\pmb{o}^{(j)})\right)\\ &\leq
\sum_{j\in[B]}\left(\frac{1+\varepsilon}{1-\varepsilon}f(\pmb{x}^{(j)})-f(\pmb{x}^{(j-1)})\right)\\
&=\frac{1+\varepsilon}{1-\varepsilon}f(\pmb{x})+\sum_{j=1}^{B-1}\left(\frac{2\varepsilon}{1-\varepsilon}f(\pmb{x}^{(j)})\right) \\
&\leq\frac{1+\varepsilon}{1-\varepsilon}f(\pmb{x})+\sum_{j=1}^{B-1}\left(\frac{2\varepsilon}{1-\varepsilon}f(\pmb{x})\right)=\frac{1-\varepsilon+2\varepsilon{}B}{1-\varepsilon}f(\pmb{x}), 
\end{aligned}}
\end{equation*}
where the first inequality is by Lemma \ref{lemma:ineq1}. Thus,  $f(\pmb{x})\geq\frac{1-\varepsilon}{2-2\varepsilon+2\varepsilon{}B}f(\pmb{o})$. 
We have $F(\pmb{x})\geq\frac{(1-\varepsilon)^2}{2(1-\varepsilon+\varepsilon{}B)(1+\varepsilon)}F(\pmb{o})$ 
after applying the $\varepsilon$-AS definition. 
\end{proof}

\subsubsection{$\varepsilon$-ADR Functions with the TS Constraint} 

\begin{lemma} \label{lemma:ineq2}
For any $j \in [B]$, it holds that\\ $\frac{1}{1-\varepsilon}\left[F(\pmb{x}^{(j)})-F(\pmb{x}^{(j-1)})\right]
    \geq{}\frac{1}{1+\varepsilon}\left[F(\pmb{o}^{(j-1)})-F(\pmb{o}^{(j)})\right]$. 
\end{lemma}
\begin{proof}
See Appendix~\ref{sec:algandproofs}.
\end{proof}

\begin{theorem} \label{theorem3}
Suppose $F$ is $\varepsilon$-approximately diminishing returns. The $k$-Greedy-TS algorithm provides an approximation ratio of 
$\frac{1-\varepsilon}{2}$ for the total size constrained maximization problem.
\end{theorem}

\begin{proof}
We have 
{\small
\begin{align*}
F(\pmb{o})-F(\pmb{x})&=\sum_{j\in[B]}\left[F(\pmb{o}^{(j-1)})-F(\pmb{o}^{j})\right]\\
&\leq{}\frac{1+\varepsilon}{1-\varepsilon}\sum_{j\in[B]}\left[F(\pmb{x}^{(j)})-F(\pmb{x}^{j-1})\right]\leq\frac{1+\varepsilon}{1-\varepsilon}F(\pmb{x}), 
\end{align*}}
\noindent where the first inequality is due to Lemma \ref{lemma:ineq2}. 
Thus, we have $F(\pmb{x})\geq{}\frac{1-\varepsilon}{2}F(\pmb{o})$.
\end{proof}

\subsection{Maximizing $\varepsilon$-AS and $\varepsilon$-ADR Functions with the IS Constraints} 
We consider the problem of maximizing an $\varepsilon$-AS or $\varepsilon$-ADR function $F$ 
subject to the individual size constraints $B_1, ..., B_k$ using Algorithm \ref{algorithm2} on the function $F$.  
In the individual size constraints maximization problem, we are given $B_1, ..., B_k$ 
restricting the maximum number of elements one can select for each subset.  
We define $B = \sum_{j \in [k] } B_j$. 
We simply state our main results here 
(see Appendix \ref{sec:ISconstraint} for key lemmas and proofs). 

\begin{theorem}
Suppose $F$ is $\varepsilon$-approximately $k$-submodular. The $k$-Greedy-IS algorithm provides an approximation ratio of 
$\frac{(1-\varepsilon)^2}{(3-3\varepsilon+2\varepsilon{}B)(1+\varepsilon)}$ for the individual size constrained maximization problem.   \label{theorem:ASIS}
\end{theorem}    

\begin{theorem} \label{theorem:ADRIS}
Suppose $F$ is $\varepsilon$-approximately diminishing returns. The $k$-Greedy-IS algorithm provides an approximation ratio of 
$\frac{1-\varepsilon}{3+\varepsilon}$ for the individual size constrained maximization problem. 
\end{theorem}        

\section{Improved Greedy Approximation Ratios When $f$ is Known}\label{sec:paper:improvedgreedy}

In this section, we consider the case where $f$ is a known monotone function that  can be constructed directly. We investigate the question of 
whether we can use such information to obtain alternative, possibly better, approximation ratios for maximizing an $\varepsilon$-AS or $\varepsilon$-ADR function $F$ 
subject to the total size or individual size constraints.  
We answer the question affirmatively via the following result 
 (see Appendix~\ref{sec:improvedgreedy} for details).

\vspace{-1mm}
\begin{theorem} \label{theo:ftoFpaper}
Let $f$ be a $k$-submodular function and $F$ be an $\varepsilon$-approximately $k$-submodular function that is bounded by $f$. 
If there is an algorithm that provides an approximation ratio of $\alpha$ for maximizing $f$ subject to constraint $\mathbb{X}$, 
then the same solution yields an approximation ratio of $\frac{1-\varepsilon}{1+\varepsilon} \alpha$ for maximizing $F$ subject to constraint $\mathbb{X}$. 
\end{theorem}
\vspace{-1mm}

This theorem implies the following: (1) By applying the greedy algorithm~\cite{Nemhauser:1978aa} on a submodular function $f$ subject to a size constraint, we obtain a solution providing an approximation 
ratio of $\frac{1-\varepsilon}{1+\varepsilon} (1-\frac{1}{e})$ to the size constrained maximization problem of an approximately submodular function $F$. (2) By applying $k$-Greedy-TS on a $k$-submodular function $f$ subject to a total size constraint, we obtain a solution providing an approximation ratio of $\frac{1-\varepsilon}{2(1+\varepsilon)}$ to the total size constrained maximization problem of an $\epsilon$-AS function $F$. (3) By applying $k$-Greedy-IS on a $k$-submodular function $f$ subject to individual size constraints, we obtain a solution providing  an approximation 
ratio of $\frac{1-\varepsilon}{3(1+\varepsilon)}$ to the individual size constrained maximization problem of an $\epsilon$-AS function $F$.
For $F$ that is $\varepsilon$-ADR, since it is also $\varepsilon$-AS $F$, these three results apply immediately.

\section{Experiments} \label{section:experiment}

We evaluate the real-world performance of the algorithms for maximizing $F$ directly, 
or indirectly through the use of a bounding $k$-submodular function $f$, by applying them to a 
variant of the $k$-sensor placement problem with IS constraints and a variant of the $k$-topic influence maximization problem with the TS constraint. In both problems, we used $\varepsilon$-AS $k$-submodular functions. 
Appendix~\ref{sec:addexp} contains additional experiments 
with $\varepsilon$-AS and with $\varepsilon$-ADR $k$-submodular functions, whose results are similar to those reported here. 

\vspace{-1mm}
\subsection{Experimental Setup}

In our experiments, we consider an $\varepsilon$-AS function $F$ and  generate the value of $F$ according to the bounding $k$-submodular function $f$, 
which is given explicitly in our application domains.
In particular, for each $\pmb{x} \in(k+1)^{V}$, the value of $F$ of $\pmb{x}$ 
should be generated such that $(1-\varepsilon) f(\pmb{x}) \le F(\pmb{x}) \le (1+\varepsilon) f(\pmb{x})$. 
When $f$ is known, we can compute the value of $f$ for a given $\pmb{x}$ directly and 
then generate the value of $F$ according to some predefined generation methods. 
To highlight the performance or solution quality of $F$ and $f$ under the greedy algorithms, we considered the following three types of generation methods. 
The functions $F$ and $f$ will be defined later for the appropriate domains. 

\vspace{+1mm}
\noindent \textbf{Adversarial Generation (AG).} 
To stress-test the greedy algorithms applied to $f$, we identify the worst-case values of $F$ 
where, when applying these algorithms, $f$'s solution obtains a better approximation ratio than the solution generated when applying the greedy algorithms to $F$. 
We first run a greedy algorithm on $f$ and obtain its solution $\pmb{x_f}$. 
We let $F(\pmb{x_f}) = (1+\varepsilon)f(\pmb{x_f})$ which yields higher weight to $f$'s solution. 
For the remaining $\pmb{x}$, we let $F(\pmb{x}) = \xi(\pmb x)\cdot f(\pmb x)$ where $\xi(\pmb x)$ is selected uniformly at random in $[1-\varepsilon,1]$. 

\vspace{+1mm}
\noindent \textbf{Max and Mean Generation (MaxG and MeanG).} 
Our goal is to consider a more structured error/noise generation setting 
where each (selected) element contributes some uncertainty to the value of $F$. 
In MaxG, $F(\pmb{x}) = \xi(\pmb x)\cdot f(\pmb x)$,  
where $\xi(\pmb x) = \max_{x\in {supp(\pmb x)}}\xi(x)$ and $\xi(x) \in  [1-\varepsilon, 1]$. Thus, we weigh $f$ with the maximum value of noise over the elements of ${\pmb x}$. 
In MeanG, $\xi({\pmb x})=\frac{\sum_{x\in supp(\pmb x)}\xi(x)}{|supp(\pmb x)|}$ and $\xi(x) \in  [1-\varepsilon, 1]$. Thus, we weigh $f$ with the expected value of noise over 
${\pmb x}$.  

Clearly, $F$ is $\varepsilon$-AS. It is also non  $k$-submodular for all  generation methods (see Appendix~\ref{sec:appendixnon}). 

We implemented $k$-Greedy-TS, $k$-Greedy-IS, as well as baselines (details below) in C++ and executed them on an Intel Cascade Lake @ 2.6GhZ with 48GB RAM. These algorithms employed the lazy evaluation technique  \cite{minoux:1978} (i.e., we maintain an upper bound on the gain of inserting each element in each dimension w.r.t. $F$ or $f$, to efficiently select the element in each iteration). We report results of the quality of these algorithms with respect to function $F$ (specifically, the mean and standard deviation of results for $10$ different runs of value generations; each with a different seed). We do not report runtime because the choice of function $F$ or $f$ did not substantially affect the runtime of the  algorithms (for efficiency results of the algorithms see~\cite{Ohsaka:2015aa}). Our source code and the datasets that we used are available at: \url{https://github.com/55199789/approx_kSubmodular.git}.

\vspace{-2mm}
\subsection{Sensor Placement with Approximately $k$-Submodular Functions and the IS Constraints}\label{sec:spexp}

The objective is to install a sufficiently large number of sensors of $k$ types into locations, so that each sensor is installed in a single location and all installed sensors together collect measurements of low uncertainty. We first define the entropy of a vector ${\pmb x}$ of sensors, following~\cite{Ohsaka:2015aa}. Let $\Omega=\{X^u_i\}_{i \in [k], u\in V}$ be the set of random variables for each sensor type $i\in[k]$ and each location $u\in V$. Each $X^u_i$ is the random variable representing the measurement collected from a sensor of type $i$ that is installed at location $u$. Thus, $X_i=\{X^u_i\} \subseteq \Omega$ is the set representing the measurements for all locations at which a sensor of type $i\in[k]$ is installed. The entropy of a vector ${\pmb x}=(X_1, \ldots, X_k)\in (k+1)^V$ is given by the monotone $k$-submodular function $H({\pmb x})=H(\cup_{i\in [k]}X_i)=-\sum_{\mathbf{s}\in \text{dom}~\cup_{i\in [k]}X_i}Pr[\mathbf{s}]\cdot \log Pr[\mathbf{s}]$, where $\text{dom}~\cup_{i\in [k]}X_i$ is the domain of $\cup_{i\in [k]}X_i$~\cite{Ohsaka:2015aa}. 

The work of~\cite{Ohsaka:2015aa} considered the problem of maximizing $H({\pmb x})$ subject to individual size constraints for ${\pmb x}$, using $H({\pmb x})$  to capture the uncertainty of measurements. 
Sensor measurements often have random noise due to hardware issues, environmental effects, and imprecision in measurement~\cite{cleaningsensors}, and weighted entropy functions are used to capture such errors~\cite{weightedentropy}. 
As such, we consider a noisy variant of the problem of~\cite{Ohsaka:2015aa}, where the uncertainty of measurements is captured by a weighted function $\tilde{H}({\pmb x})=\xi(\pmb x)\cdot H({\pmb x})$ and $\xi()$ is generated by our AG, MaxG, or MeanG generation method.  

\begin{figure*}[!ht]
  \centering
    \begin{subfigure}{0.23\textwidth}
        \centering
        \captionsetup{justification=centering}
        \includegraphics[width = \linewidth]{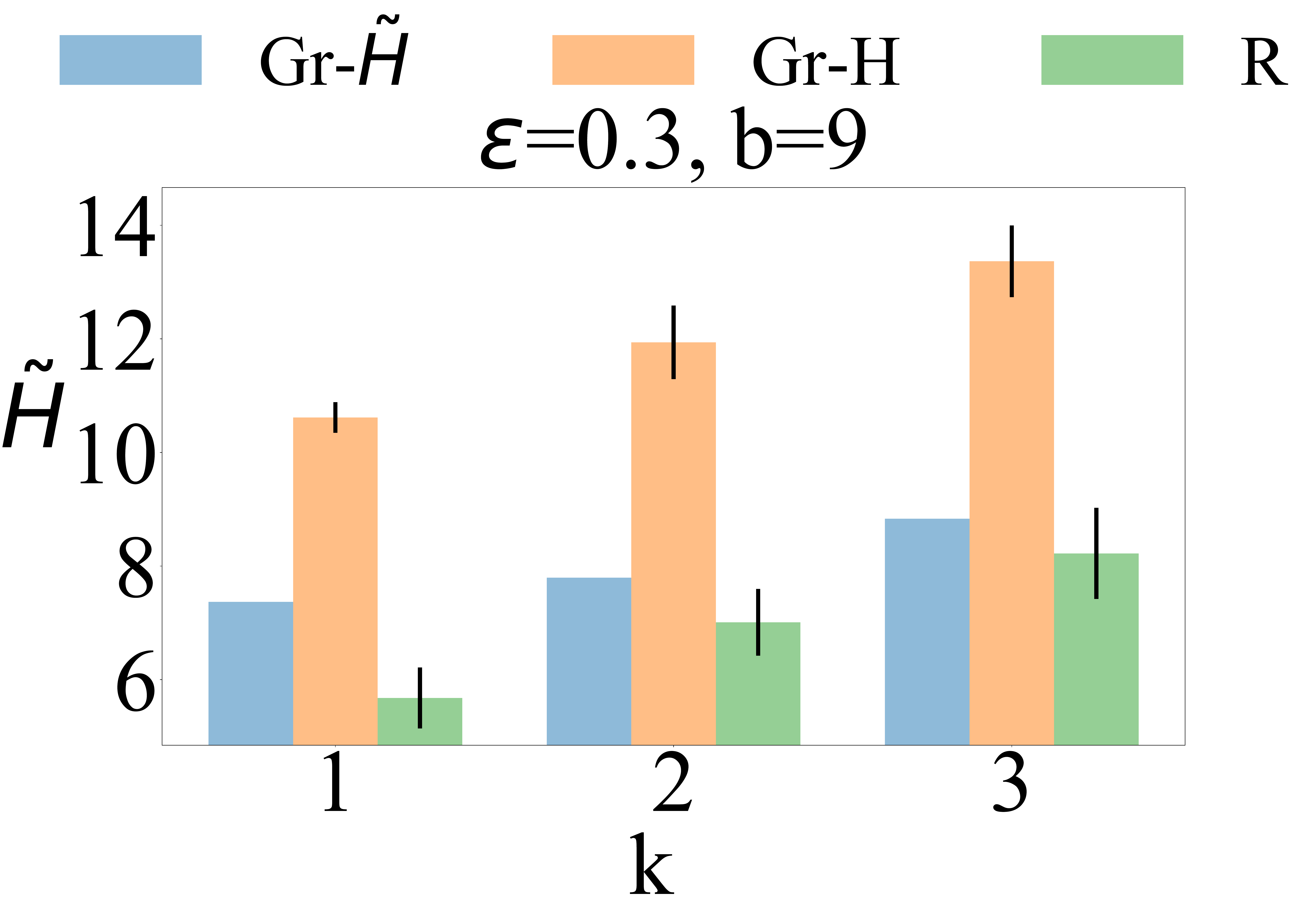}
        \caption{}\label{fig:spAG_IS_k}
    \end{subfigure}%
    \begin{subfigure}{0.23\textwidth}
        \centering
        \captionsetup{justification=centering}
        \includegraphics[width = \linewidth]{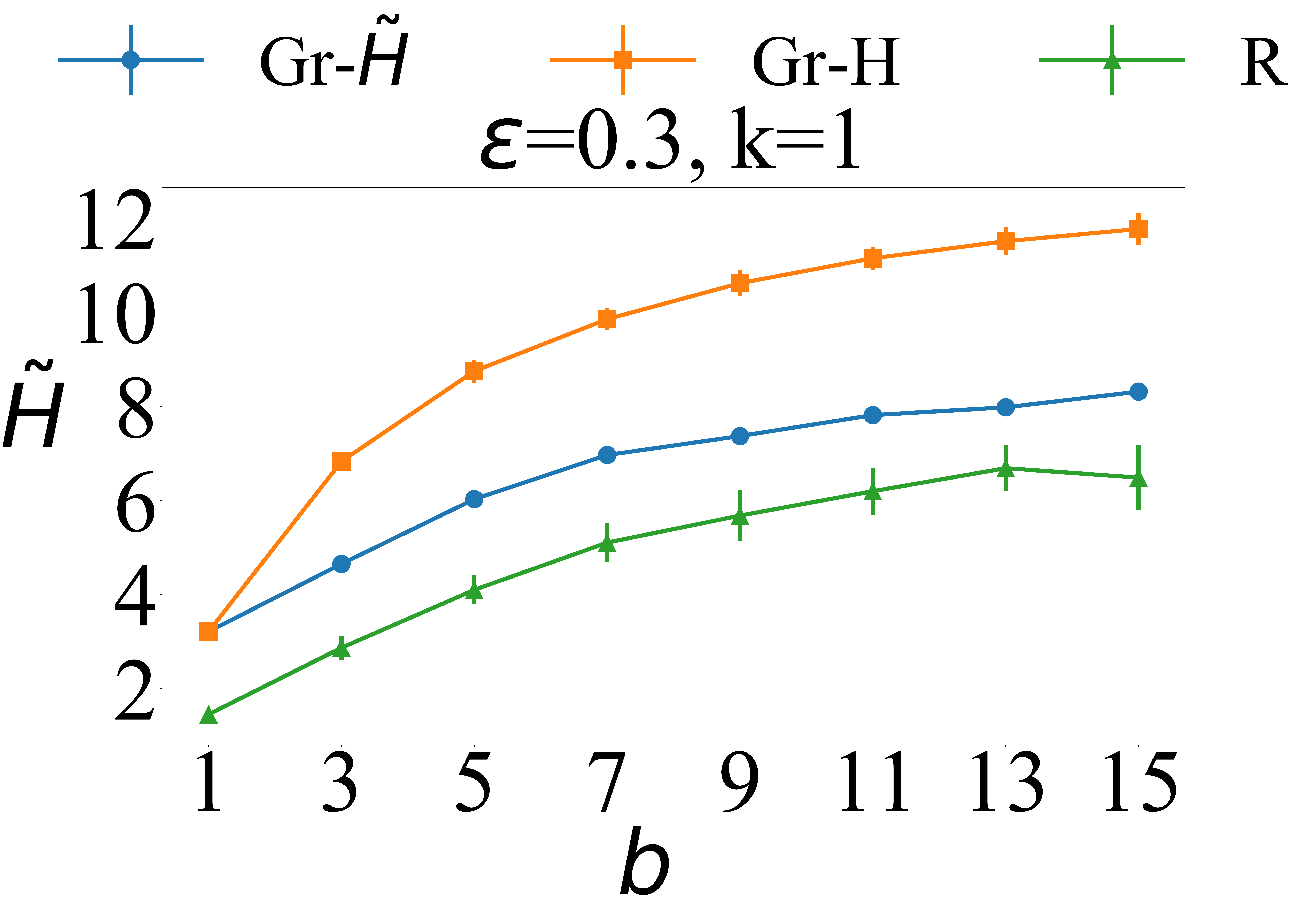}
        \caption{}\label{fig:spAG_IS_B}
    \end{subfigure}%
    \begin{subfigure}{0.23\textwidth}
        \centering
        \captionsetup{justification=centering}
        \includegraphics[width = \linewidth]{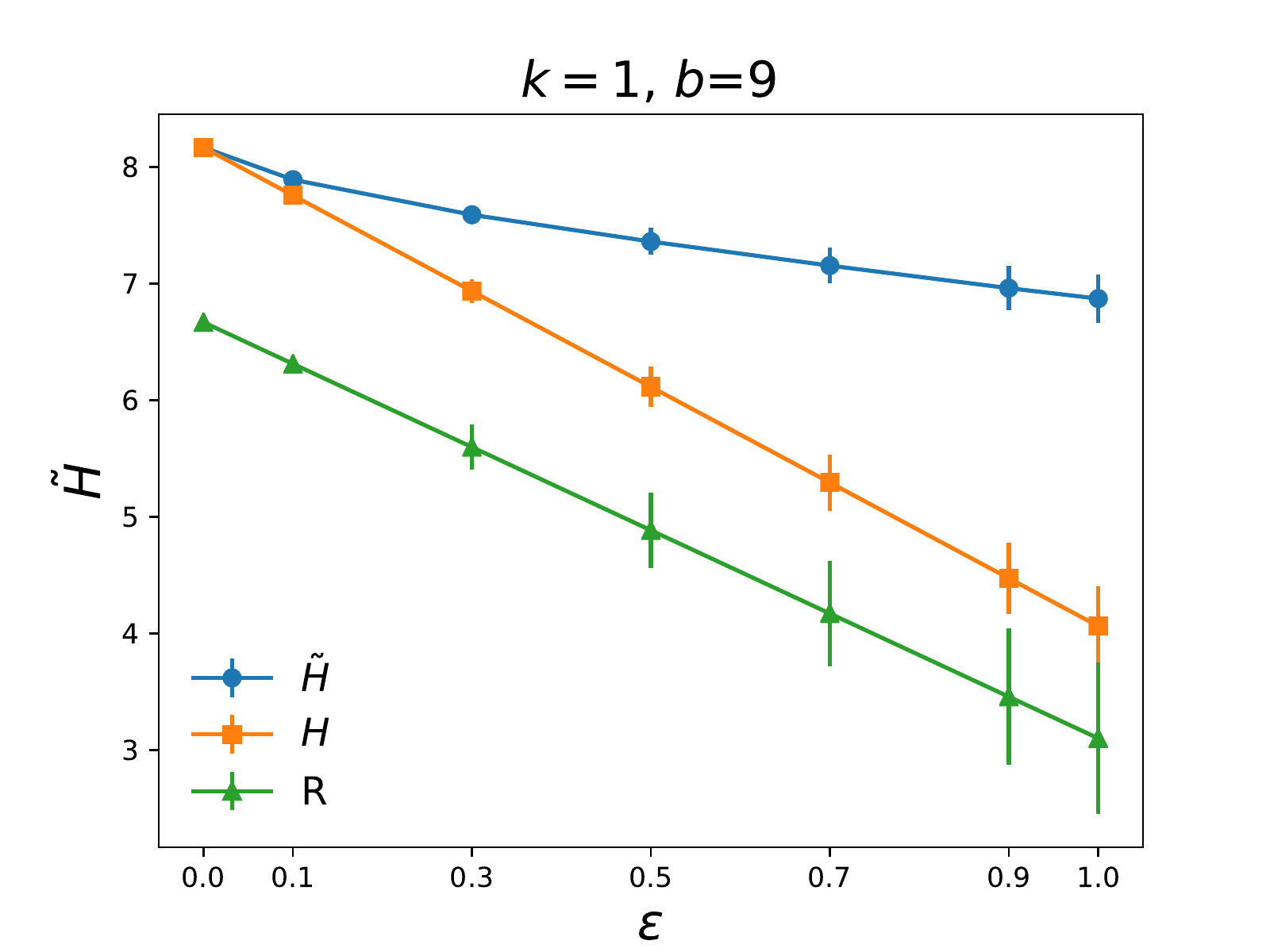}
        \caption{}\label{fig:spAG_IS_eps}
    \end{subfigure}%
    \begin{subfigure}{0.23\textwidth}
        \centering
        \captionsetup{justification=centering}
        \includegraphics[width = \linewidth]{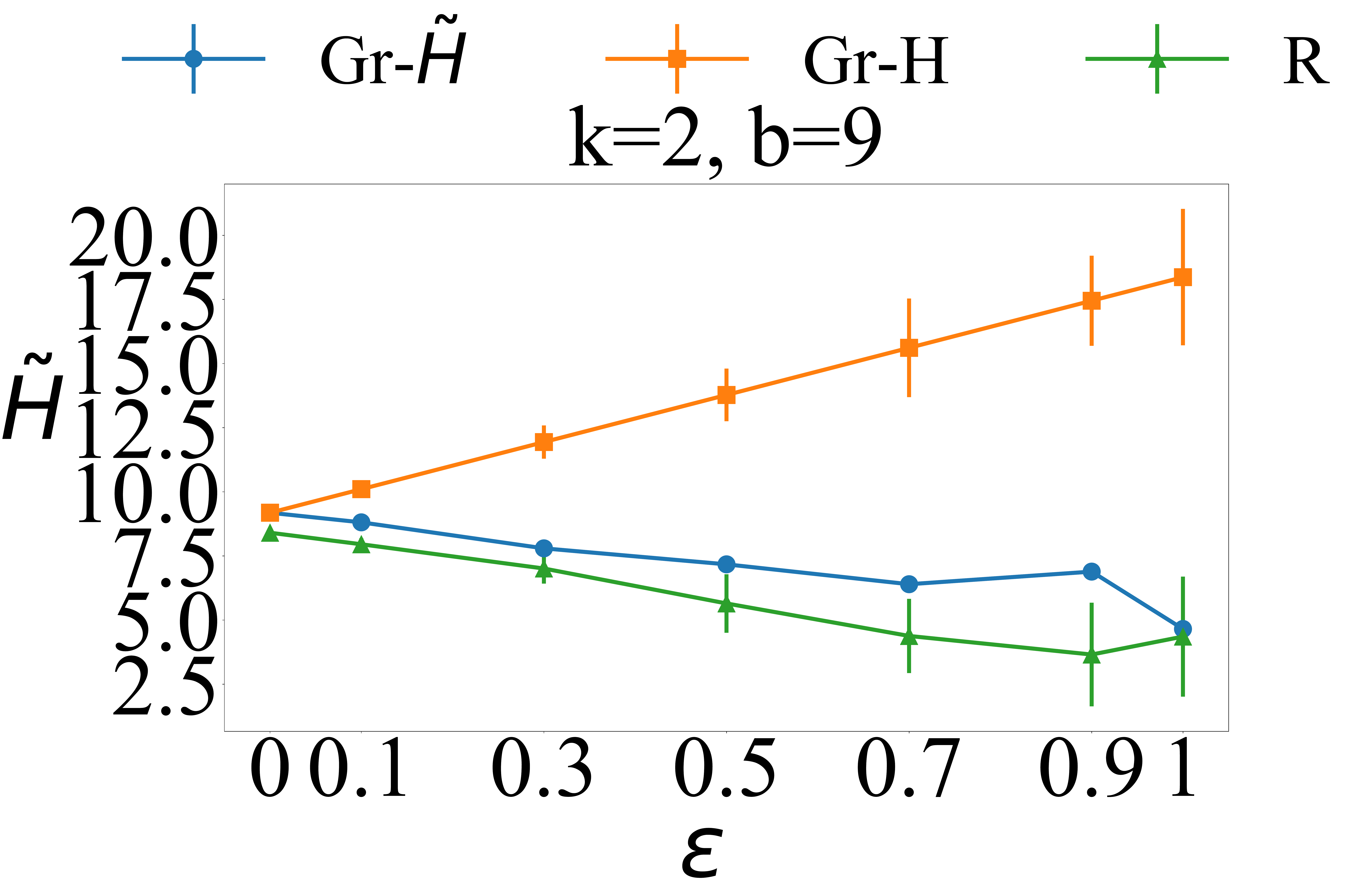}
        \caption{}\label{fig:spAG_IS_eps2}
    \end{subfigure}%
    \vspace{-2mm}
\caption{$\tilde{H}$ in AG setting vs: (a) number of dimensions $k$, (b) individual size threshold $b$, (c,d) $\varepsilon$.}\label{fig:spAG_IS}
\end{figure*}

\begin{figure*}[!ht]
\centering
    \begin{subfigure}{0.23\textwidth}
        \centering
        \captionsetup{justification=centering}
        \includegraphics[width = \linewidth]{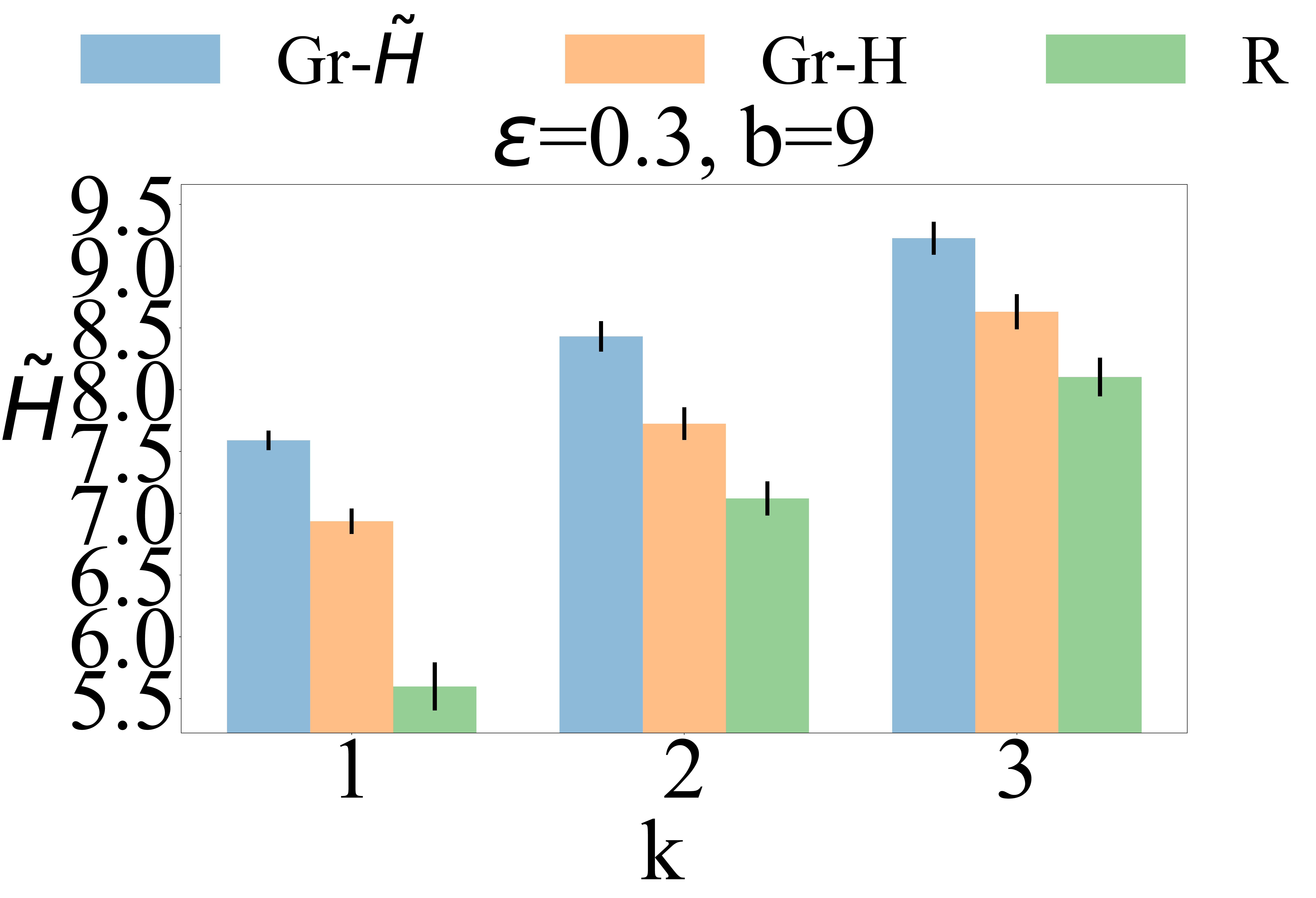}
        \caption{}\label{fig:spMean1_IS_k}
    \end{subfigure}\hspace{+3mm}%
    \begin{subfigure}{0.23\textwidth}
        \centering
        \captionsetup{justification=centering}
        \includegraphics[width = \linewidth]{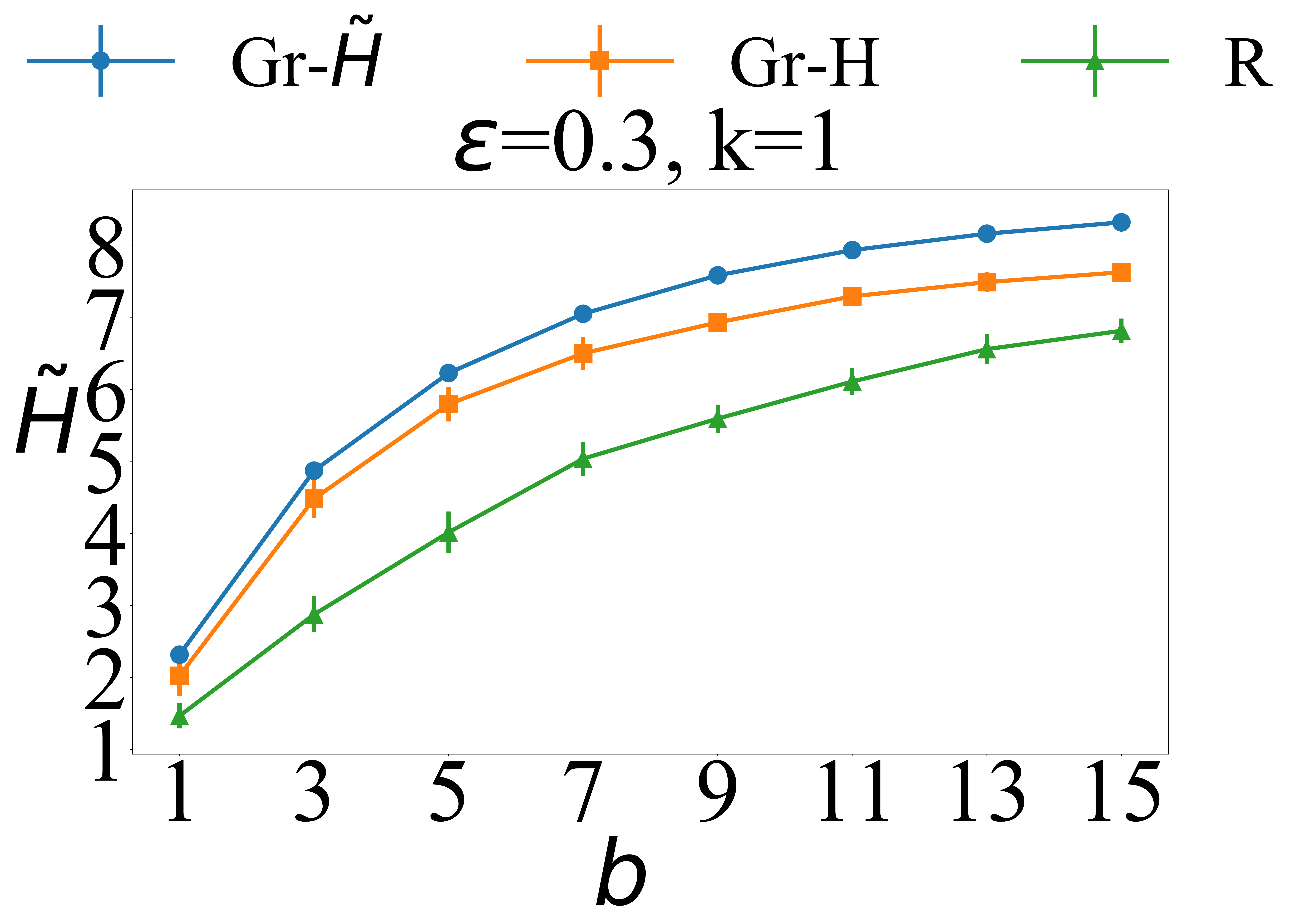}
        \caption{}\label{fig:spMean1_IS_B}
    \end{subfigure}\hspace{+3mm}%
    \begin{subfigure}{0.23\textwidth}
        \centering
        \captionsetup{justification=centering}
        \includegraphics[width = \linewidth]{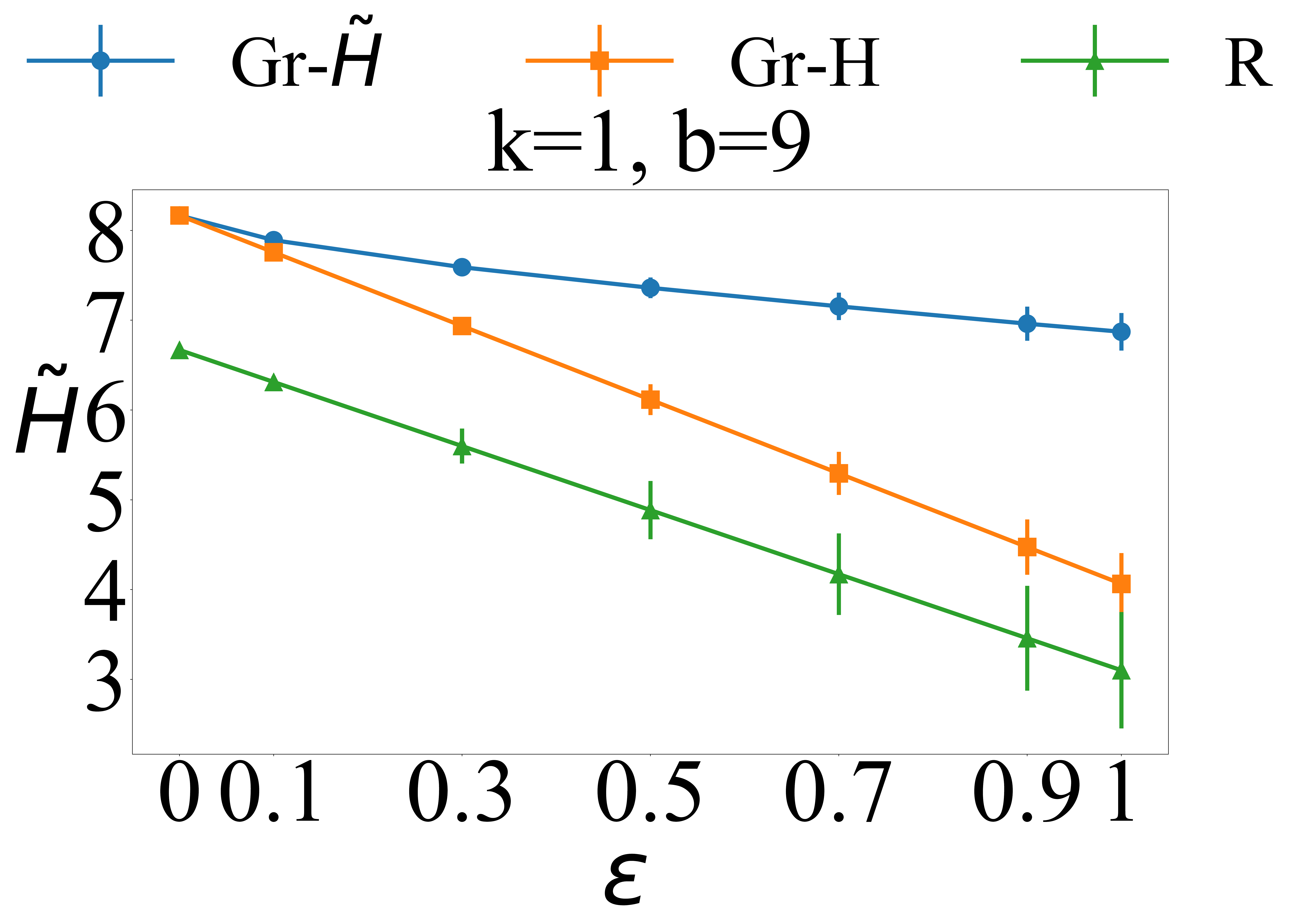}
        \caption{}\label{fig:spMean1_IS_B}
    \end{subfigure}%
    \begin{subfigure}{0.23\textwidth}
        \centering
        \captionsetup{justification=centering}
        \includegraphics[width = \linewidth]{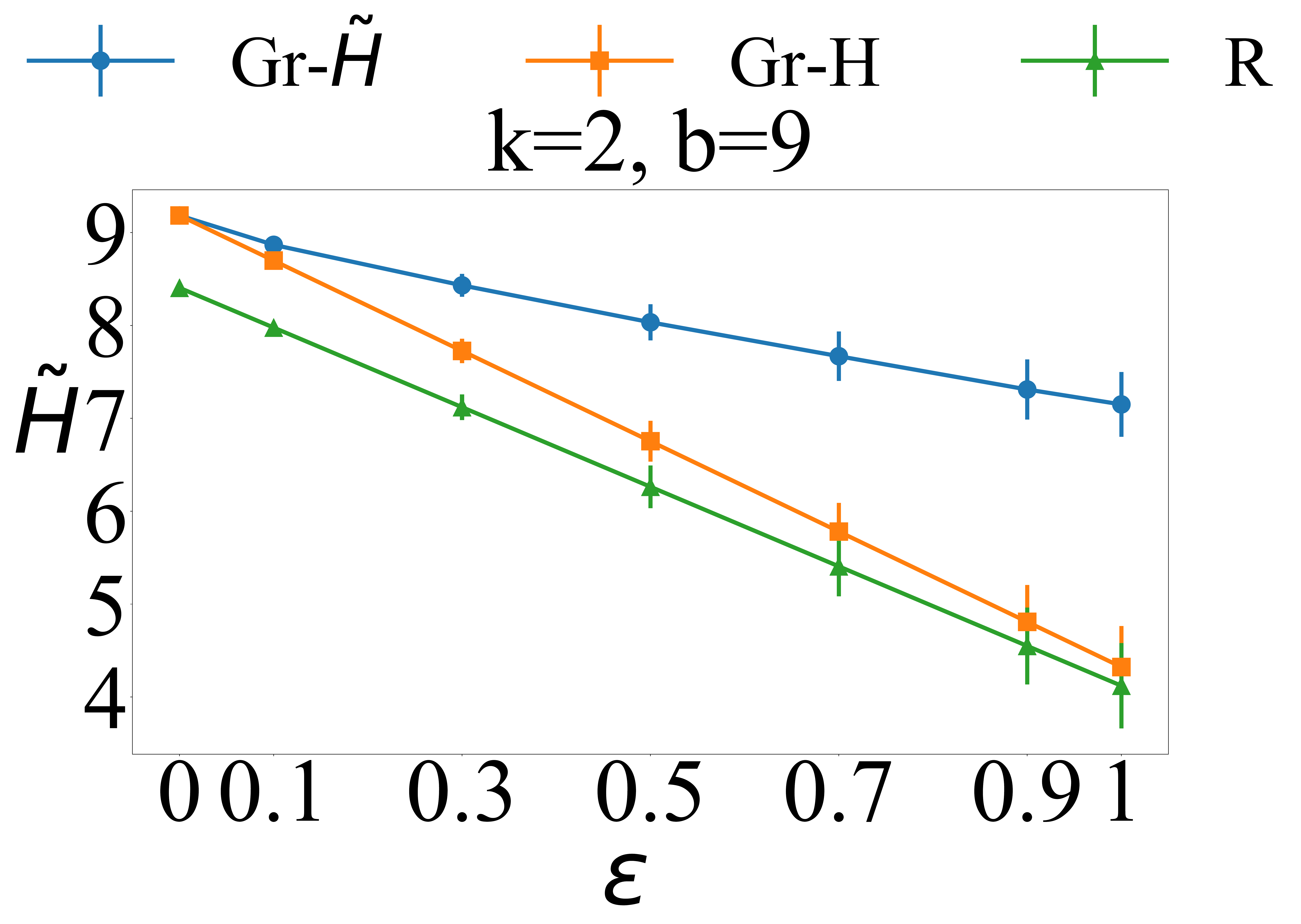}
        \caption{}\label{fig:spMean1_IS_B2}
    \end{subfigure}%
    \vspace{-2mm}
\caption{$\tilde{H}$ in MeanG setting vs: (a) number of dimensions $k$, (b) individual size threshold $b$, (c,d) $\varepsilon$.}\label{fig:spMean}
   \end{figure*}
\begin{figure*}[!ht]
   \vspace{-3mm}
   \centering
    \begin{subfigure}{0.23\textwidth}
        \centering
        \captionsetup{justification=centering}
        \includegraphics[width = \linewidth]{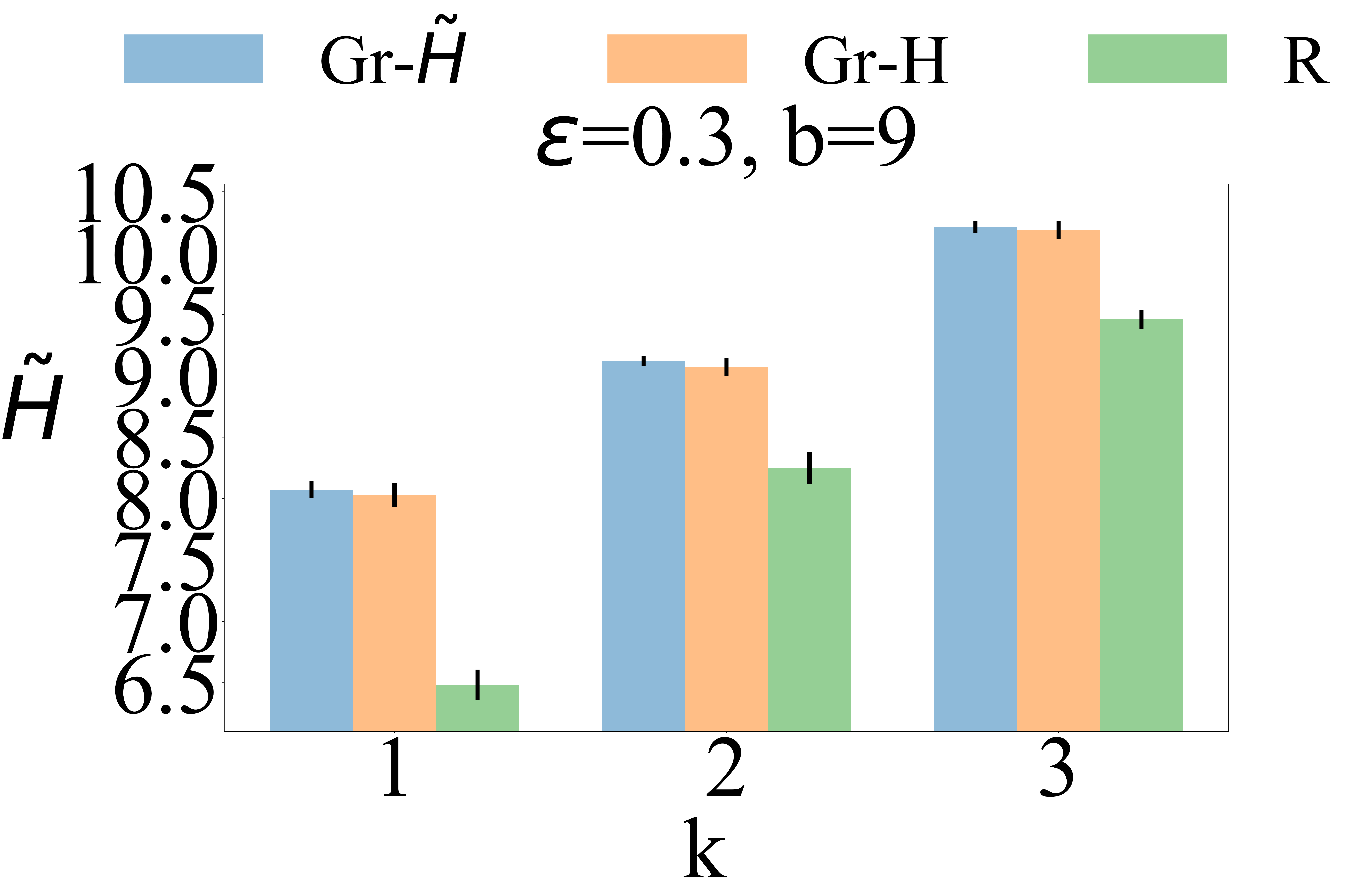}
        \caption{}\label{fig:spMax1_IS_k}
    \end{subfigure}%
    \begin{subfigure}{0.23\textwidth}
        \centering
        \captionsetup{justification=centering}
        \includegraphics[width = \linewidth]{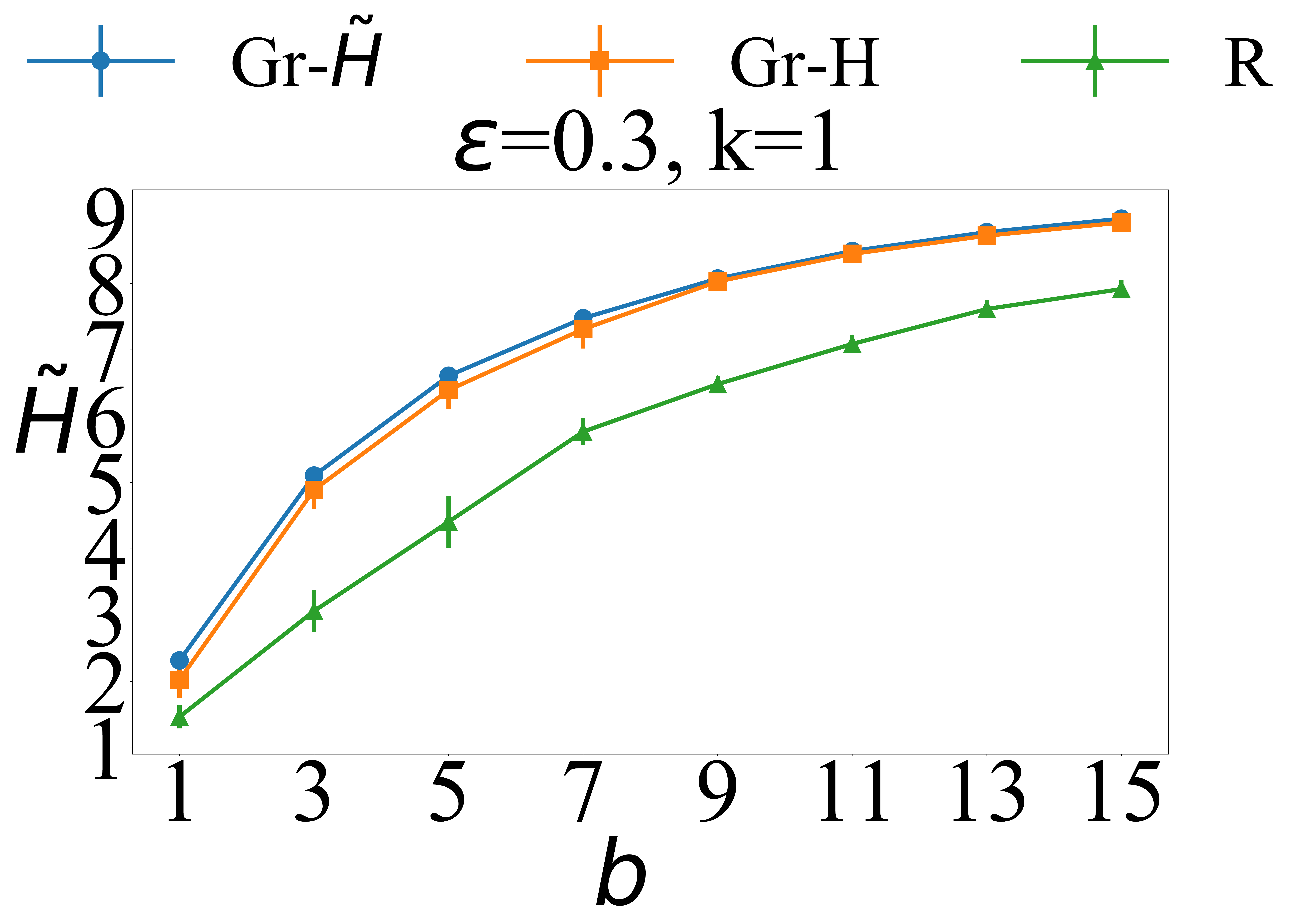}
        \caption{}\label{fig:spMax1_IS_B}
    \end{subfigure}%
    \begin{subfigure}{0.23\textwidth}
        \centering
        \captionsetup{justification=centering}
        \includegraphics[width = \linewidth]{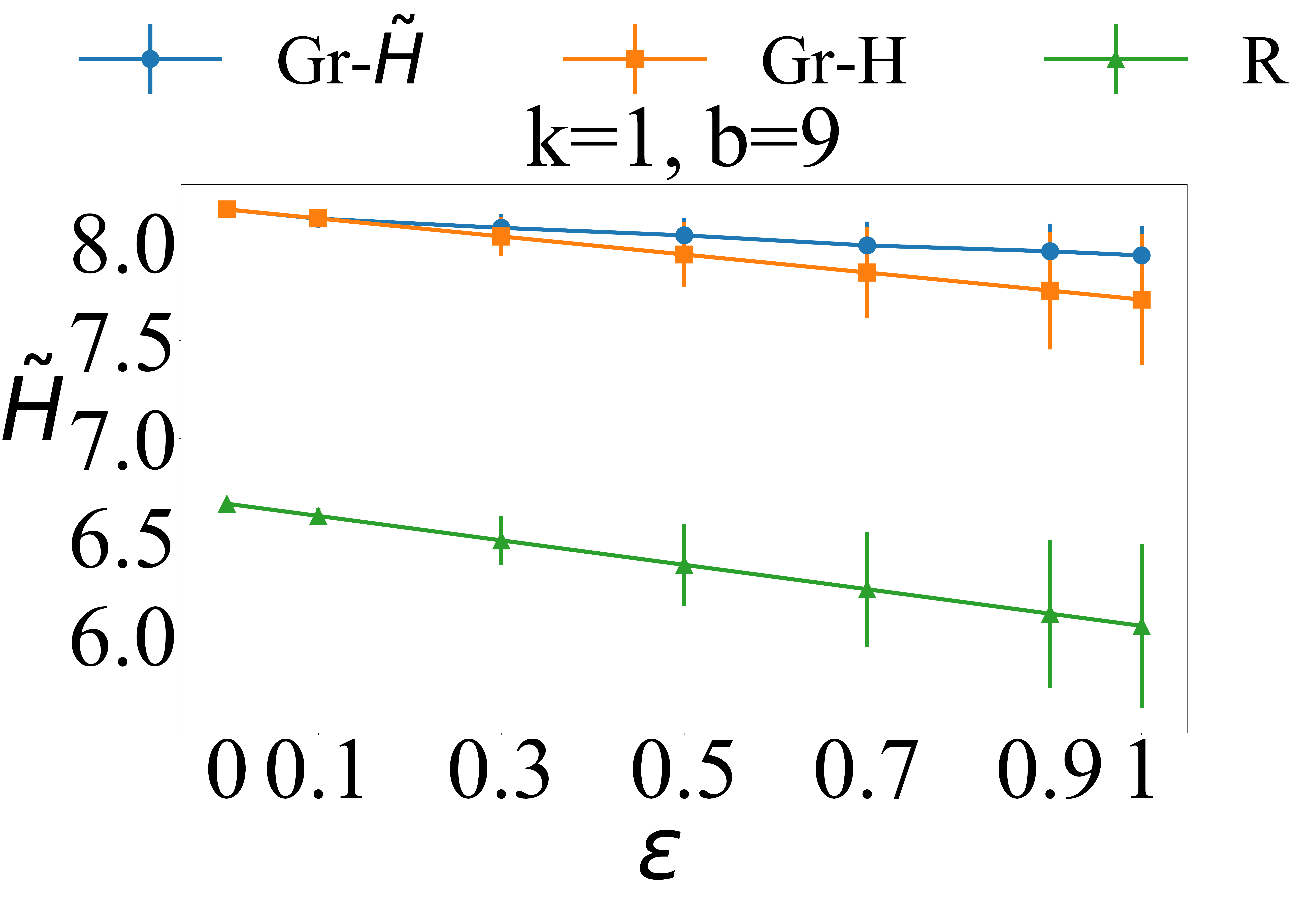}
        \caption{}\label{fig:spMax1_IS_B}
    \end{subfigure}%
    \begin{subfigure}{0.23\textwidth}
        \centering
        \captionsetup{justification=centering}
        \includegraphics[width = \linewidth]{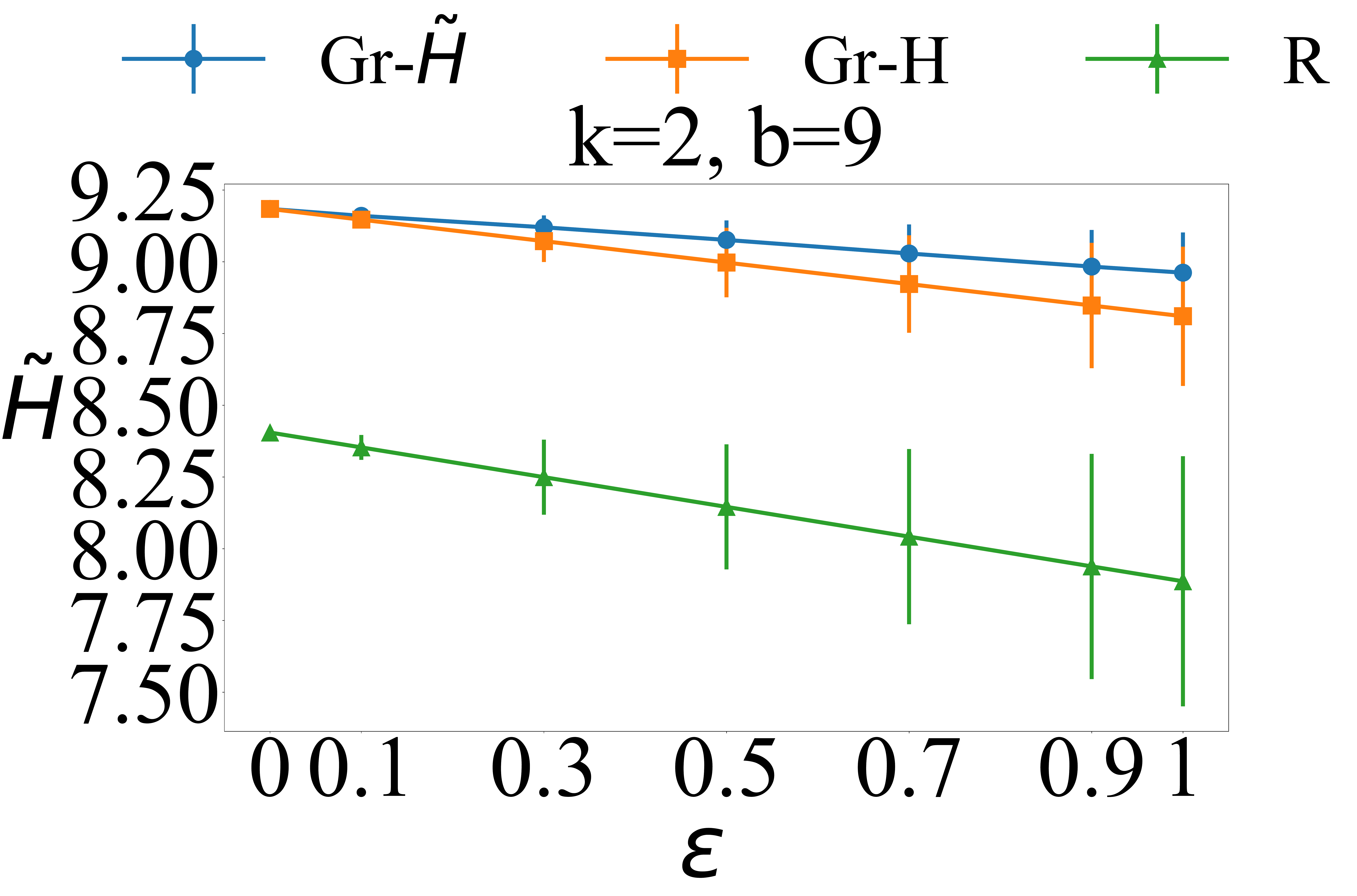}
        \caption{}\label{fig:spMax1_IS_B2}
    \end{subfigure}%
    \vspace{-2mm}
\caption{$\tilde{H}$ in MaxG setting vs: (a) number of dimensions $k$, (b) individual size threshold $b$, (c,d) $\varepsilon$.}\label{fig:spMax}
\vspace{-0.5cm}
  \end{figure*}

\vspace{+1mm}
\noindent\textbf{Algorithms.}~ We first apply $k$-Greedy-IS using $H$ as $f$ and then using $\tilde{H}$ as $F$, and we compare their solutions in terms of $\tilde{H}$. We refer to these algorithms as $Gr$-$H$ and $Gr$-$\tilde{H}$ respectively. 
Since no algorithms can address our problem, we compare against a baseline, Random (referred to as $R$), which outputs as a solution a vector ${\pmb x}$ with $B_i$ randomly selected elements in each dimension. Random was also used in~\cite{Ohsaka:2015aa}. In our experiments, each $B_i$ has the same value $b$. We configured the  algorithms with $k\in\{1,2,3\}$, $\varepsilon \in \{0, 0.1, \ldots, 1\}$, and $B_i\in \{1,3, \ldots, 15\}$. Unless otherwise stated, $k=1$, $\varepsilon=0.3$, and $b=9$. 
Other parameter settings showed similar behaviors.

\vspace{+1mm}
\noindent\textbf{Dataset.}~ We used the Intel Lab dataset which is available at  \url{http://db.csail.mit.edu/labdata/labdata.html} and is  preprocessed as in \cite{Ohsaka:2015aa}. The dataset is a log of approximately 2.3 million values that are collected from 54 sensors installed in 54 locations in the Intel Berkeley research lab. There are three types of sensors. Sensors of type 1, 2, and 3 collect temperature, humidity, and light values, respectively. $F$ or $f$ take as argument a vector: (1) $\pmb{x}=(X_1)$ of sensors of type 1, when $k=1$; (2) $\pmb{x}=(X_1, X_2)$ of sensors of type 1 and of type 2, when $k=2$, or (3) $\pmb{x}=(X_1, X_2, X_3)$ of sensors of type 1 and of type 2 and of type 3, when $k=3$.

\begin{figure*}[!ht]\centering
    \begin{subfigure}{0.23\textwidth}
        \centering
        \captionsetup{justification=centering}
        \includegraphics[width = \linewidth]{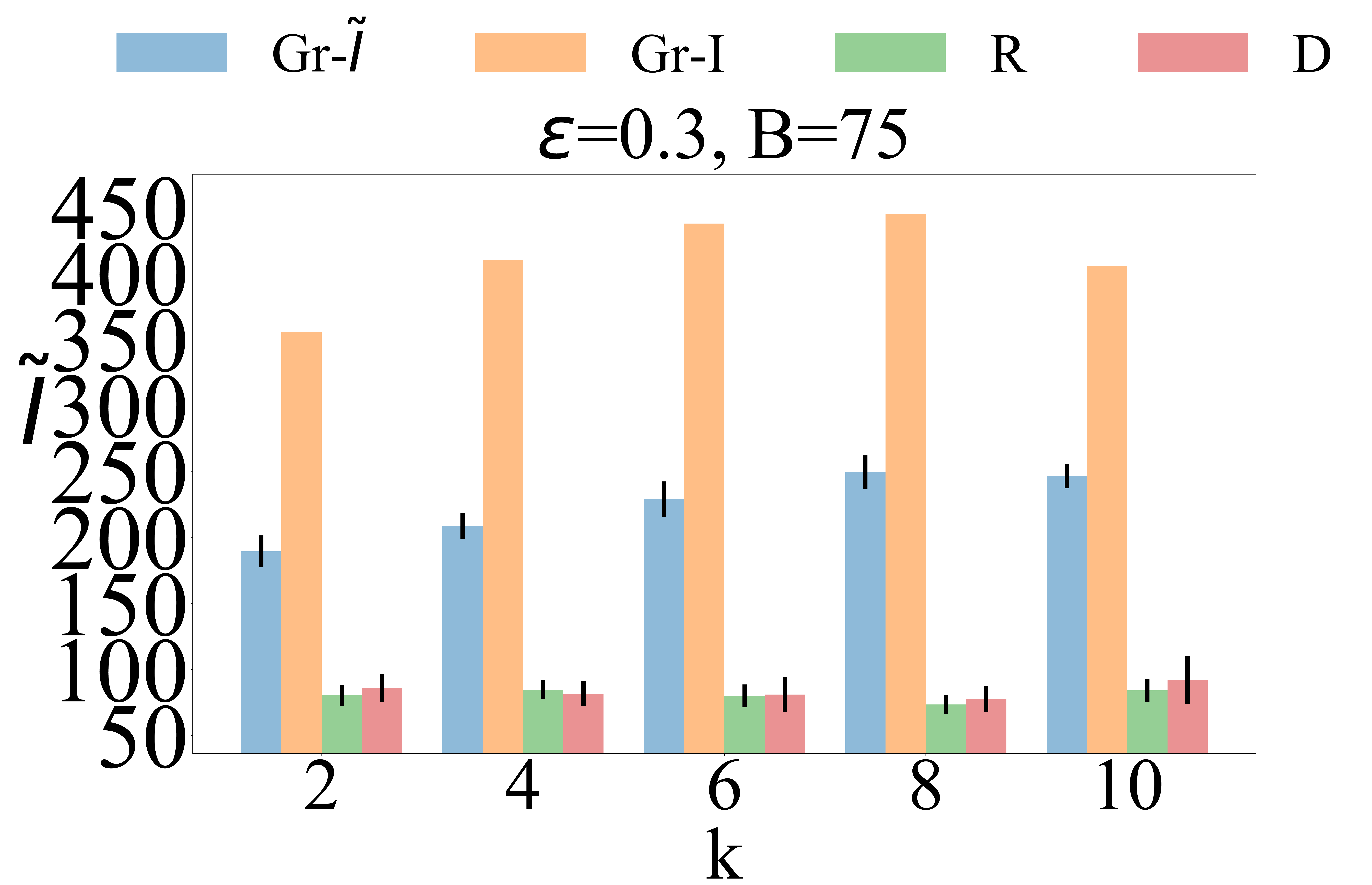}
        \caption{}\label{fig:IMk1_paper}
    \end{subfigure}%
    \begin{subfigure}{0.23\textwidth}
        \centering
        \captionsetup{justification=centering}
        \includegraphics[width = \linewidth]{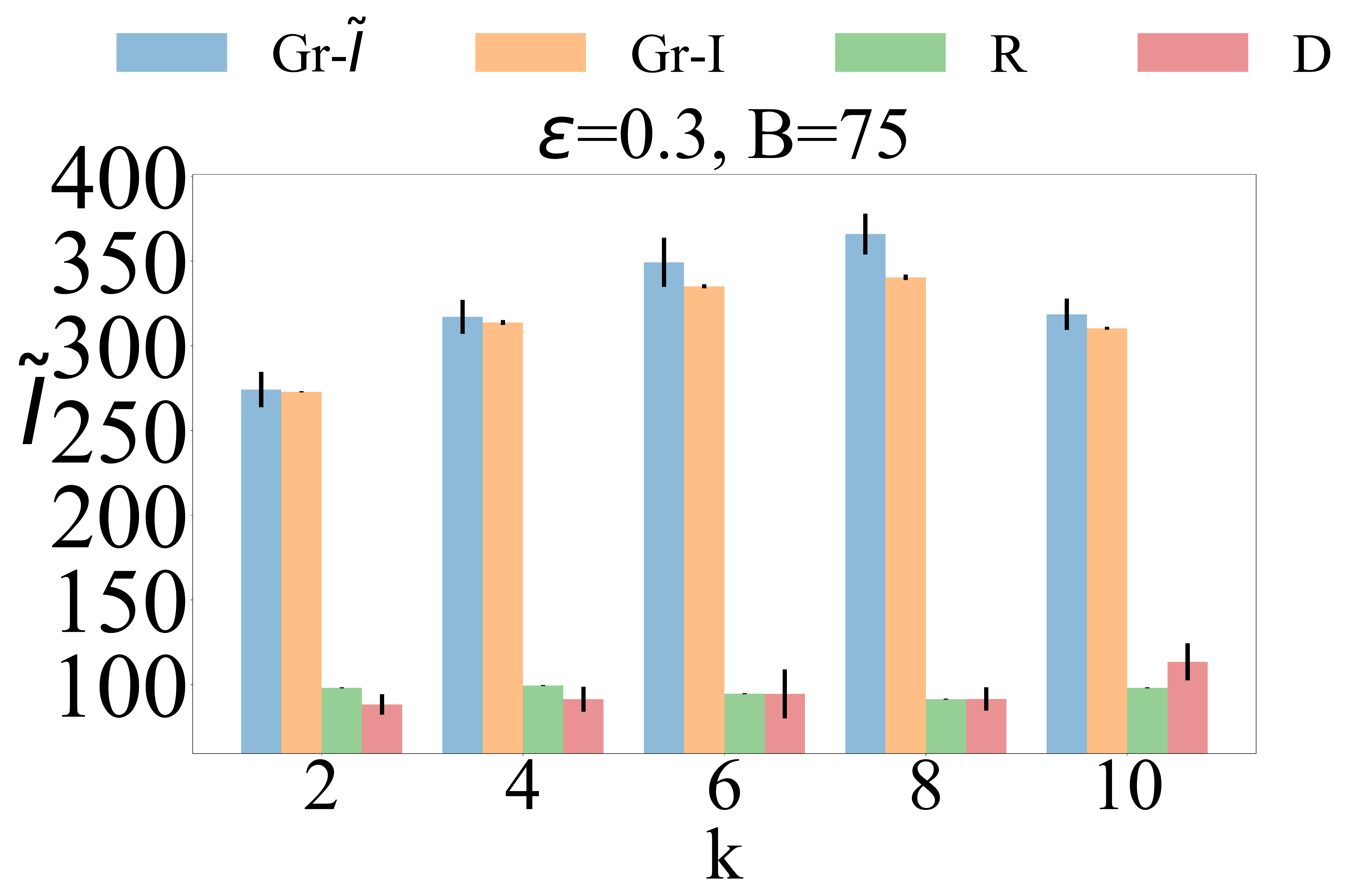}
        \caption{}\label{fig:IMk2_paper}
    \end{subfigure}%
    \begin{subfigure}{0.23\textwidth}
        \centering
        \captionsetup{justification=centering}
        \includegraphics[width = \linewidth]{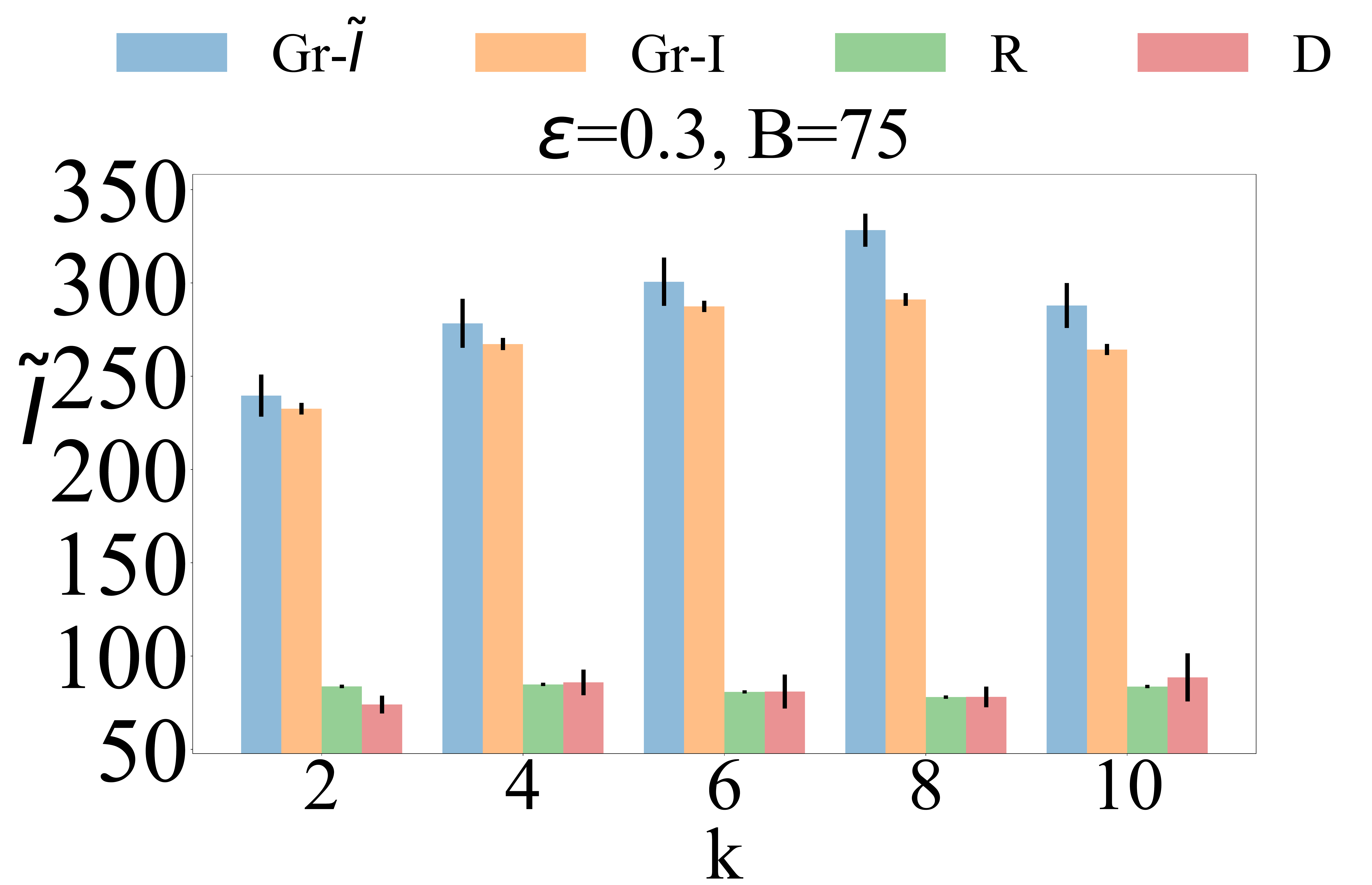}
        \caption{}\label{fig:IMk3_paper}
    \end{subfigure}%
\caption{$\tilde{I}$ for varying $k$ in: (a) AG, (b) MeanG, and (c) MaxG setting.}\label{fig:imkall_paper}
\vspace{-5mm}
   \end{figure*}

\vspace{+1mm}
\noindent\textbf{Results.}~ Fig.~\ref{fig:spAG_IS} shows that, in the AG setting, $Gr$-$H$ substantially outperformed $Gr$-$\tilde{H}$ across all $k$, $b$, and $\varepsilon$ values. 
This can be explained by the fact that, in the AG setting, $Gr$-$H$ is favored by the construction of $F$, because $F$ is based on its solution $\pmb{x_f}$.   
Both $Gr$-$H$ and $Gr$-$\tilde{H}$ outperformed $R$ (the latter by a smaller margin due to the adversarial noise construction). This happened even when $\varepsilon=1$, the case in which they do not offer approximation guarantees. 
The performance of $Gr$-$\tilde{H}$ suffers as $\varepsilon$ goes to 1, due to the uniform range $[1-\varepsilon, 1]$ used to generate $F$'s values. 

Figs. \ref{fig:spMean} and \ref{fig:spMax} show that, in the MaxG and MeanG settings, $Gr$-$\tilde{H}$ outperformed $Gr$-$H$ in almost all tested cases on average, and especially for larger  $b$ and $\varepsilon$. This is because 
the noise is more structured and suggests that $Gr$-$\tilde{H}$ may be a practical algorithm 
(e.g., in applications where the maximum or expected noise of sensors is taken as an aggregate of  the noise of sensors). 
We also observe that $Gr$-$\tilde{H}$ has a larger  performance gain and less variability (low standard deviation bars) over $Gr$-$H$ under MeanG than under MaxG. 

\vspace{-2mm}
\subsection{Influence Maximization with Approximately $k$-Submodular Functions and the TS Constraint} \label{exp:IM}

The objective is to select a sufficiently large number of users in a social network who would influence the largest expected number of users in the social network through word-of-mouth effects. The selected users are called \emph{seeds}. To measure influence, we adapt the $k$-IC influence diffusion model proposed in~\cite{Ohsaka:2015aa}. In the $k$-IC model, $k$ different topics spread through a social network independently. 
At 
$t=0$, there is a vector  $\pmb{x}=(X_1, \ldots, X_k)$ of seeds who are influenced.  Each $u$ in $X_i$, $i\in[k]$, is influenced about topic $i$ and has a single chance to influence its out-neighbor $v$, if $v$ is not already influenced. The node $v$ is influenced at $t=1$ by $u$ on topic $i$ with probability $p_{u,v}^i$. Also, $v$ is influenced  by any of its in-neighbors (other seeds) at time $t=1$ on topic $i$. When $v$ becomes influenced, it stays influenced and has a single chance to influence each of its out-neighbors that is not already influenced. The process proceeds until no new nodes are influenced. The expected number of influenced users (\emph{spread}) is $I(\pmb{x})=\mathrm{E}[|\cup_{i\in[k]}A_i(X_i)|]$, where $A_i(X_i)$ is a random variable representing the set of users influenced about topic $i$ through $X_i$. 

Our adapted $k$-IC model differs from the $k$-IC model in that we measure spread by $\tilde{I}(\pmb{x})=\xi(\pmb x) \cdot I({\pmb x})$, instead of $I(\pmb{x})$,  where $\xi()$ is the noise function in the AG, MaxG, or MeanG setting. The noise models empirical evidence that the spread may be non-submodular and difficult to quantify accurately. This happens because users find information diffused by many in-neighbors as already known and less interesting, in which case the noise depends on the data~\cite{noveltydecay}. 
Furthermore, it happens because the combined influence of subsets of influenced in-neighbors of a node also affects the influence probability of the node and hence the spread, in which case the  noise depends on the subset of influenced in-neighbor of a  user~\cite{hypergraph}. $\tilde{I}({\pmb x})$ is an $\varepsilon$-AS function, since $(1-\varepsilon)\cdot I({\pmb x}) \leq \tilde{I}({\pmb x})\leq  (1+\varepsilon)\cdot I({\pmb x})$ and $I(\pmb x)$ is monotone $k$-submodular~\cite{Ohsaka:2015aa}.

\vspace{+1mm}
\noindent \textbf{Algorithms.}~ We first apply $k$-Greedy-TS using $I$ as $f$ and then using $\tilde{I}$ as $F$, and we compare their solutions in terms of $\tilde{I}$. We refer to them as $Gr$-$I$ and $Gr$-$\tilde{I}$. 
We also evaluate $Gr$-$I$ and $Gr$-$\tilde{I}$ against two baselines, also used in~\cite{Ohsaka:2015aa}: (1) Random (R), which outputs a random vector ${\pmb x}$ of size $B$, and (2) Degree (D), which sorts all nodes in decreasing order based on their out-degree and then assigns each of them to a random topic (dimension). We simulated the influence process based on Monte Carlo simulation as in \cite{Ohsaka:2015aa}. We configured the  algorithms with $k\in\{2, 4, \ldots, 10\}$, $\varepsilon \in \{0, 0.1, \ldots, 1\}$, and $B\in \{5,10 , \ldots,  100\}$. By default, $k=8$, $\varepsilon=0.3$, and $B=75$.

\vspace{+1mm}
\noindent \textbf{Dataset.}~ We used the Digg social news dataset that is available at  \url{http://www.isi.edu/~lerman/downloads/digg2009.html},  following the setup of \cite{Ohsaka:2015aa}.
The dataset consists of a graph and a log of user votes for stories. Each node represents a user and each edge $(u,v)$ represents that user $u$ can watch the activity of node $v$.
The edge probabilities $p_{u,v}^i$ for each edge $(u,v)$ and topic $i$ were obtained from~\cite{Ohsaka:2015aa}.

\vspace{+1mm}
\noindent \textbf{Results.}~The results in  Fig.~\ref{fig:imkall_paper} are similar to those 
of Section~\ref{sec:spexp}. That is, in the AG setting  $Gr$-$I$ outperformed $Gr$-$\tilde{I}$ (see  Fig.~\ref{fig:IMk1_paper}). 
This is because $F$ is based on the solution of $Gr$-$I$ and thus this algorithm is favored over $Gr$-$\tilde{I}$. On the other hand,  in the MaxG and MeanG setting  $Gr$-$\tilde{I}$ outperformed $Gr$-$I$ in all tested cases (see Figs. \ref{fig:IMk2_paper} and \ref{fig:IMk3_paper}). This is because the noise is more structured and suggests that $Gr$-$\tilde{I}$ may be a practical algorithm, when the noise is structured and not adversarially chosen. In all tested cases, as expected, both $Gr$-$I$ and $Gr$-$\tilde{I}$ outperformed $R$ and $D$. We observed similar trends, when we varied the parameters $B$ and $\varepsilon$ (see  Appendix~\ref{appendix:E1}).

\section{Conclusion and Discussion} \label{section:discussion}
In this paper, we show that simple greedy algorithms  
can obtain reasonable approximation ratios for an $\varepsilon$-AS or $\varepsilon$-ADR function $F$ 
subject to total size and individual size constraints. 
The analysis (i.e., proofs of Theorem \ref{theorem:ASIS} and Theorem \ref{theorem:ADRIS}) 
for the individual size constraint can be extended to capture  a \emph{group} size constraint. 
Let $G_1, ..., G_m$ be a partition of $\{1, ..., k\}$ and $B_1, ..., B_m$ be some positive integer numbers. 
The maximization problem of an $\varepsilon$-AS or $\varepsilon$-ADR function $F$ subject to the group size constraint is defined to be 
$\max_{\pmb x \in (k+1)^V : \sum_{j \in G_i} |supp_j(\pmb{x})|  \le B_i \; \forall i \in [m]} F(\pmb x)$, 
where the total size of all of subsets within a group $G_i$ is at most $B_i$.  
The same approximation ratios from Theorem \ref{theorem:ASIS} and Theorem \ref{theorem:ADRIS} can be obtained for $\varepsilon$-AS and $\varepsilon$-ADR function $F$, respectively, 
using a modified greedy algorithm (similar to Algorithm \ref{algorithm2}) where the condition in line 5 can be changed to account for the group constraint. 


Our definitions for $\varepsilon$-AS and $\varepsilon$-ADR depend on the lower bound constant, $(1-\varepsilon)$ and upper bound constant, $(1+\varepsilon)$. 
Similar approximation ratio results can be derived when replacing $(1-\varepsilon)$ with some lower bound constant $a$ and upper bound constant $b$ where $0 < a \le b$. The approximation ratios will depend on $a$ and $b$ and can be obtained following the same proof ideas. 

Finally, it would be interesting to evaluate our algorithms using datasets that are inherently noisy. 

\bibliographystyle{plain}
\bibliography{short_references} 

\newpage

\appendix
\appendixpage

\section{Proofs in Section \ref{section:prelim}}\label{appendA}

\begin{proof}[Proof of Theorem \ref{lemma1}]
To prove the theorem, it is sufficient to show that the claim holds for any $\pmb{x} \in (k+1)^V$. 
For any $\pmb{x} = (X_1, ..., X_k) \in (k+1)^V$, for each $i \in [k]$, we order the elements of $X_ i$ such that 
$X_i = \{e_{i1}, e_{i2}, ..., e_{i |X_i|} \}$. 
It follows that 
\begin{align*}
F(\pmb{x}) &= F(X_1, \hdots, X_k) \\ 
&= \sum_{i \in [k]} \sum_{j = 1}^{|X_i|}  \Delta_{e_j,i} F(X_1, ..., X_i \setminus \{e_{i1}, ..., e_{ij}\}, ..., X_k), 
\end{align*}
where the equality is from adding/subtracting common terms. 
Since $F$ is $\varepsilon$-ADR, we have 
{\small
\begin{align*}
&(1-\varepsilon) f(\pmb{x}) \\
&= (1- \varepsilon) \cdot  \sum_{i \in [k]} \sum_{j = 1}^{|X_i|}  \Delta_{e_j,i} f(X_1, ..., X_i \setminus \{e_{i1}, ..., e_{ij}\}, ..., X_k) \\ 
& \le  F(\pmb{x}) \\
&\le (1+\varepsilon)\cdot \sum_{i \in [k]} \sum_{j = 1}^{|X_i|}  \Delta_{e_j,i} f(X_1, ..., X_i \setminus \{e_{i1}, ..., e_{ij}\}, ..., X_k) \\
&\leq (1+\varepsilon) f(\pmb{x}), 
\end{align*}}
\noindent where the inequality is due to the application of $\varepsilon$-ADR definition to each summation term. 
Thus, we have shown that $F$ is $\varepsilon$-AS under the same $k$-submodular function $f$. 
\end{proof}

\section{Proofs in Section \ref{section:kg1}} \label{sec:algandproofs}


\begin{proof}[Proof of Lemma \ref{lemma:ineq1}]
First note that $\Delta_{e^{(j)},i^{(j)}}F(\pmb{x}^{(j-1)})\geq{}\Delta_{o^{(j)},\pmb{o}^{(j-1)}(o^{(j)})}F(\pmb{x}^{(j-1)})$ as $(e^{(j)}, i^{(j)})$ 
is selected by the greedy algorithm, which must provide as much marginal gain as $(o^{(j)},\pmb{o}^{(j-1)}(o^{(j)}))$. 
By using the definition of $\varepsilon$-AS on the above inequality, we have that
$(1+\varepsilon)f(\pmb{x}^{(j)}) \geq 
(1-\varepsilon)f\left(\left(X^{(j-1)}_1,\cdots, X^{(j-1)}_{\pmb{o}^{(j-1)}(o^{(j)})}\cup\{o^{(j)}\},\cdots,X^{(j-1)}_k\right)\right)$. 
By submodularity and $\pmb{x}^{(j-1)}\preceq\pmb{o}^{(j-1)}$, 
we have that $\Delta_{o^{(j)},\pmb{o}^{(j-1)}(o^{(j)})}f(\pmb{x}^{(j-1)})
\geq{}\Delta_{o^{(j)},\pmb{o}^{(j-1)}(o^{(j)})}f(\pmb{o}^{(j-\frac{1}{2})}) = f(\pmb{o}^{(j-1)}) - f(\pmb{o}^{(j-\frac{1}{2})})
\ge f(\pmb{o}^{(j-1)})-f(\pmb{o}^{(j)})$. 
Our result follows immediately after combining the above inequalities. 
\end{proof}

\begin{proof}[Proof of Lemma \ref{lemma:ineq2}]
First note that $\Delta_{e^{(j)},i^{(j)}}F(\pmb{x}^{(j-1)})\geq{}\Delta_{o^{(j)},\pmb{o}^{(j-1)}(o^{(j)})}F(\pmb{x}^{(j-1)})$ as $(e^{(j)}, i^{(j)})$ 
is selected by the greedy algorithm, which must provide as much marginal gain as $(o^{(j)},\pmb{o}^{(j-1)}(o^{(j)}))$. 
By submodularity and $\pmb{x}^{(j-1)}\preceq\pmb{o}^{(j-1)}$, 
we have 
$\Delta_{o^{(j)},\pmb{o}^{(j-1)}(o^{(j)})}f(\pmb{x}^{(j-1)}) \geq \Delta_{o^{(j)},\pmb{o}^{(j-1)}(o^{(j)})}f(\pmb{o}^{(j-\frac{1}{2})})
\ge \frac{1}{1+\varepsilon} \Delta_{o^{(j)},\pmb{o}^{(j-1)}(o^{(j)})}F(\pmb{o}^{(j-\frac{1}{2})})$ (the last inequality is by the $\varepsilon$-ADR definition)
and  
$\Delta_{o^{(j)},\pmb{o}^{(j-1)}(o^{(j)})} F(\pmb{o}^{(j-\frac{1}{2})}) - \Delta_{e^{(j)},i^{(j)}} F(\pmb{o}^{(j-\frac{1}{2})}) = F(\pmb{o}^{(j-1)}) - F(\pmb{o}^{(j)})$. 
It follows that $\frac{1}{1-\varepsilon}\Delta_{o^{(j)},\pmb{o}^{(j-1)}(o^{(j)})}F(\pmb{x}^{(j-1)})\geq{}
\frac{1}{1+\varepsilon}\left[F(\pmb{o}^{(j-1)})-F(\pmb{o}^{(j)})\right]$ by the definition of $\varepsilon$-ADR. 
Our claim follows immediately from the last inequalities. 
\end{proof}

\subsection{Maximizing $\varepsilon$-AS and $\varepsilon$-ADR Functions with the Individual Size Constraints} \label{sec:ISconstraint}

In this subsection, we consider the problem of maximizing $\varepsilon$-AS or $\varepsilon$-ADR function $F$ 
subject to the individual size constraint $B_1, ..., B_k$ using Algorithm \ref{algorithm2} on the function $F$.  
Recall that in the individual size constraint maximization problem, we are given $B_1, ..., B_k$ 
restricting the maximum number of elements one can select for each subset.  
We define $B = \sum_{j \in [k] } B_j$. 
We use the following notations as in \cite{Ohsaka:2015aa} 
and the same notations $e^{(j)}$, $i^{(j)}$, and $\pmb{x}^{(j)}$ from  Section \ref{section1}.

        

As before, we iteratively define $\pmb{o}^{(0)}=\pmb{o}, \pmb{o}^{(1)}, \cdots, \pmb{o}^{(B)}$ as follows. 
For each $j\in[B]$, we let $S^{(j)}_i=supp_i(\pmb{o}^{(j-1)})\setminus{}supp_i(\pmb{x}^{(j-1)})$. We consider the following two cases.

\begin{itemize}
\item[\bf C1:] Suppose there exists $i'\neq i^{(j)}$ such that $e^{(j)}\in{}S_{i'}^{(j)}$. In this case, we set $o^{(j)}$ to be an arbitrary element in $S_{i^{(j)}}^{(j)}$.
We construct $\pmb{o}^{(j-\frac{1}{2})}$ from $\pmb{o}^{(j-1)}$ by assigning $0$ to the $e^{(j)}$-th element and the $o^{(j)}$-th element.
Then we construct $\pmb{o}^{(j)}$ from $\pmb{o}^{(j-\frac{1}{2})}$ by assigning $i^{(j)}$ to the $e^{(j)}$-th element and $i'$ to the $o^{(j)}$-th element. 
We may use $\pmb{o}^{(j-\frac{1}{2})}_{(e,i)}$ to denote $\pmb{o}^{(j-\frac{1}{2})}$ with $i$ assigned to the $e$-th element. 

\item[\bf C2:]  Suppose, for any $i'\neq i^{(j)}$, we have $e^{(j)}\notin{}S_{i'}^{(j)}$. 
In this case, we let $o^{(j)}=e^{(j)}$ if $e^{(j)}\in S_{i^{(j)}}^{(j)}$,
and let $o^{(j)}$ be an arbitrary element in $S_{i^{(j)}}^{(j)}$ otherwise.
We construct $\pmb{o}^{(j-\frac{1}{2})}$ from $\pmb{o}^{(j-1)}$ by assigning $0$ to the $o^{(j)}$-th element. 
We then construct $\pmb{o}^{(j)}$ from $\pmb{o}^{(j-\frac{1}{2})}$ by assigning $i^{(j)}$ to the $e^{(j)}$-th element.
\end{itemize}
By construction, we have $|supp_i(\pmb{o}^{(j)}|=B_i$ for each $i\in[k]$ and $j\in\{0, 1, \cdots, B\}$. 
Moreover, $\pmb{x}^{(j-1)}\preceq\pmb{o}^{(j-\frac{1}{2})}$ for each $j\in[B]$.

    
We first consider $\varepsilon$-AS Functions with the Individual Size Constraints in Section~\ref{b11}. Then, we consider $\varepsilon$-ADR functions with the Individual Size constraints in Section~\ref{b12}. 

\subsubsection{$\varepsilon$-AS Functions with the Individual Size Constraints}\label{b11}

    
\begin{lemma} \label{lemma:ineq3}
For any $j \in [B]$, $2\left[\frac{1+\varepsilon}{1-\varepsilon}f(\pmb{x}^{(j)})-f(\pmb{x}^{(j-1)})\right]\geq{}f(\pmb{o}^{(j-1)})-f(\pmb{o}^{(j)})$. 
\end{lemma}

\begin{proof}[Proof of Lemma \ref{lemma:ineq3}]
%
The arguments for {\bf Case 1} and {\bf Case 2} are similar. We begin with {\bf Case 1}. 

{\bf Case 1.}  First note that  $\Delta_{e^{(j)},i^{(j)}}F(\pmb{x}^{(j-1)})\geq{}\Delta_{o^{(j)},i^{(j)}}F(\pmb{x}^{(j-1)})$ and 
$\Delta_{e^{(j)},i^{(j)}}F(\pmb{x}^{(j-1)})\geq{}\Delta_{e^{(j)},i'}F(\pmb{x}^{(j-1)})$ because $(e^{(j)}, i^{(j)})$ 
is selected by the greedy algorithm. 

It follows that $(1+\varepsilon)f(\pmb{x}^{(j)})\geq \\ (1-\varepsilon)f\left(\left(X^{(j-1)}_1,\cdots, X^{(j-1)}_{i^{(j)}}\cup\{o^{(j)}\},\cdots,X^{(j-1)}_k\right)\right)$  and
$(1+\varepsilon)f(\pmb{x}^{(j)})\geq \\ (1-\varepsilon)f\left(\left(X^{(j-1)}_1,\cdots, X^{(j-1)}_{i'}\cup\{e^{(j)}\},\cdots,X^{(j-1)}_k\right)\right)$ 
by using the definition of approximate submodularity.
Since $\pmb{x}^{(j-1)}\preceq\pmb{o}^{(j-\frac{1}{2})} \preceq\pmb{o}^{(j-\frac{1}{2})}_{(e^{(j)},i')} $, 
we have that $\Delta_{o^{(j)},i^{(j)}}f(\pmb{x}^{(j-1)}) \geq \Delta_{o^{(j)},i^{(j)}}f(\pmb{o}^{(j-\frac{1}{2})})$, 
$\Delta_{o^{(j)},i^{(j)}}f(\pmb{x}^{(j-1)})\geq \Delta_{o^{(j)},i^{(j)}}f(\pmb{o}^{(j-\frac{1}{2})}_{(e^{(j)},i')})$, 
and $\Delta_{e^{(j)},i'}f(\pmb{x}^{(j-1)})\geq\Delta_{e^{(j)},i'}f(\pmb{o}^{(j-\frac{1}{2})})$ from orthant submodularity. 
From the above inequalities, we have that $\frac{1+\varepsilon}{1-\varepsilon}f(\pmb{x}^{(j)})-f(\pmb{x}^{(j-1)})$ is greater than or equal to 
$\Delta_{o^{(j)},i^{(j)}}f(\pmb{o}^{(j-\frac{1}{2})})$, $\Delta_{o^{(j)},i^{(j)}}f(\pmb{o}^{(j-\frac{1}{2})}_{(e^{(j)},i')})$, and $\Delta_{e^{(j)},i'}f(\pmb{o}^{(j-\frac{1}{2})})$. 
As a result, we have 
{\small
\begin{align*}
f(\pmb{o}^{(j-1)})-f(\pmb{o}^{(j)})&=\Delta_{e^{(j)},i'}f(\pmb{o}^{(j-\frac{1}{2})})-\Delta_{o^{(j)},i'}f(\pmb{o}^{(j-\frac{1}{2})}_{(e^{(j)},i^{(j)})})\\
&+\Delta_{o^{(j)},i^{(j)}}f(\pmb{o}^{(j-\frac{1}{2})}_{(e^{(j)},i')}) - \Delta_{e^{(j)},i^{(j)}}f(\pmb{o}^{(j-\frac{1}{2})})\\
&\leq\Delta_{e^{(j)},i'}f(\pmb{o}^{(j-\frac{1}{2})})+\Delta_{o^{(j)},i^{(j)}}f(\pmb{o}^{(j-\frac{1}{2})}_{(e^{(j)},i^{(j)})})\\
&\leq 2\left[\frac{1+\varepsilon}{1-\varepsilon}f(\pmb{x}^{(j)})-f(\pmb{x}^{(j-1)})\right]. 
\end{align*}}

{\bf Case 2.} The argument follows similarly as {\bf Case 1} 
where one can show $\frac{1+\varepsilon}{1-\varepsilon}f(\pmb{x}^{(j)})-f(\pmb{x}^{(j-1)})$ is greater than or equal to 
$\Delta_{o^{(j)},i^{(j)}}f(\pmb{o}^{(j-\frac{1}{2})})$. Therefore, we have 
{\small
\begin{align*}
f(\pmb{o}^{(j-1)})-f(\pmb{o}^{(j)})&=\Delta_{o^{(j)},i^{(j)}}f(\pmb{o}^{(j-\frac{1}{2})}) -\Delta_{e^{(j)},i^{(j)}}f(\pmb{o}^{(j-\frac{1}{2})})\\
&\leq\Delta_{o^{(j)},i^{(j)}}f(\pmb{o}^{(j-\frac{1}{2})})\\ &\leq 2\left[\frac{1+\varepsilon}{1-\varepsilon}f(\pmb{x}^{(j)})-f(\pmb{x}^{(j-1)})\right].\end{align*}
}
\end{proof}

\begin{proof}[Proof of Theorem \ref{theorem:ASIS}]
We have 
\begin{align*}
f(\pmb{o})-f(\pmb{x}) &=\sum_{j\in[B]}\left(f(\pmb{o}^{(j-1)})-f(\pmb{o}^{(j)})\right) \\
&\leq 2\sum_{j\in[B]}\left(\frac{1+\varepsilon}{1-\varepsilon}f(\pmb{x}^{(j)})-f(\pmb{x}^{j-1})\right)\\
&\leq\frac{2-2\varepsilon+2\varepsilon{}B}{1-\varepsilon}f(\pmb{x}),
\end{align*}
where the first inequality is due to Lemma \ref{lemma:ineq3}. 
Hence, we have $F(\pmb{x})\geq\frac{(1-\varepsilon)^2}{(3-3\varepsilon+2\varepsilon{}B)(1+\varepsilon)}F(\pmb{o})$.
\end{proof}

%
%
%


\vspace{+3mm}
\subsubsection{$\varepsilon$-ADR Functions with the Individual Size Constraints}\label{b12}

    

\begin{lemma} \label{lemma:ineq4}
For any $j \in [B]$, $2\left[F(\pmb{x}^{(j)})-F(\pmb{x}^{(j-1)})\right] \geq \frac{1-\varepsilon}{1+\varepsilon}\left[F(\pmb{o}^{(j-1)})-F(\pmb{o}^{(j)})\right]$.
\end{lemma}


\begin{proof}[Proof of Lemma \ref{lemma:ineq4}]
The arguments for {\bf Case 1} and {\bf Case 2} are similar. We begin with {\bf Case 1}. 

{\bf Case 1.} First note that since $(e^{(j)}, i^{(j)})$ is selected by the greedy algorithm, 
we have $\Delta_{e^{(j)},i^{(j)}}F(\pmb{x}^{(j-1)})\geq{}\Delta_{o^{(j)},i^{(j)}}F(\pmb{x}^{(j-1)})$ and 
$\Delta_{e^{(j)},i^{(j)}}F(\pmb{x}^{(j-1)})\geq{}\Delta_{e^{(j)},i'}F(\pmb{x}^{(j-1)})$. 
Since $\pmb{x}^{(j-1)}\preceq\pmb{o}^{(j-\frac{1}{2})} \preceq\pmb{o}^{(j-\frac{1}{2})}_{(e^{(j)},i')}$, 
we have that $\Delta_{o^{(j)},i^{(j)}}f(\pmb{x}^{(j-1)})\geq \Delta_{o^{(j)},i^{(j)}}f(\pmb{o}^{(j-\frac{1}{2})}) \ge \frac{1}{1+\varepsilon} \Delta_{o^{(j)},i^{(j)}} F(\pmb{o}^{(j-\frac{1}{2})}) $,
$\Delta_{o^{(j)},i^{(j)}}f(\pmb{x}^{(j-1)})\geq \Delta_{o^{(j)},i^{(j)}}f(\pmb{o}^{(j-\frac{1}{2})}_{(e^{(j)},i')}) \ge \frac{1}{1+\varepsilon} \Delta_{o^{(j)},i^{(j)}}F(\pmb{o}^{(j-\frac{1}{2})}_{(e^{(j)},i')})$, 
and 
$\Delta_{e^{(j)},i'}f(\pmb{x}^{(j-1)}) \geq\Delta_{e^{(j)},i'}f(\pmb{x}^{(j-1)}) \ge \frac{1}{1+\varepsilon} \geq\Delta_{e^{(j)},i'}F(\pmb{x}^{(j-1)}) $ 
from orthant submodularity. 

Therefore, we obtain that $\frac{1}{1-\varepsilon}\Delta_{e^{(j)},i^{(j)}}F(\pmb{x}^{(j-1)})$ is greater than or equal to 
$\frac{1}{1+\varepsilon} \Delta_{o^{(j)},i^{(j)}} F(\pmb{o}^{(j-\frac{1}{2})})$, $\frac{1}{1+\varepsilon}  \Delta_{o^{(j)},i^{(j)}} F(\pmb{o}^{(j-\frac{1}{2})})$, 
and $\frac{1}{1+\varepsilon} \Delta_{e^{(j)},i'} F(\pmb{x}^{(j-1)})$ by the definition of $\varepsilon$-ADR. 
\begin{align*}
&\frac{1}{1+\varepsilon} \left[ F(\pmb{o}^{(j-1)}) - F(\pmb{o}^{(j)}) \right] \\
=&\frac{1}{1+\varepsilon} \left[ \Delta_{e^{(j)},i'} F(\pmb{o}^{(j-\frac{1}{2})})-\Delta_{o^{(j)},i'} F(\pmb{o}^{(j-\frac{1}{2})}_{(e^{(j)},i^{(j)})})\right.\\
&~~\quad\quad\quad \left.+\Delta_{o^{(j)},i^{(j)}} F(\pmb{o}^{(j-\frac{1}{2})}_{(e^{(j)},i')}) - \Delta_{e^{(j)},i^{(j)}} F(\pmb{o}^{(j-\frac{1}{2})}) \right]\\
\leq&\frac{1}{1+\varepsilon}  \left [ \Delta_{e^{(j)},i'} F(\pmb{o}^{(j-\frac{1}{2})}) + \Delta_{o^{(j)},i^{(j)}} F(\pmb{o}^{(j-\frac{1}{2})}_{(e^{(j)},i^{(j)})}) \right]\\
\leq&\frac{2}{1-\varepsilon}\left[F(\pmb{x}^{(j)})-F(\pmb{x}^{(j-1)})\right]. 
\end{align*}

{\bf Case 2.} The argument follows similarly as {\bf Case 1} where one can show $\frac{1}{1-\varepsilon}\Delta_{e^{(j)},i^{(j)}}F(\pmb{x}^{(j-1)})$ is greater than or equal to 
$\frac{1}{1+\varepsilon} \Delta_{o^{(j)},i^{(j)}} F(\pmb{o}^{(j-\frac{1}{2})})$. Therefore, we have 

\begin{align*}
&\frac{1}{1+\varepsilon} \left[F(\pmb{o}^{(j-1)}) - F(\pmb{o}^{(j)}) \right] \\
=&\frac{1}{1+\varepsilon} \left[ \Delta_{o^{(j)},i^{(j)}} F(\pmb{o}^{(j-\frac{1}{2})}) - \Delta_{e^{(j)},i^{(j)}} F(\pmb{o}^{(j-\frac{1}{2})}) \right]\\
\leq&\frac{1}{1+\varepsilon}\Delta_{o^{(j)},i^{(j)}} F(\pmb{o}^{(j-\frac{1}{2})}) \\
\leq&\frac{2}{1-\varepsilon}\left[F(\pmb{x}^{(j)})-F(\pmb{x}^{(j-1)})\right]. 
\end{align*}
\end{proof}
    
\begin{proof}[Proof of Theorem \ref{theorem:ADRIS}]     
We have
\begin{align*}
F(\pmb{o})-F(\pmb{x})&=\sum_{j\in[B]}\left[F(\pmb{o}^{(j-1)})-F(\pmb{o}^{(j)})\right]\\
&\leq\frac{2(1+\varepsilon)}{1-\varepsilon}\sum_{j\in[B]}\left[F(\pmb{x}^{(j)})-F(\pmb{x}^{(j-1)})\right]\\
&\leq\frac{2(1+\varepsilon)}{1-\varepsilon}F(\pmb{x})
\end{align*}
where the first inequality is due to Lemma \ref{lemma:ineq4}. 
Hence, we have $F(\pmb{x})\geq{}\frac{1-\varepsilon}{3+\varepsilon}F(\pmb{o})$.
\end{proof}


\section{Improved Greedy Approximation Ratios When $f$ is Known: Additional material}\label{sec:improvedgreedy}


We start by restating the theorem presented in Section~\ref{sec:paper:improvedgreedy}.   

\begin{theorem} \label{theo:ftoF}
Let $f$ be a $k$-submodular function and $F$ be an $\varepsilon$-approximately $k$-submodular function that is bounded by $f$. 
If there is an algorithm that provides an approximation ratio of $\alpha$ for maximizing $f$ subject to constraint $\mathbb{X}$, 
then the same solution yields an approximation ratio of $\frac{1-\varepsilon}{1+\varepsilon} \alpha$ for maximizing $F$ subject to constraint $\mathbb{X}$. 
\end{theorem}
\begin{proof}
Let $\pmb{o}_f$ and $\pmb{o}_F$ be the optimal solutions of $f$ and $F$, respectively, subject to constraint $\mathbb{X}$. 
Let $\pmb{x}_f$ be a solution of $f$ returned by an algorithm with an approximation ratio of $\alpha$. 
We have
{\small
\begin{align*}
\frac{1}{1 - \varepsilon} F(\pmb{x}_f) \ge f(\pmb{x}_f) \ge \alpha f(\pmb{o}_f) \ge \alpha f(\pmb{o}_F) \ge \frac{1}{1+\varepsilon} \alpha F(\pmb{o}_F), 
\end{align*}
}
where the first inequality is by applying the definition of $\varepsilon$-AS, 
the second inequality is by the definition of approximation ratios, 
the third inequality is by replacing $\pmb{o}_f$ with a less optimal solution $\pmb{o}_F$, 
and the last inequality is by the definition of $\varepsilon$-AS. 
\end{proof}

The above theorem provides a set of results for our settings. 

 \begin{corollary}
 \label{5.1again}
 Suppose $F$ is approximately submodular.  
 By applying the greedy algorithm on $f$ subject to a size constraint, we obtain a solution providing an approximation 
 ratio of $\frac{1-\varepsilon}{1+\varepsilon} (1-\frac{1}{e})$ to the size constrained maximization problem of $F$. 
 \end{corollary}

The above result follows from the known approximation ratio of greedy algorithms for monotone submodular functions \cite{Cornuejols:1977aa,Nemhauser:1978aa}. 
Comparing to the $\varepsilon$-AS result of \cite{Horel:2016aa}, when $B\geq\frac{1}{4(e-1)\varepsilon}+\frac{1}{2}$ \footnote{$B\geq\frac{1}{4(e-1)\varepsilon}+\frac{1}{2}$ implies $B\geq\frac{1}{4(1-1/e)}\left(\frac{1}{\varepsilon}-\varepsilon-\left(1-\frac{1}{e}\right)\left(\frac{1}{\varepsilon}-2+\varepsilon\right)\right)$. Thus, $B\geq\frac{1-\varepsilon^2-(1-1/e)(1-\varepsilon)^2}{4(1-1/e)\varepsilon}$. Then we have $\frac{1}{1+\frac{4B\varepsilon}{(1-\varepsilon)^2}}\leq\frac{1-\varepsilon}{1+\varepsilon}(1-\frac{1}{e})$}, 
Corollary \ref{5.1again} yields a better approximation ratio. Specifically, we obtain the corollaries of Section~\ref{sec:paper:improvedgreedy} which we copy below for convenience. 

\begin{corollary}
Suppose $F$ is approximately $k$-submodular. 
By applying the $k$-Greedy-TS algorithm on $f$ subject to a total size constraint, we obtain a solution providing an approximation 
ratio of $\frac{1-\varepsilon}{2(1+\varepsilon)}$ to total size constrained maximization problem of $F$. 
\end{corollary}

\begin{corollary}
Suppose $F$ is approximately $k$-submodular. 
By applying the $k$-Greedy-IS algorithm on $f$ subject to an individual size constraint, we obtain a solution providing  an approximation 
ratio of $\frac{1-\varepsilon}{3(1+\varepsilon)}$ to individual size constrained maximization problem of $F$. 
\end{corollary}

The above corollaries follow from the results of \cite{Ohsaka:2015aa} where one can obtain approximation ratios 
of $\frac{1}{2}$ and $\frac{1}{3}$ for total size and individual size constraints, respectively, for maximizing monotone $k$-submodular functions. 
It turns out that, for any value of $\varepsilon$, we can derive better theoretical guarantees 
using the the greedy solutions from $f$ according to the above corollaries for $\varepsilon$-AS $F$. 
For $F$ that is $\varepsilon$-ADR, since it is also $\varepsilon$-AS $F$, the above results apply immediately. 
However, the above results provide a weaker guarantee than applying greedy algorithms on $F$ directly. 

\section{Additional Experiment Details} \label{sec:addexp}
\subsection{Non-submodularity}\label{sec:appendixnon}

To see that the F constructed in the AG setting is not k-submodular, consider $k=1$, $\xi(u)=1-\varepsilon$, $\xi(v)=1$, $\xi(\{u, v\})=1$, as well as $f(\{u, v\})-f(\{v\})=f(\{u\})$. It follows that $F(\{u\})-F(\mathbf{0})<F(\{u,v\}-F(\{v\}) \Leftrightarrow (1-\varepsilon)f(\{u\})=(1-\varepsilon)(f(\{u,v\}-f(\{v\}))<f(\{u,v\}-f(\{v\})=F(\{u,v\}-F(\{v\})$ for $\varepsilon>0$.

To see that $F(\pmb x)=\max_{x\in {supp(\pmb x)}}\xi(x)\cdot f(\pmb x)$ constructed in the MaxG setting is not $k$-submodular, 
consider $k=1$ and two elements $\{u,v\}$, 
$\xi(u)=1-\varepsilon$, $\xi(v)=1$, and  $f(\{u,v\})-f(\{v\})=f(\{u\})$ for some $u, v$. 
It follows that $F(\{u\})-F(\mathbf{0})<F(\{u,v\})-F(\{v\}) \Leftrightarrow (1-\varepsilon)f(\{u\}) < f(\{u,v\})-f(\{v\})$ for $\varepsilon>0$. 

To see that $F(\pmb x)=\xi(\pmb x)\cdot f(\pmb x)$ with $\xi({\pmb x})=\frac{\sum_{x\in supp(\pmb  x)}\xi(x)}{|supp(\pmb x)|}$, which is constructed in the MeanG setting, is not $k$-submodular, consider $k=1$, $\xi(u)=1-\varepsilon$, and $\xi(v)=1$, as well as 
$f(\{u, v\})=\frac{3}{2}f(\{u\})$ and $f(\{u, v\})-f(\{v\})=f(\{u\})$ for some $u, v$. It follows that $F(\{u\})-F(\mathbf{0})<F(\{u, v\}-F(\{v\}) \Leftrightarrow (1-\varepsilon)f(\{u\})<(1-\frac{3}{4}\varepsilon)f(\{u\})=f(\{u\})-\frac{3\varepsilon}{4}f(\{u\})=f(\{u,v\})-f(\{v\})-\frac{\varepsilon}{2}f(\{u, v\})=\frac{2-\varepsilon}{2}f(\{u,v\})-f(\{v\})=F(\{u,v\})-F(\{v\})$ for $\varepsilon>0$.


\subsection{Influence Maximization with Approximately $k$-Submodular Functions and the TS Constraint 
}\label{appendix:E1}
In the main paper, we considered Influence Maximization with the impact of $k$. Here, we present results in Figs.~\ref{fig:imAGAll}, ~\ref{fig:imMeanAll} and ~\ref{fig:imMaxAll} for the same problem with the impact of $B$ and $\varepsilon$. The trends are similar to those for parameter $k$ reported in the paper; $Gr$-$H$ outperformed $Gr$-$\tilde{H}$ in the AG setting, and the opposite happened in the MeanG and MaxG settings. 

\begin{figure}[!ht]
    \begin{subfigure}{0.22\textwidth}
        \centering
        \captionsetup{justification=centering}
        \includegraphics[width = \linewidth]{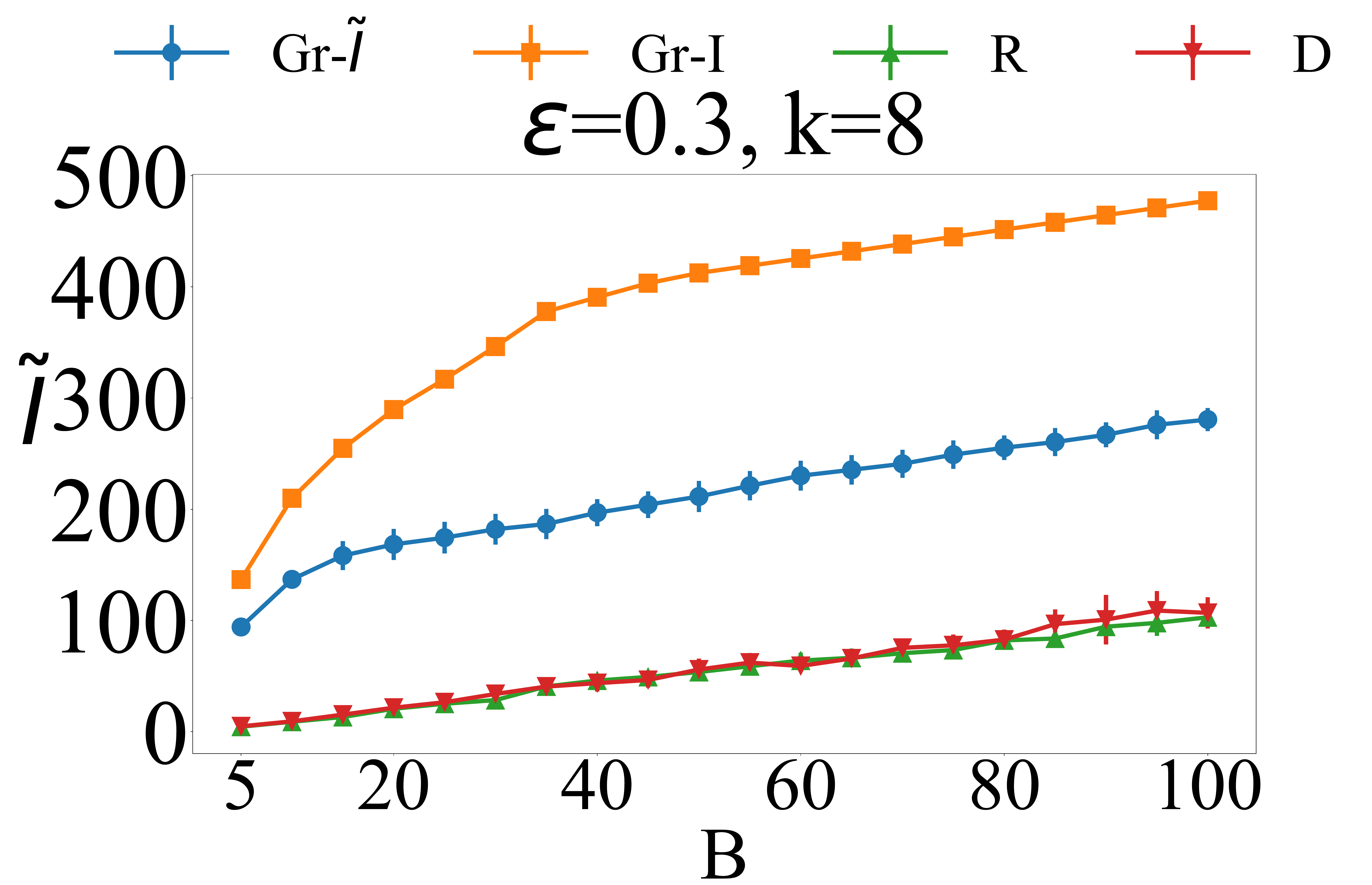}
        \caption{}\label{fig:imAG1}
    \end{subfigure}%
    \begin{subfigure}{0.22\textwidth}
        \centering
        \captionsetup{justification=centering}
        \includegraphics[width = \linewidth]{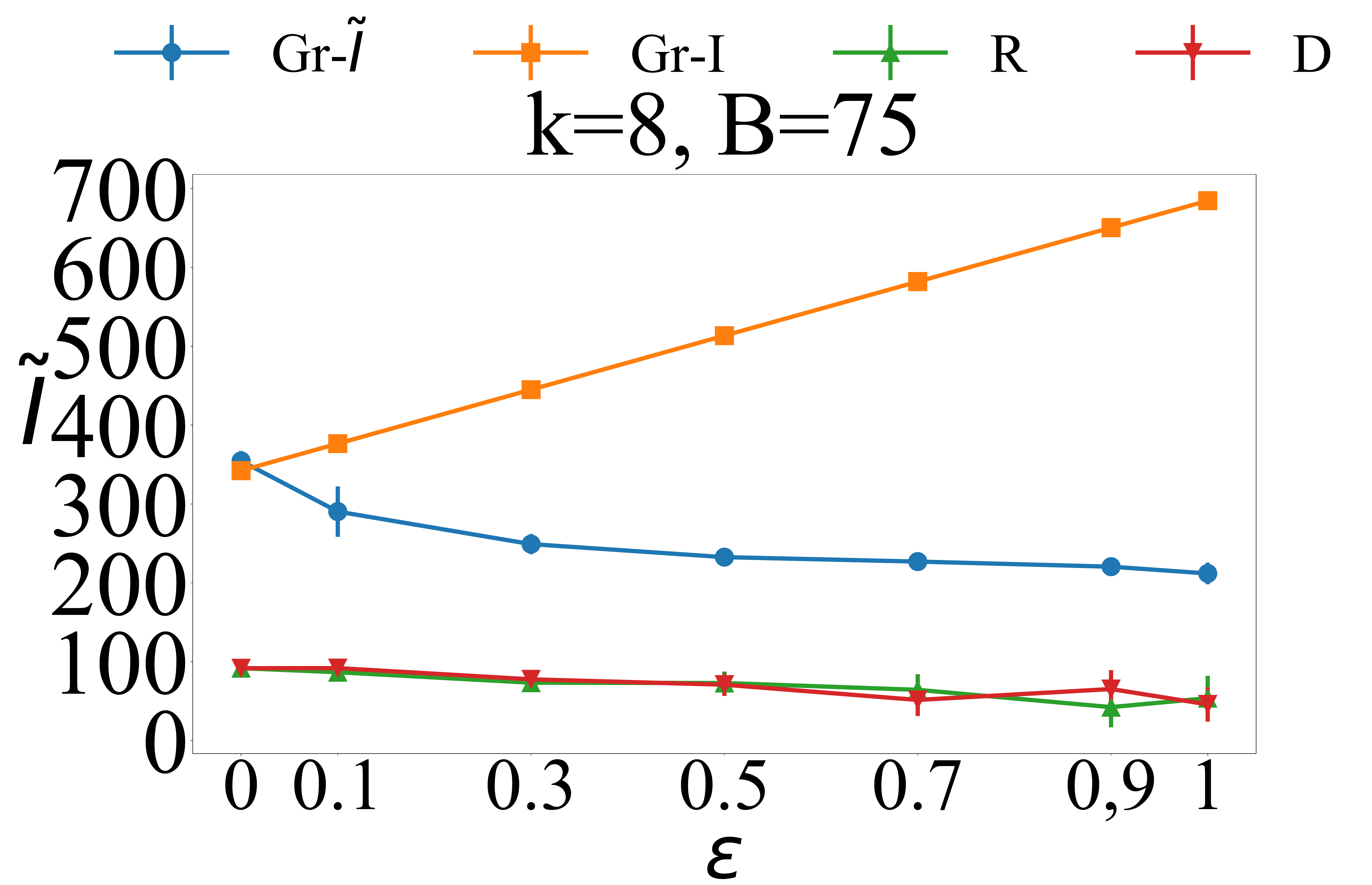}
        \caption{}\label{fig:imAG2}
    \end{subfigure}%
    \vspace{-2mm}
    \caption{$\tilde{I}$ in AG setting vs: (a) $B$, (b) $\varepsilon$.}\label{fig:imAGAll}
\end{figure}
 
\begin{figure}[!ht]
    \begin{subfigure}{0.22\textwidth}
        \centering
        \captionsetup{justification=centering}
        \includegraphics[width = \linewidth]{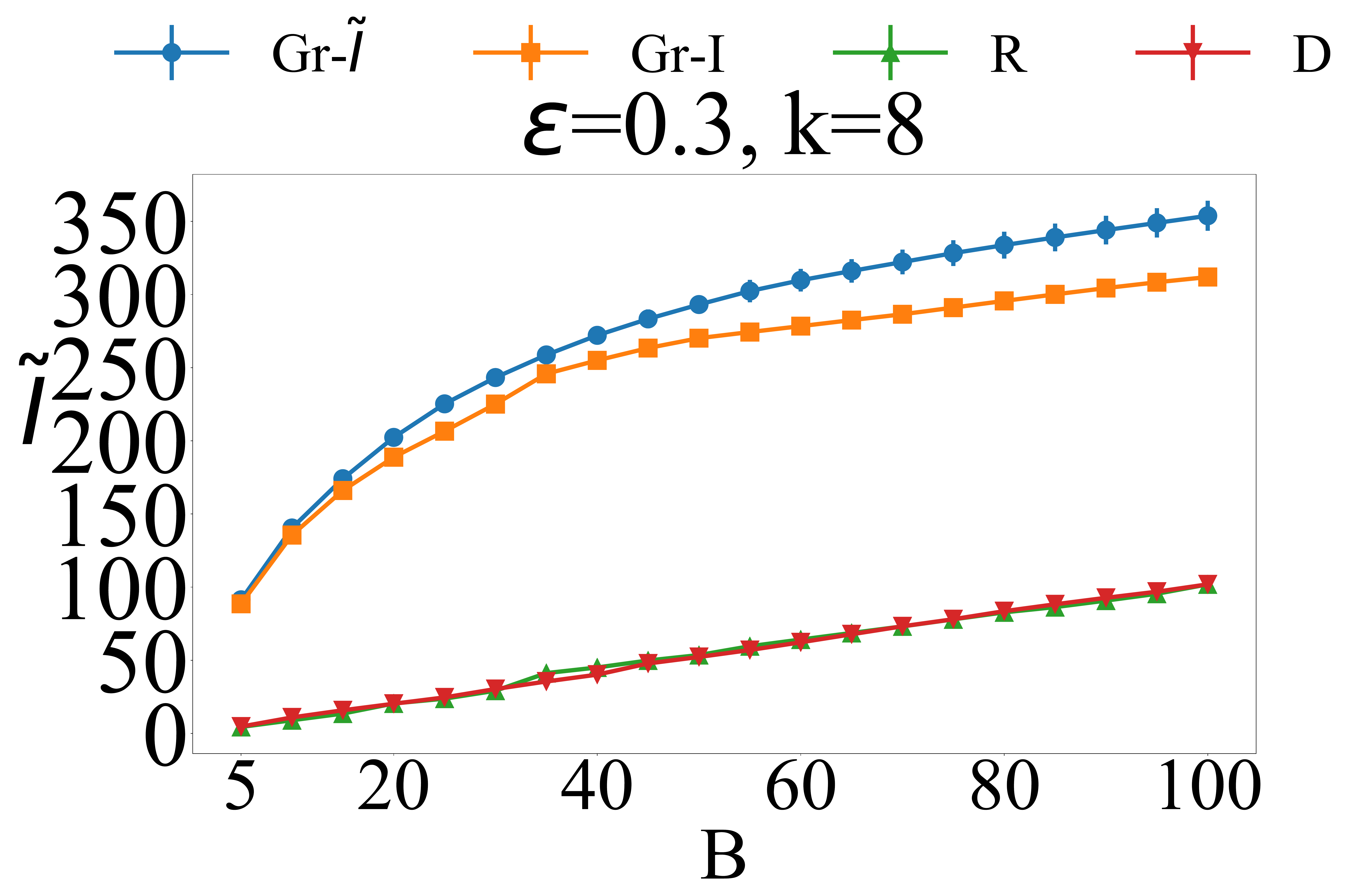}
        \caption{}\label{fig:imMean1}
    \end{subfigure}%
    \begin{subfigure}{0.22\textwidth}
        \centering
        \captionsetup{justification=centering}
        \includegraphics[width = \linewidth]{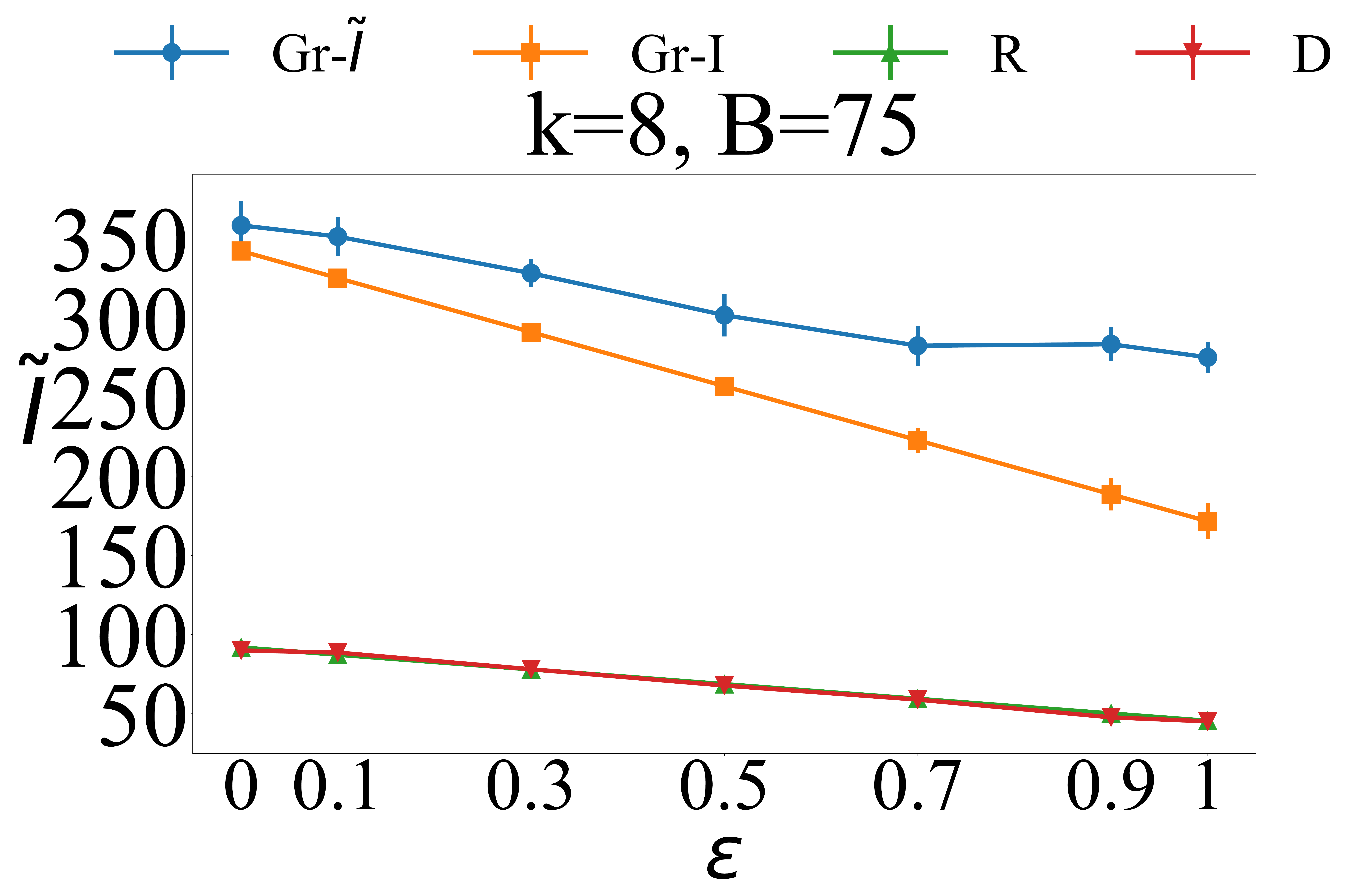}
        \caption{}\label{fig:imMean2}
    \end{subfigure}%
    \vspace{-2mm}
    \caption{$\tilde{I}$ in Mean setting vs: (a) $B$, (b) $\varepsilon$.}\label{fig:imMeanAll}
\end{figure}

\begin{figure}[!ht]
    \begin{subfigure}{0.22\textwidth}
        \centering
        \captionsetup{justification=centering}
        \includegraphics[width = \linewidth]{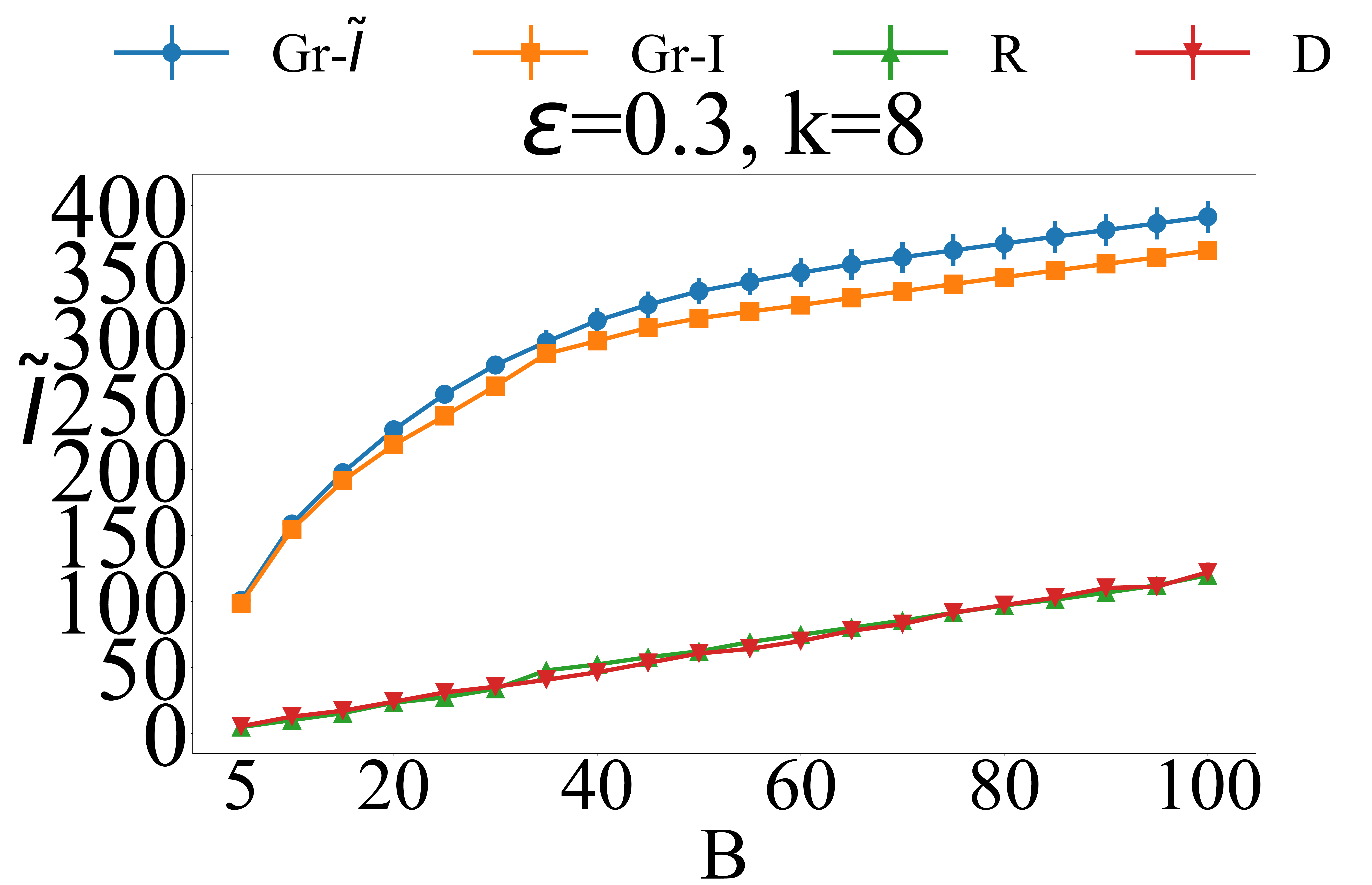}
        \caption{}\label{fig:imMax1}
    \end{subfigure}%
    \begin{subfigure}{0.22\textwidth}
        \centering
        \captionsetup{justification=centering}
        \includegraphics[width = \linewidth]{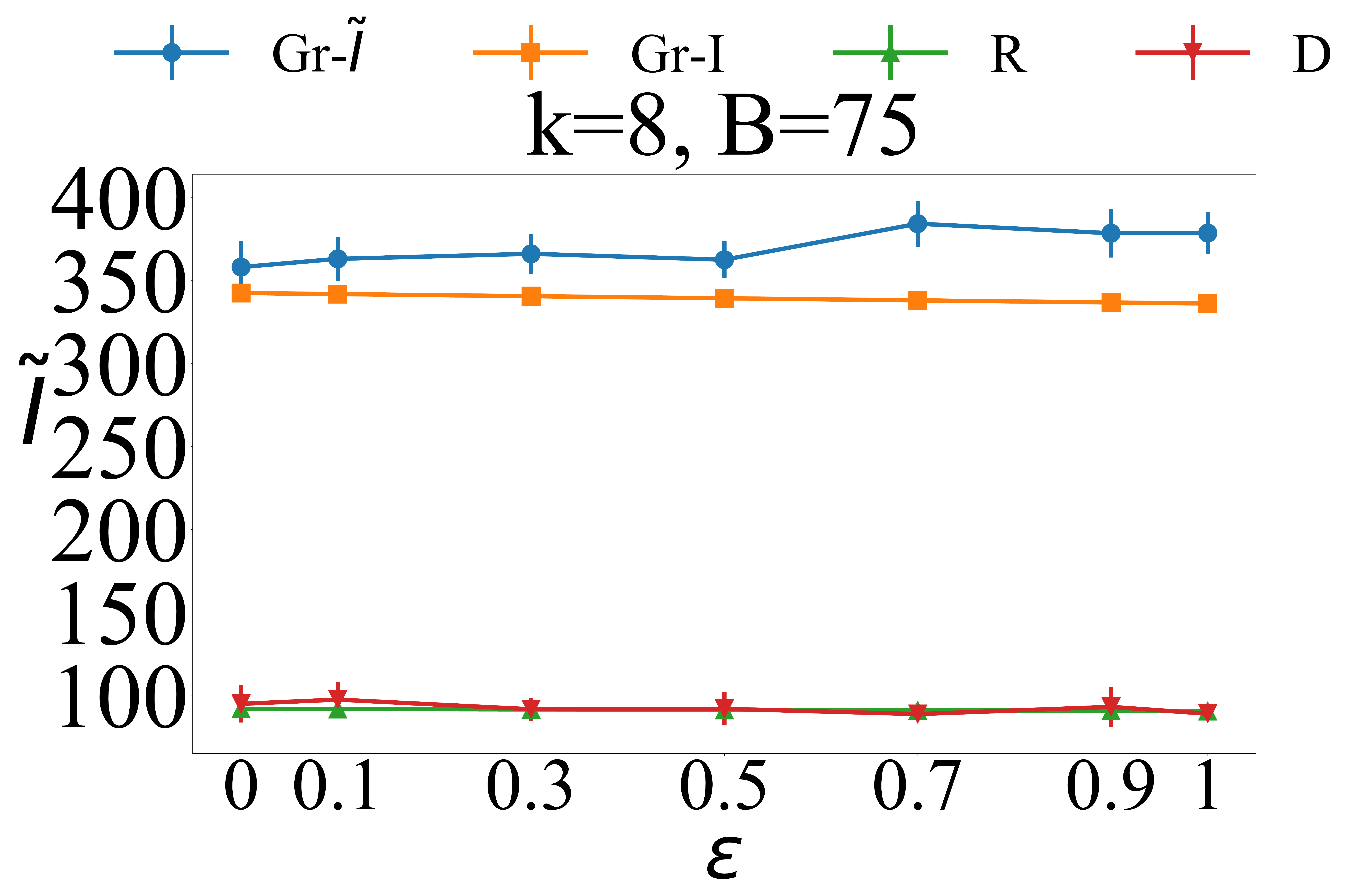}
        \caption{}\label{fig:imMax2}
    \end{subfigure}%
    \caption{$\tilde{I}$ in MaxG setting vs: (a) $B$, (b) $\varepsilon$.}\label{fig:imMaxAll}
     \vspace{+0.5cm}
\end{figure}


\vspace{+3mm}
\subsection{Sensor Placement with Approximately $k$-Submodular Functions and the TS Constraint}\label{sec:appendix:sp-TS-epsAS}

In the main paper, we considered Sensor Placement with IS constraints. Here, we present results for the same problem with TS constraint in Figs.~\ref{fig:spAGTSall}, ~\ref{fig:spMeanTSall} and ~\ref{fig:spMaxTSall}. As can be seen, the results are qualitatively similar to those for the problem with IS constraints. That is, $Gr$-$H$ outperformed $Gr$-$\tilde{H}$ in the AG setting, and the opposite happened in the MeanG and MaX settings.




\vspace{+3mm}
  \begin{figure}[!ht]
   
    \begin{subfigure}{0.17\textwidth}
        \centering
        \captionsetup{justification=centering}
        \includegraphics[width = \linewidth]{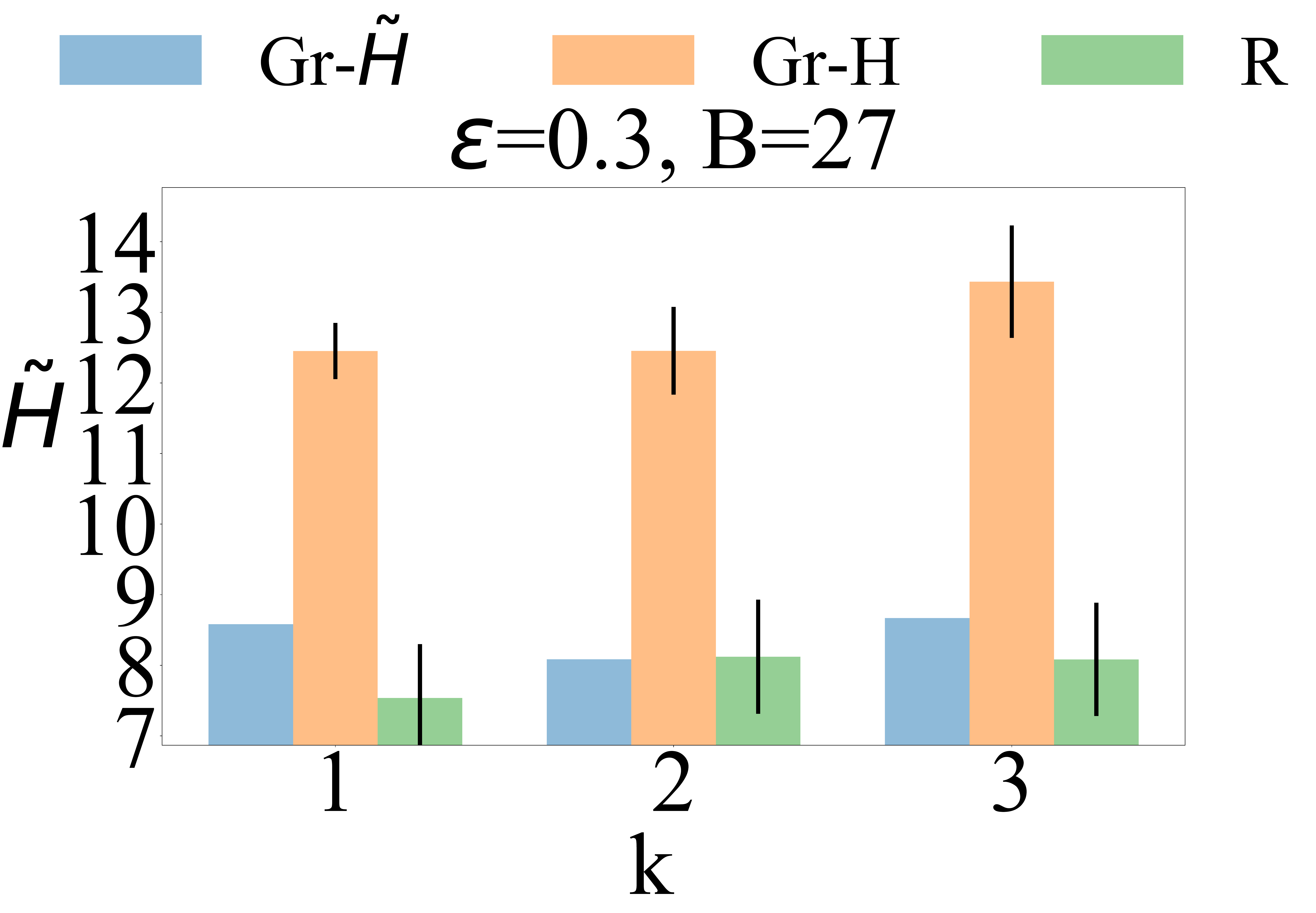}
        \caption{}\label{fig:spNew1_TS_k}
    \end{subfigure}%
    \begin{subfigure}{0.17\textwidth}
        \centering
        \captionsetup{justification=centering}
        \includegraphics[width = \linewidth]{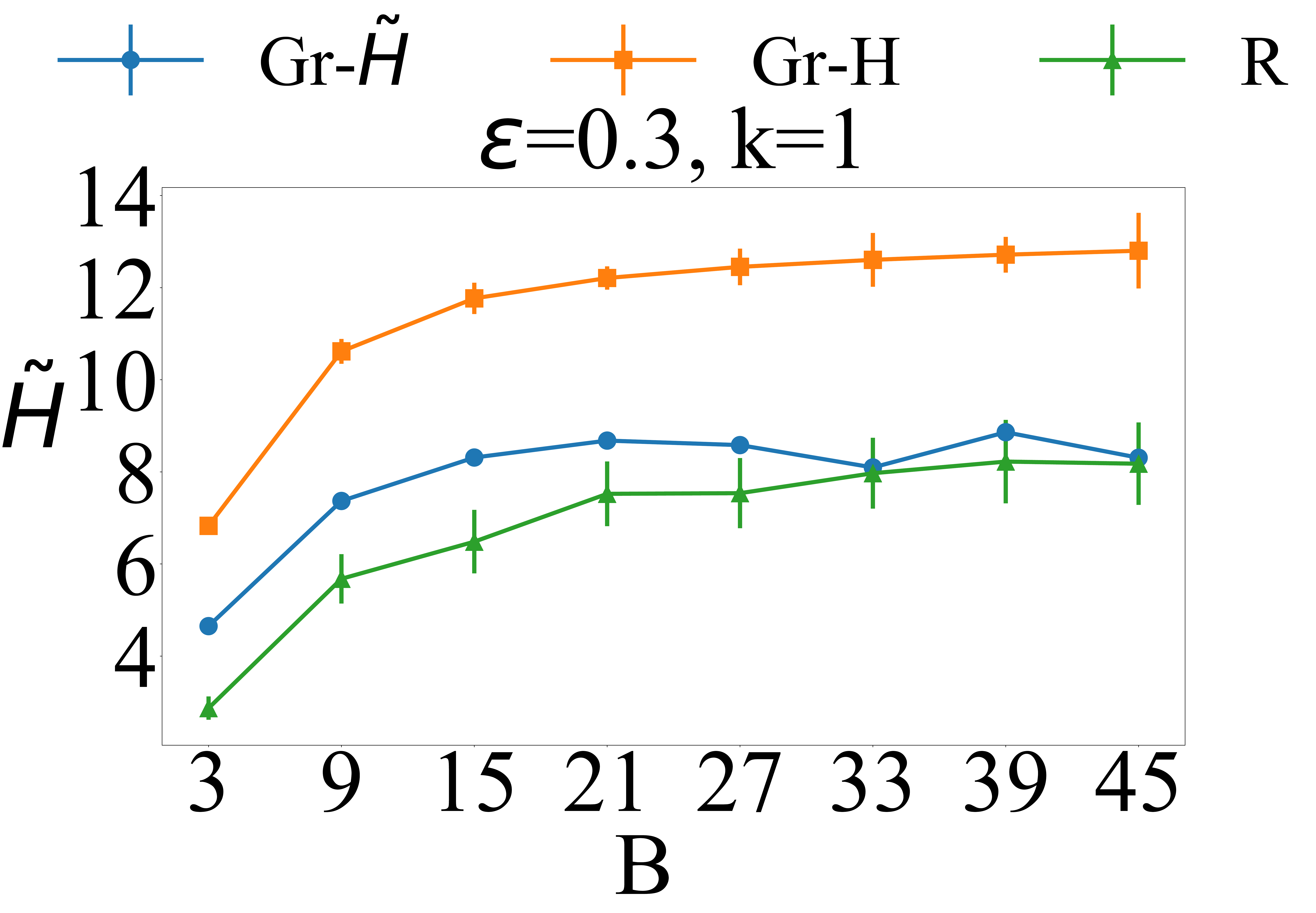}
        \caption{}\label{fig:spAG1_TS_B}
    \end{subfigure}%
    \begin{subfigure}{0.17\textwidth}
        \centering
        \captionsetup{justification=centering}
        \includegraphics[width = \linewidth]{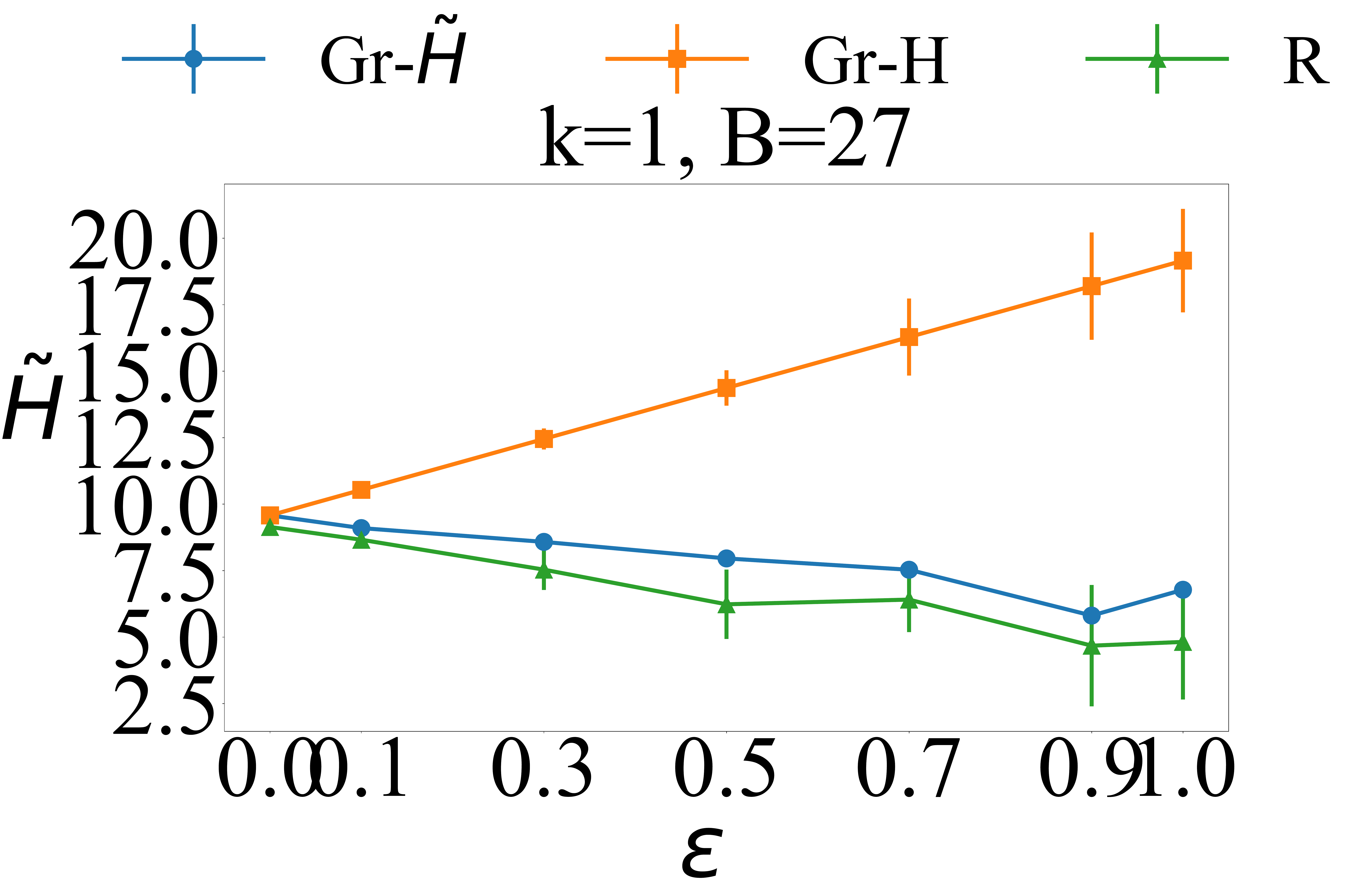}
        \caption{}\label{fig:spAG1_TS_B}
    \end{subfigure}%
    
\caption{$\tilde{H}$ in AG setting vs: (a) $k$, (b)  $B$, (c) $\varepsilon$.}\label{fig:spAGTSall}
  \end{figure}
  
  \begin{figure}[!ht]
       \begin{subfigure}{0.17\textwidth}
        \centering
        \captionsetup{justification=centering}
        \includegraphics[width = \linewidth]{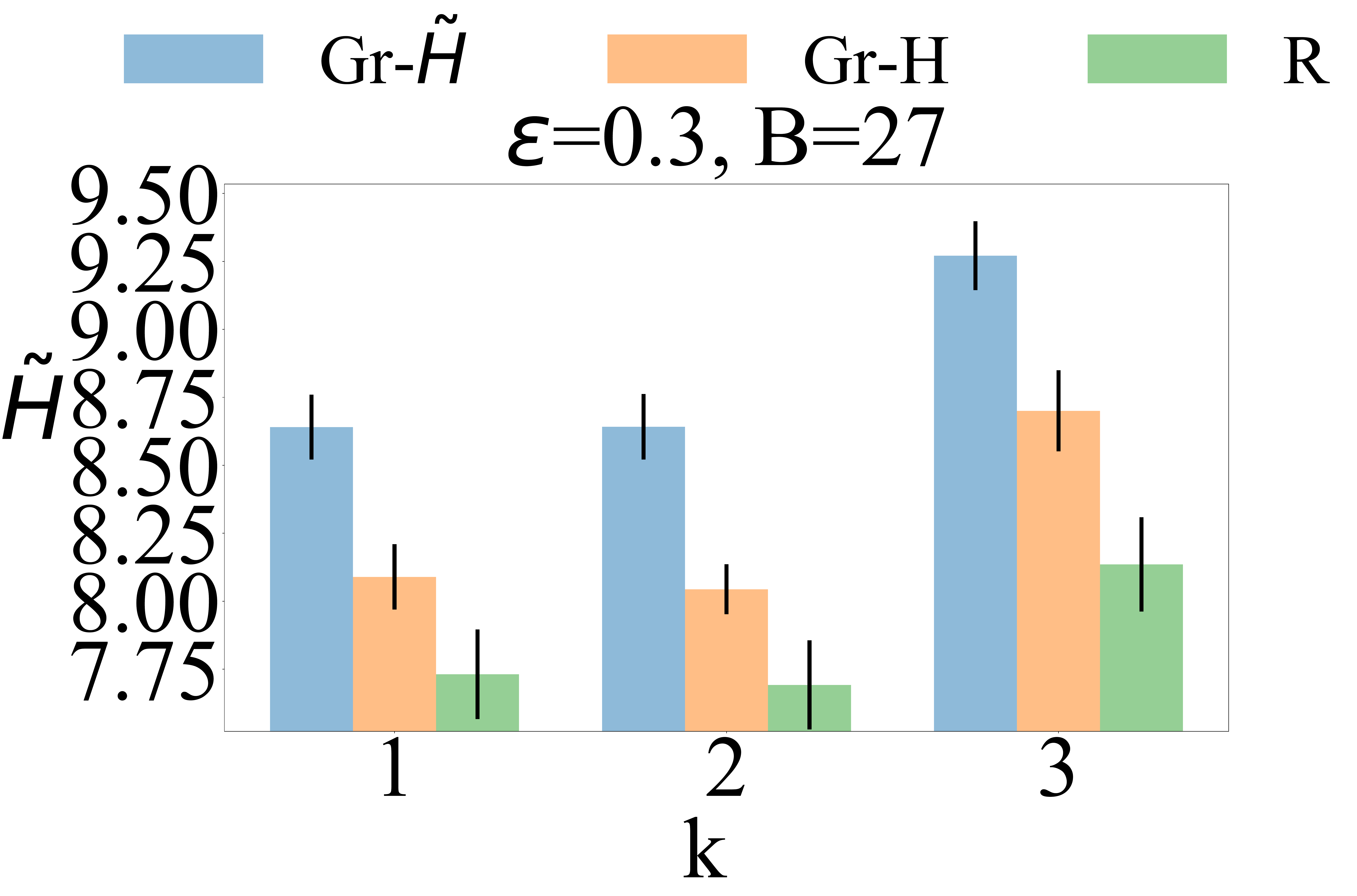}
        \caption{}\label{fig:spMean1_TS_k}
    \end{subfigure}%
    \begin{subfigure}{0.17\textwidth}
        \centering
        \captionsetup{justification=centering}
        \includegraphics[width = \linewidth]{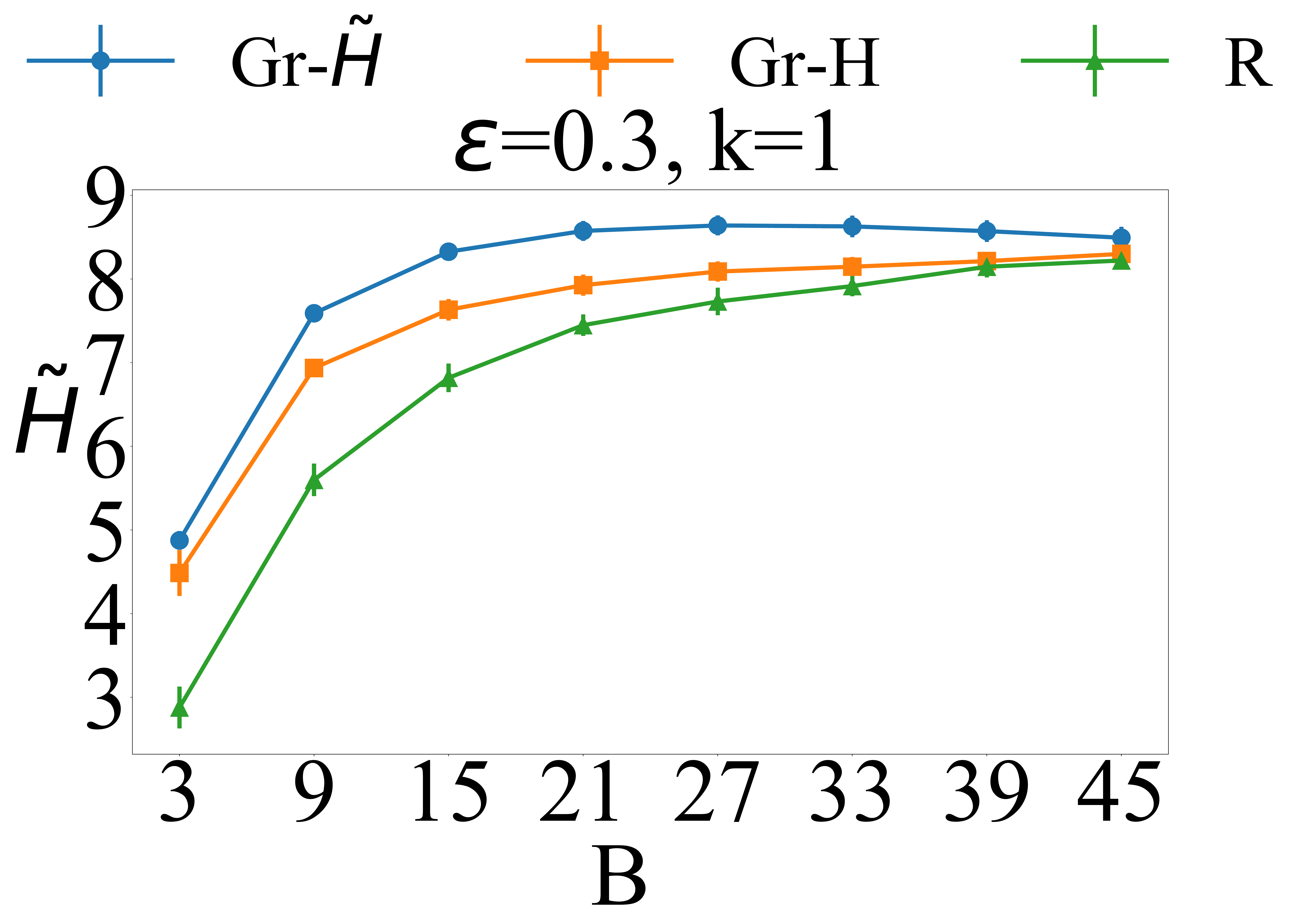}
        \caption{}\label{fig:spMean1_TS_B}
    \end{subfigure}%
    \begin{subfigure}{0.17\textwidth}
        \centering
        \captionsetup{justification=centering}
        \includegraphics[width = \linewidth]{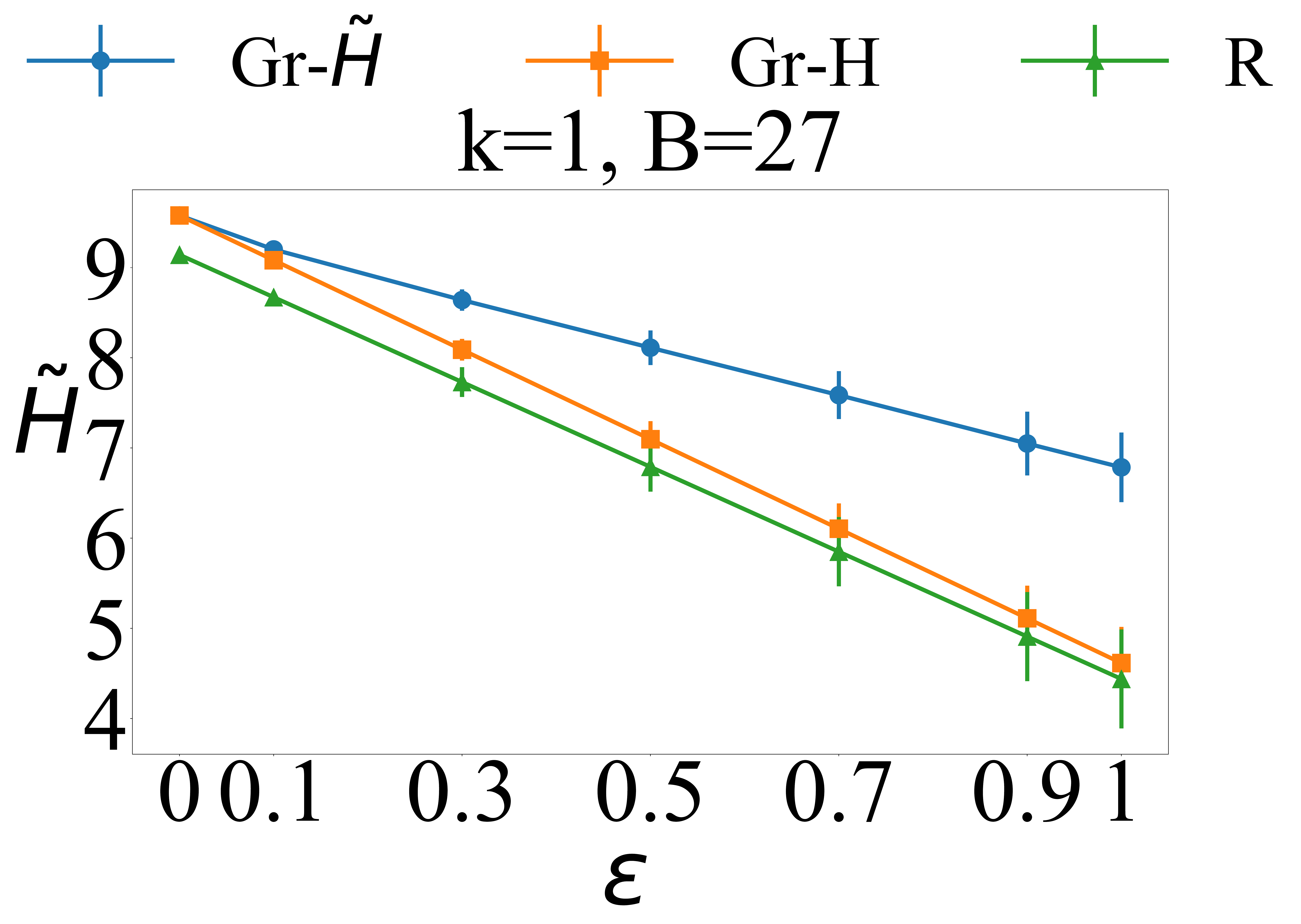}
        \caption{}\label{fig:spMean1_TS_B}
    \end{subfigure}%
    
\caption{$\tilde{H}$ in MeanG setting vs: (a) $k$, (b) $B$, (c) $\varepsilon$.}\label{fig:spMeanTSall}
  \end{figure}
  
  \begin{figure}[!ht]
    \begin{subfigure}{0.17\textwidth}
        \centering
        \captionsetup{justification=centering}
        \includegraphics[width = \linewidth]{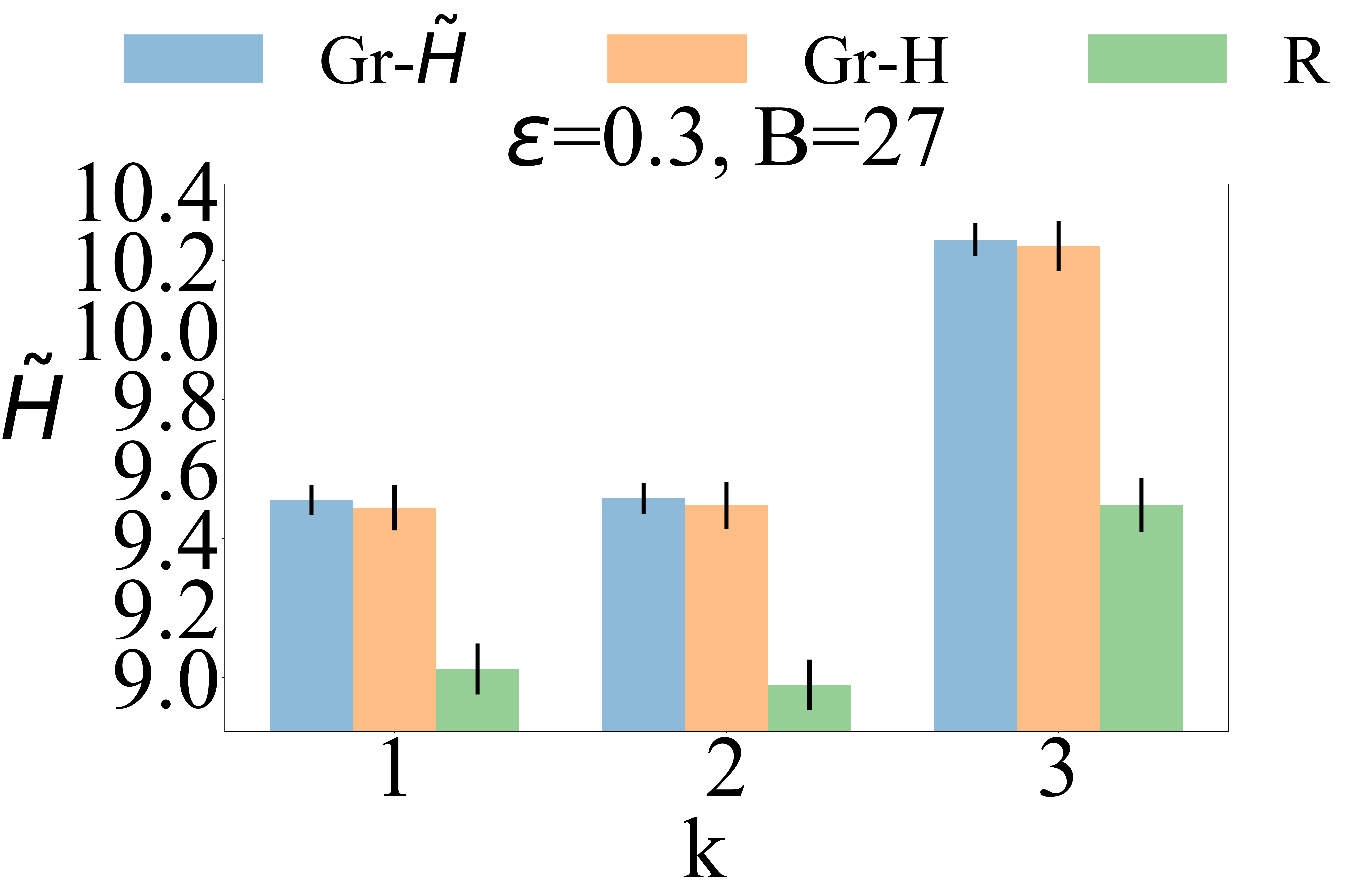}
        \caption{}\label{fig:spMax1_TS_k}
    \end{subfigure}%
    \begin{subfigure}{0.17\textwidth}
        \centering
        \captionsetup{justification=centering}
        \includegraphics[width = \linewidth]{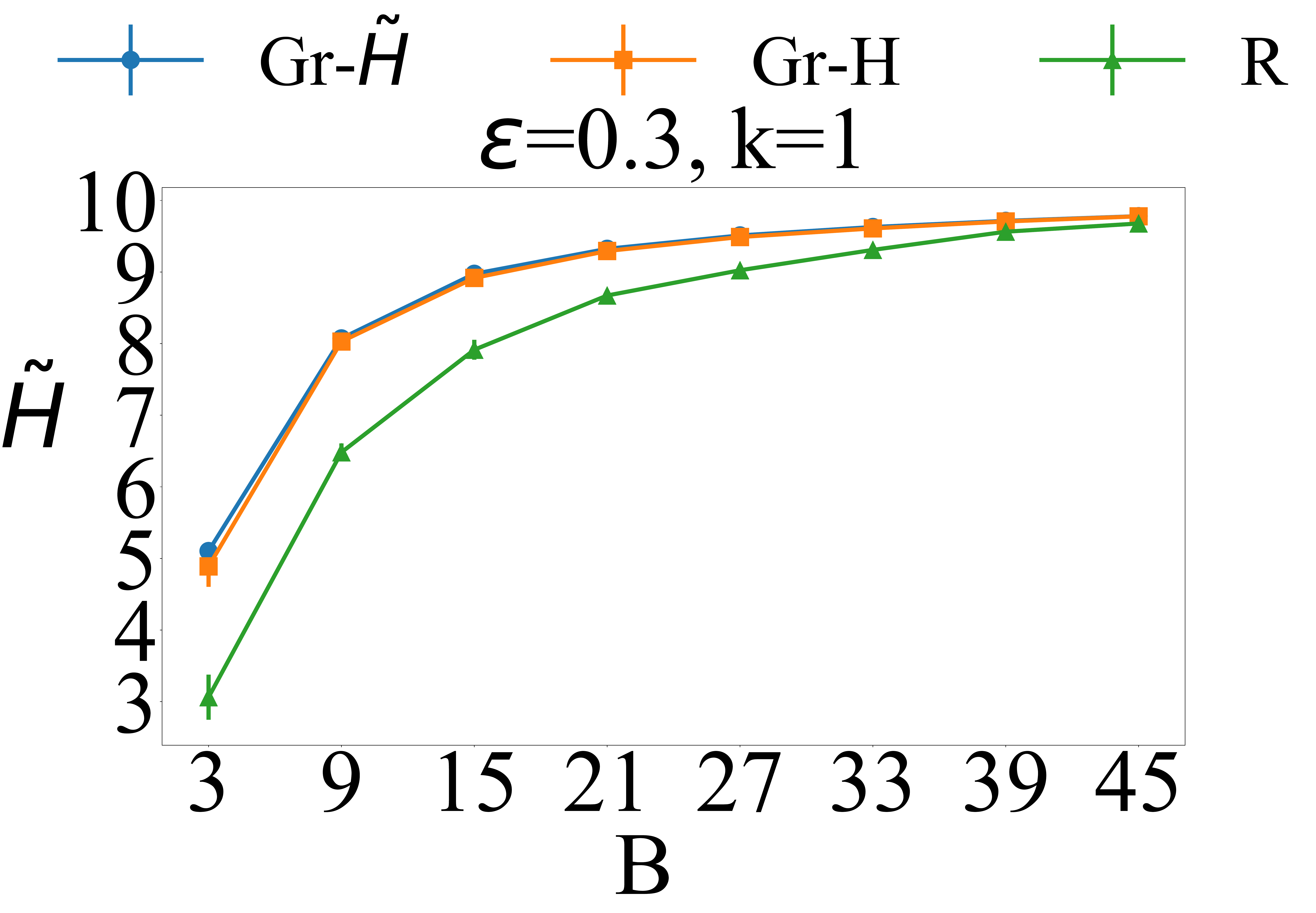}
        \caption{}\label{fig:spMax1_TS_B}
    \end{subfigure}%
    \begin{subfigure}{0.17\textwidth}
        \centering
        \captionsetup{justification=centering}
        \includegraphics[width = \linewidth]{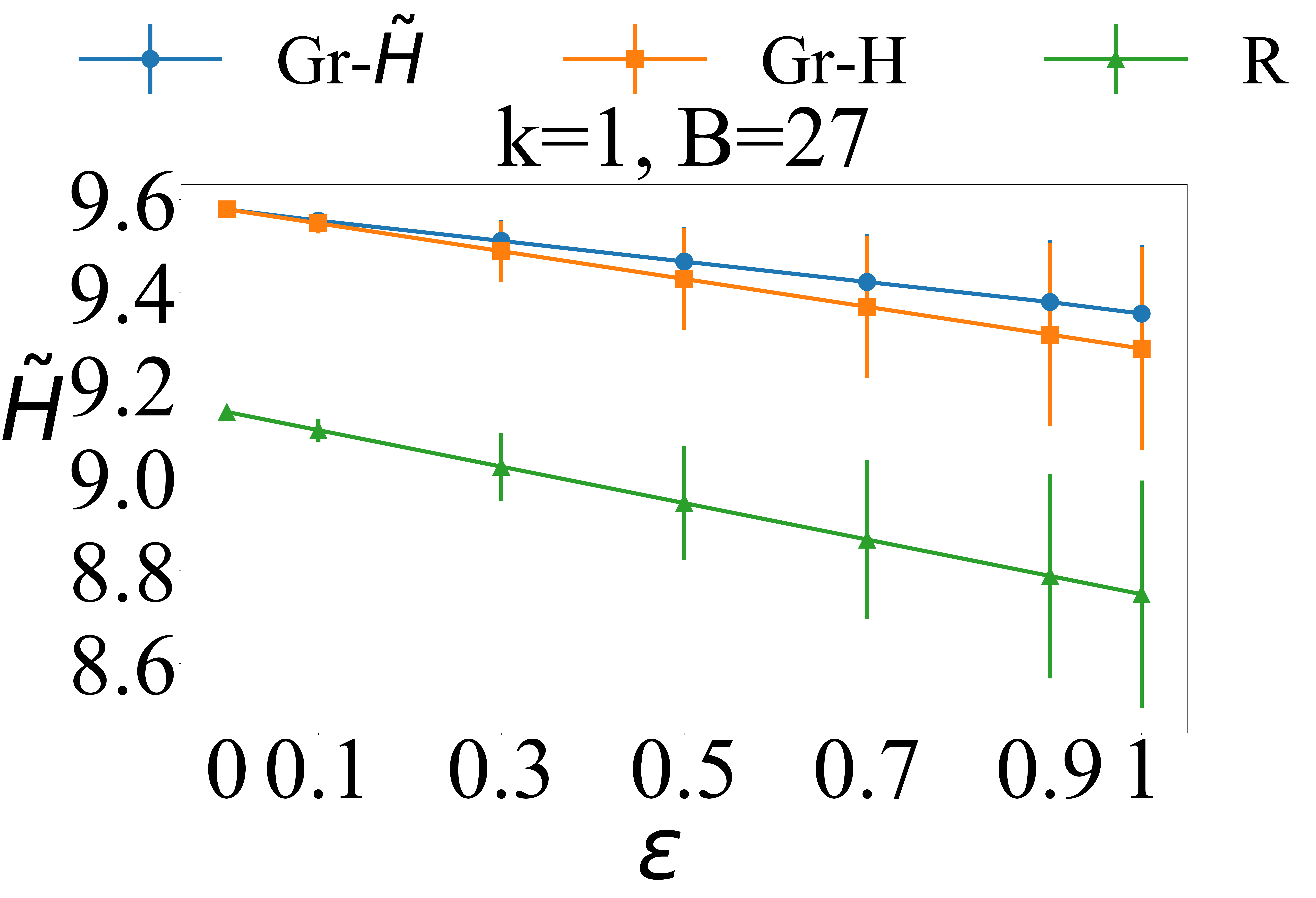}
        \caption{}\label{fig:spMax1_TS_B}
    \end{subfigure}%
    
\caption{$\tilde{H}$ in MaxG setting vs: (a) $k$, (b) $B$, (c) $\varepsilon$.}\label{fig:spMaxTSall}
\vspace{+3mm}
  \end{figure}

\newpage
\subsection{Sensor Placement and Influence Maximization with Approximately Diminishing Returns $k$-Submodular Functions}

We constructed $\varepsilon$-ADR functions for each setting, as follows. 

In the AG setting, we ran $k$-Greedy-TS on $f$ and set $\Delta_{e,i}F(\pmb{x_f})=(1+\varepsilon)\Delta_{e,i}f(\pmb{x_f})$ for its solution $\pmb{x_f}$. 
For any other $\pmb{x}\neq \pmb{x_f}$, we selected  $\xi(\pmb{x})$ uniformly at random in $[1-\varepsilon, 1]$ and set $\Delta_{e,i}F(\pmb{x_f})=\xi(\pmb{x})\cdot\Delta_{e,i} f(\pmb{x_f})$, for each $(e,i)$. Then, we summed  up   $\max_{(e,i)}\Delta_{e,i}F(\pmb{x})$ in each iteration to obtain $F(\pmb{x})$. 



In the MaxG setting, we set $\Delta_{e, i}F(\pmb{x})=\xi(\pmb{x},e)\cdot \Delta_{e,i}f(\pmb{x})$, where $\xi(\pmb{x},e)=\max(\xi(e), \max_{x\in supp(\pmb{x})}\xi({x}))$ and $\xi(x)\in [1-\varepsilon, 1]$. We then summed up $\max_{(e,i)}\Delta_{e, i}F(\pmb{x})$ in each iteration to obtain $F(\pmb{x})$. 
Similarly, in the MeanG setting, we used  $\xi(\pmb{x}, e)=\frac{\xi(e)+\sum_{x\in supp(\pmb{x})}\xi(x)}{|supp(\pmb{x})|+1}$, where $\xi(x)\in[1-\varepsilon, 1]$. 

\subsubsection{Non Submodularity}
To see that the function $F$ constructed in the AG setting is not k-submodular, consider $k=1$, $\xi(\{v\}, u)=1$, $\xi(\mathbf{0}, u)=1-\varepsilon$, as well as $f(\{u, v\})-f(\{v\})=f(\{u\})$ for some $u, v$. It follows that $F(\{u\})-F(\mathbf{0})<F(\{u,v\}-F(\{v\}) \Leftrightarrow \Delta_{u}F(\mathbf{0})=(1-\varepsilon)\Delta_{u}f(\mathbf{0})<f(\{u,v\}-f(\{v\})=\Delta_{u}F(\{v\})=F(\{u,v\}-F(\{v\})$ for $\varepsilon>0$.

To see that  
the $F$ function constructed in the MaxG setting is not $k$-submodular, 
consider $k=1$ and two elements $\{u,v\}$, 
$\xi(u)=1-\varepsilon$, $\xi(v)=1$, and  $f(\{u,v\})-f(\{v\})=f(\{u\})$ for some $u, v$. 
It follows that $F(\{u\})-F(\mathbf{0})<F(\{u,v\})-F(\{v\}) \Leftrightarrow \Delta_{u}F(\mathbf{0})=(1-\varepsilon)\Delta_{u}f(\mathbf{0}) < f(\{u, v\})-f(\{v\})=\Delta_{u}f(\{v\})=F(\{u,v\})-F(\{v\})$ for $\varepsilon>0$. 

To see that 
the $F$ function constructed in the MeanG setting is not $k$-submodular, consider $k=1$, $\xi(u)=1-\varepsilon$, and $\xi(v)=1$, as well as 
$\Delta_{u}f(\{v\})=\Delta_{u}f(\mathbf{0})$ for some $u, v$. It follows that $F(\{u\})-F(\mathbf{0})<F(\{u, v\}-F(\{v\}) \Leftrightarrow \Delta_{u}F(\mathbf{0})=(1-\varepsilon)\Delta_{u}f(\mathbf{0})<\frac{2-\varepsilon}{2}\Delta_{u}f(\mathbf{0})=\frac{2-\varepsilon}{2}\Delta_{u}f(\{v\})=\Delta_{u}F(\{v\})=F(\{u,v\})-F(\{v\})$ for $\varepsilon>0$.

We first considered sensor placement with the TS constraint, but with an $\varepsilon$-ADR instead of an $\varepsilon$-AS function. $k$-Greedy-TS using $H$ as $f$ is denoted with 
$Gr$-$H$-$ADR$, and $k$-Greedy-TS using $\tilde{H}$ as $F$ is denoted with $Gr$-$\tilde{H}$-$ADR$. The results in Figs.~\ref{fig:spNew_adr_TS}, \ref{fig:spMean_adr_TS}, and~\ref{fig:spMax_adr_TS} are analogous to those for the $\varepsilon$-AS function in Section~\ref{sec:appendix:sp-TS-epsAS} and confirm our analysis in Section~\ref{sec:paper:improvedgreedy}. 
 
 \begin{figure}[!ht]
    \begin{subfigure}{0.17\textwidth}
        \centering
        \captionsetup{justification=centering}
        \includegraphics[width = \linewidth]{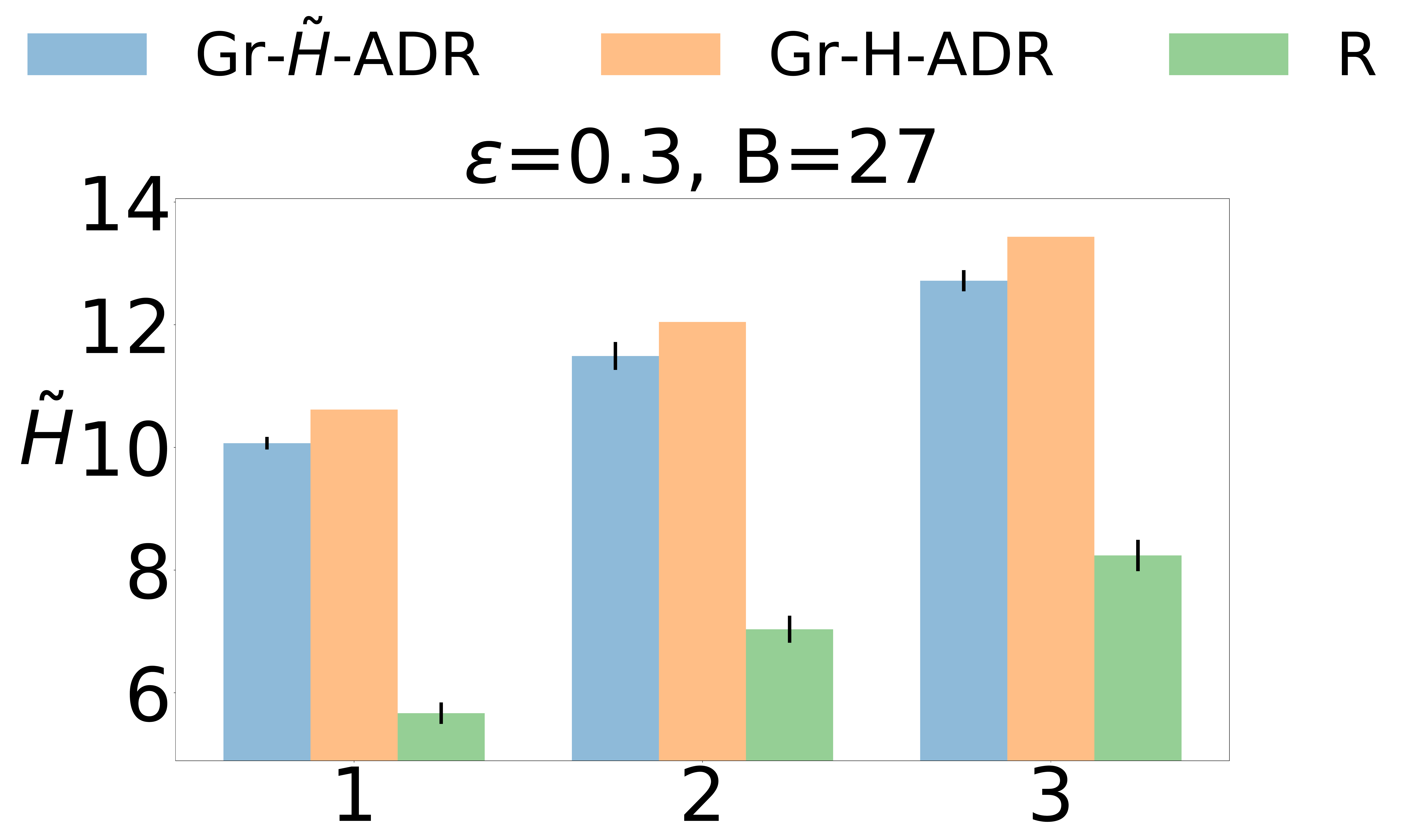}
        \caption{}\label{fig:spNew_adr1_TS_k}
    \end{subfigure}%
    \begin{subfigure}{0.17\textwidth}
        \centering
        \captionsetup{justification=centering}
        \includegraphics[width = \linewidth]{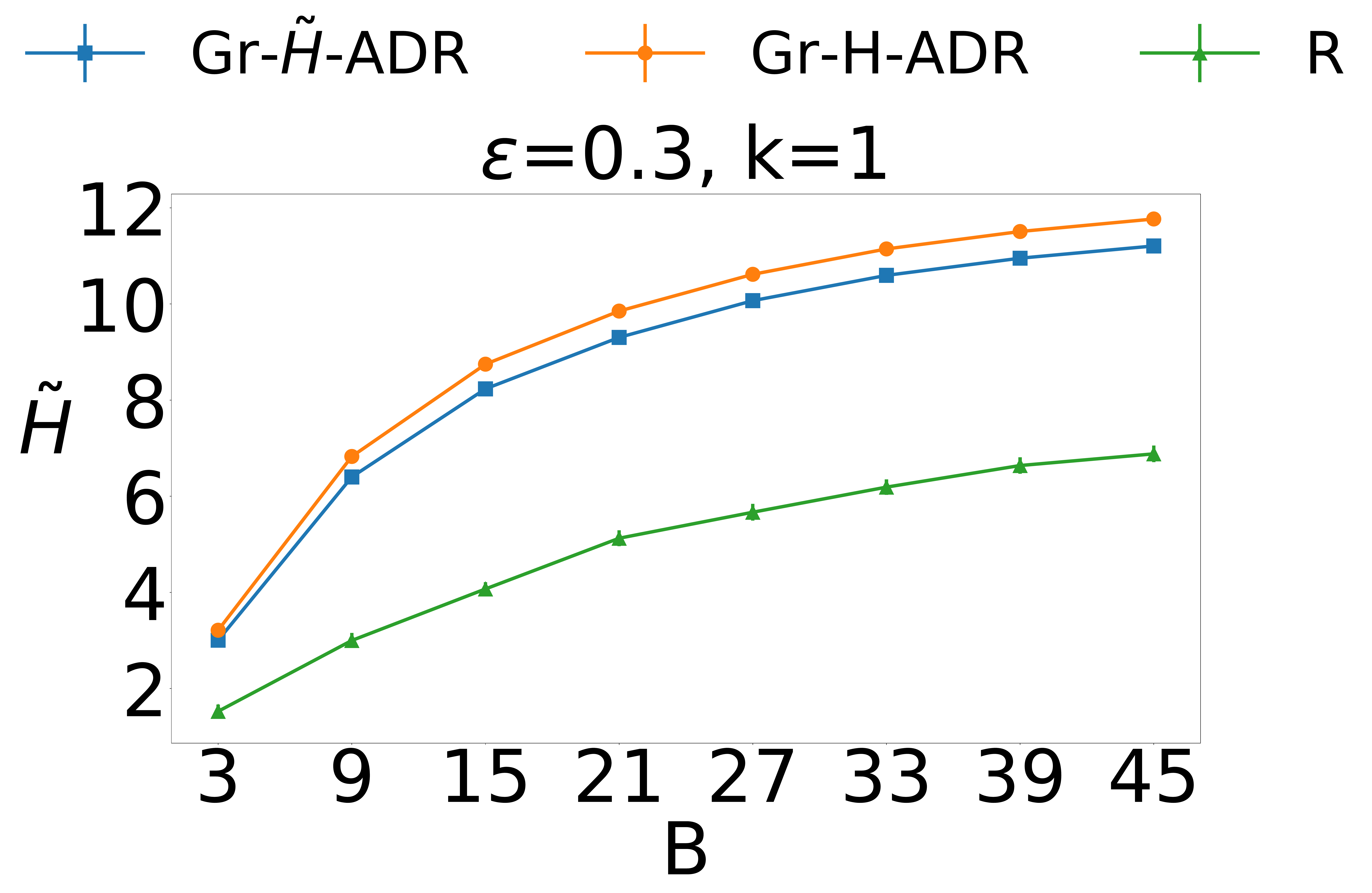}
        \caption{}\label{fig:spNew_adr1_TS_B}
    \end{subfigure}%
    \begin{subfigure}{0.17\textwidth}
        \centering
        \captionsetup{justification=centering}
        \includegraphics[width = \linewidth]{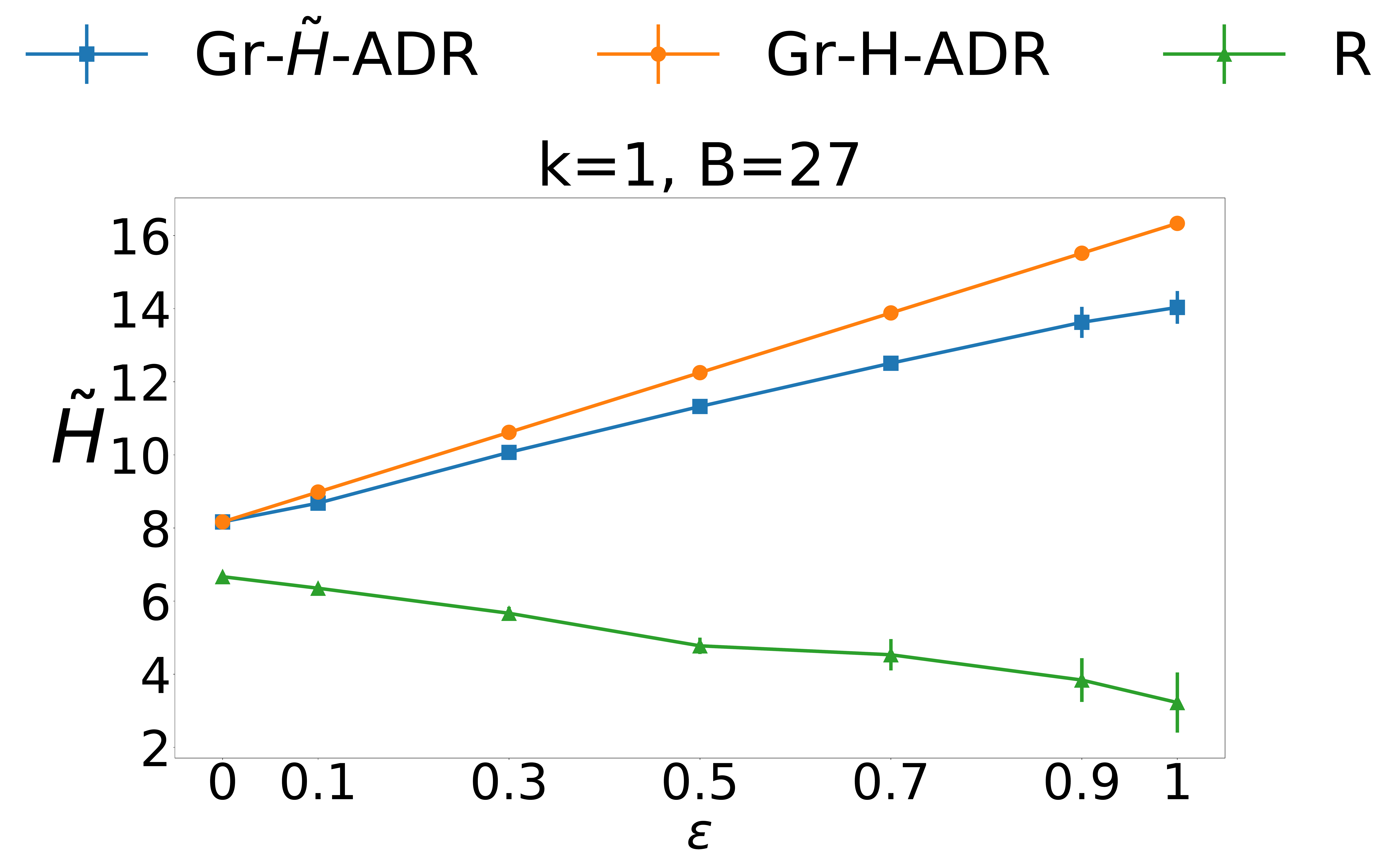}
        \caption{}\label{fig:spNew_adr1_TS_eps}
    \end{subfigure}%
\caption{$\tilde{H}$ in AG setting vs: (a)  $k$, (b) $B$, (c) $\varepsilon$.}\label{fig:spNew_adr_TS}
  \end{figure}
  
  \begin{figure}[!ht]
    \begin{subfigure}{0.17\textwidth}
        \centering
        \captionsetup{justification=centering}
        \includegraphics[width = \linewidth]{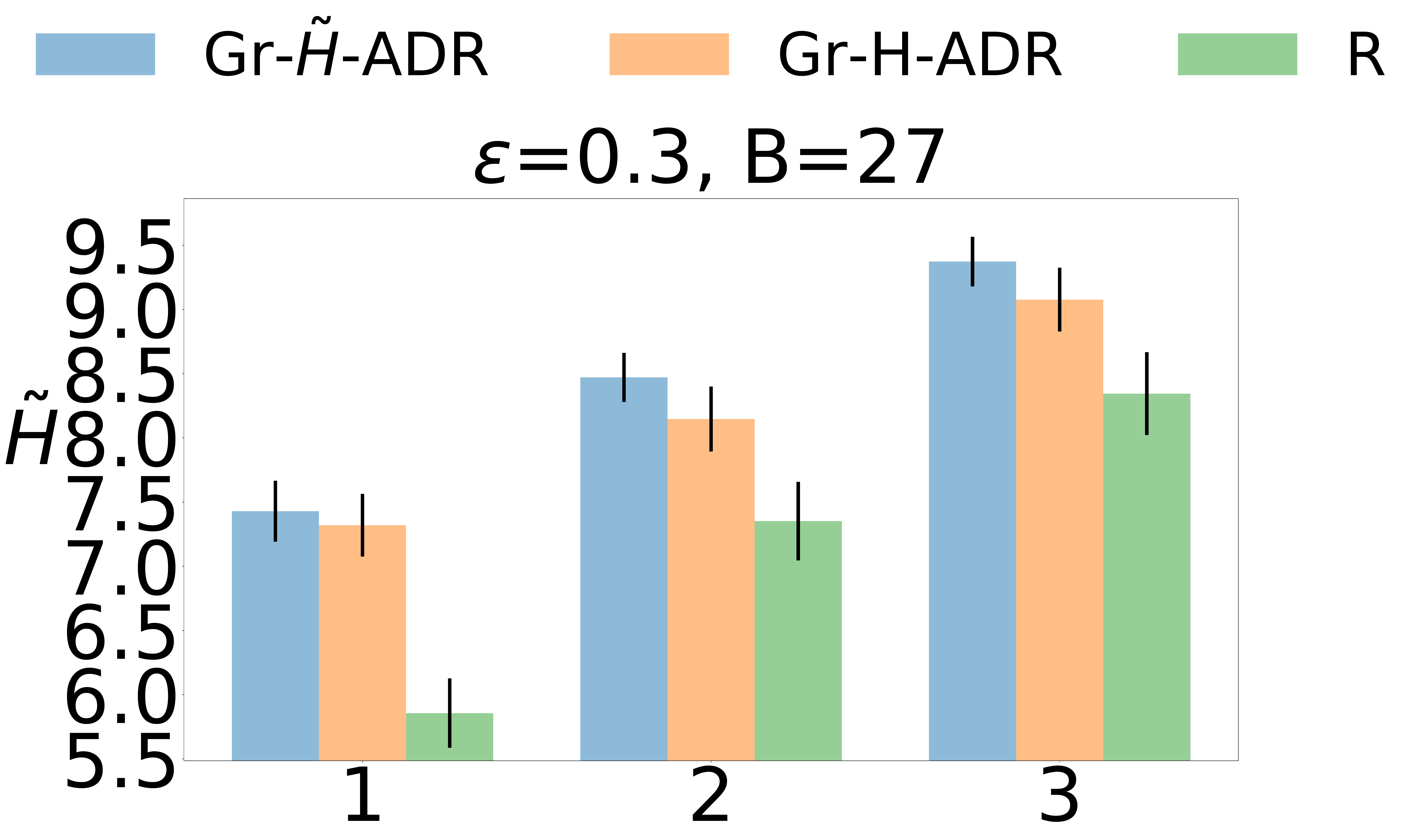}
        \caption{}\label{fig:spNew_adr1_IS_k}
    \end{subfigure}%
    \begin{subfigure}{0.17\textwidth}
        \centering
        \captionsetup{justification=centering}
        \includegraphics[width = \linewidth]{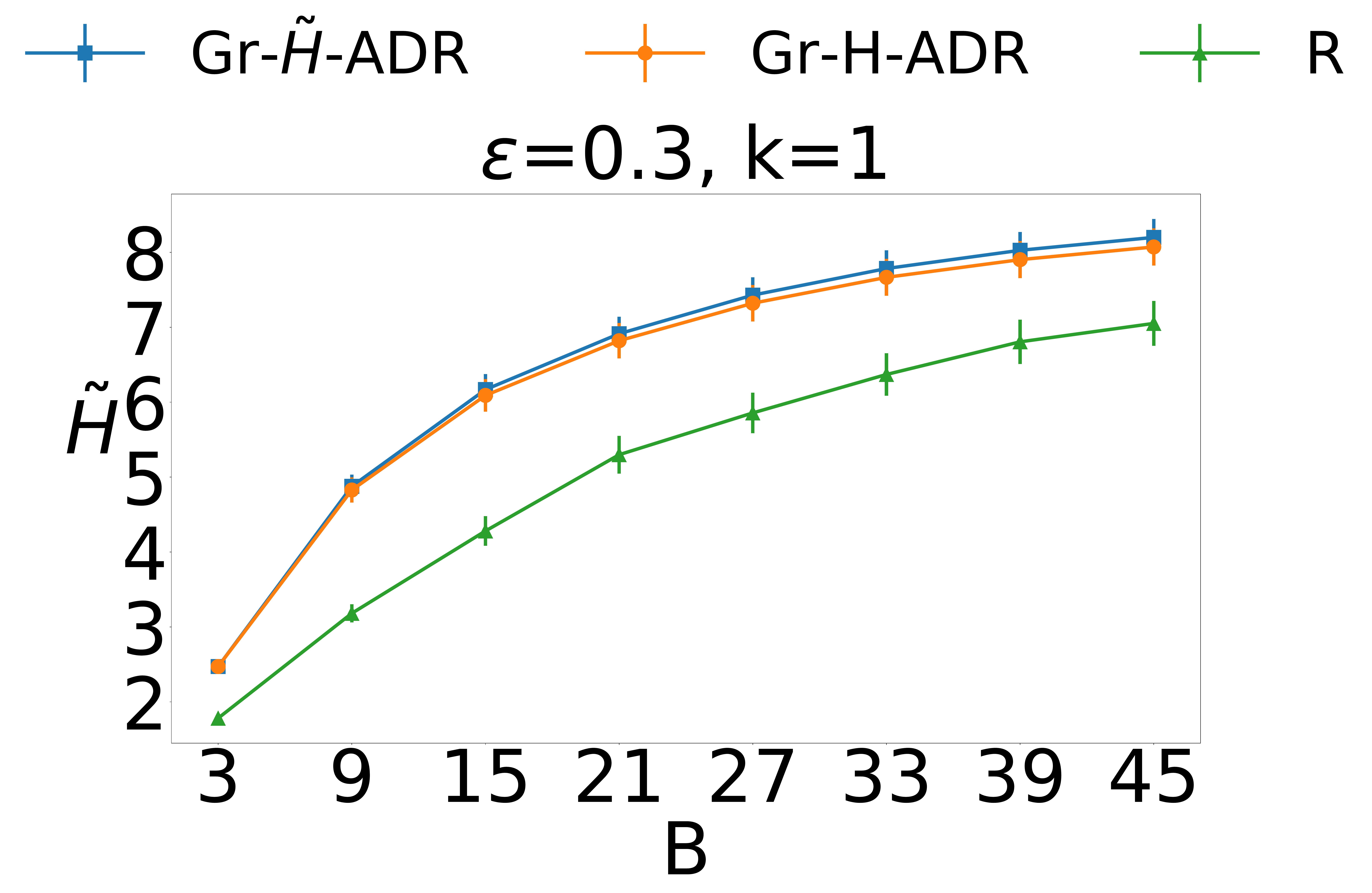}
        \caption{}\label{fig:spNew_adr1_IS_B}
    \end{subfigure}%
    \begin{subfigure}{0.17\textwidth}
        \centering
        \captionsetup{justification=centering}
        \includegraphics[width = \linewidth]{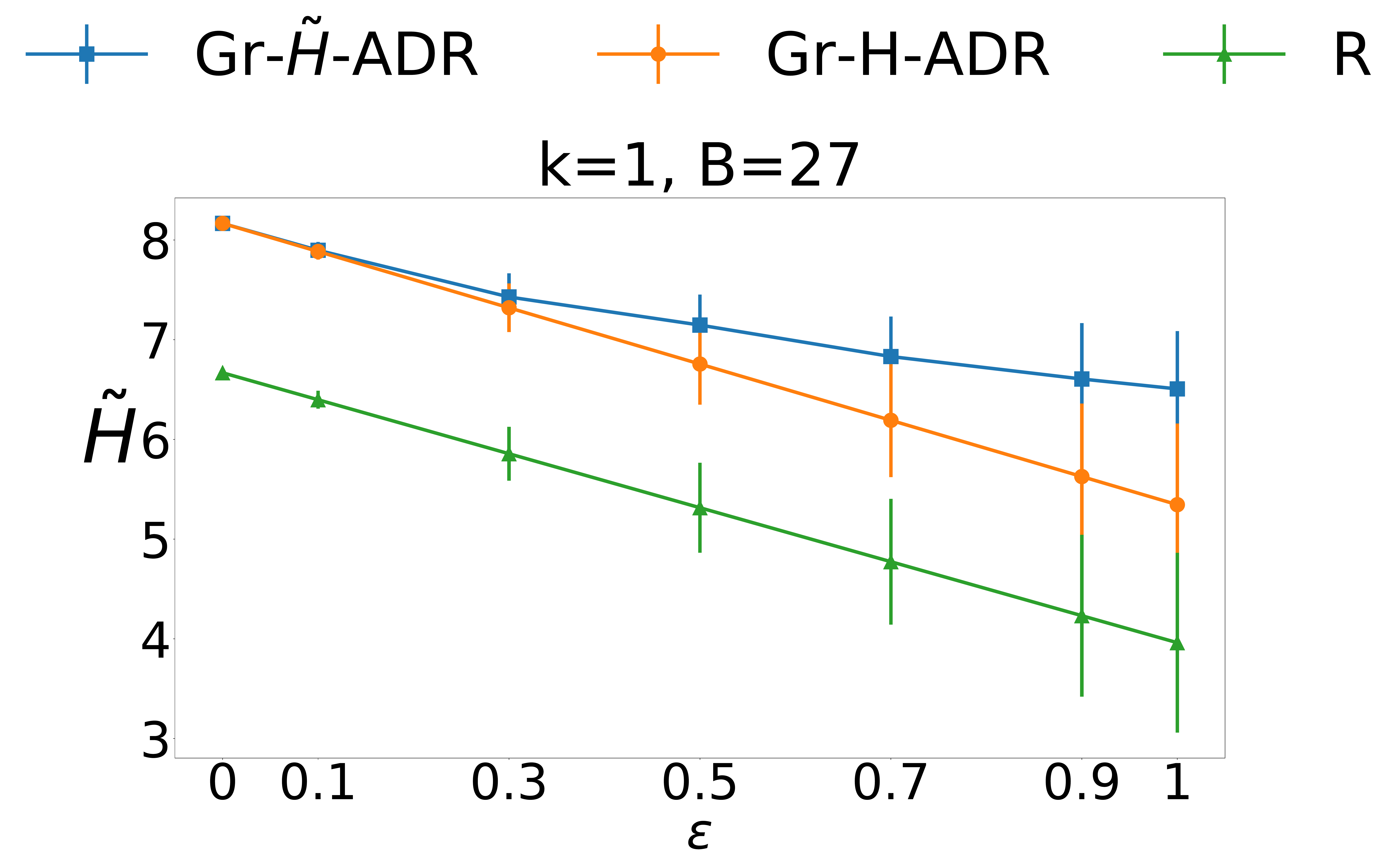}
        \caption{}\label{fig:spNew_adr1_IS_eps}
    \end{subfigure}%
\caption{$\tilde{H}$ in MeanG setting vs: (a) $k$, (b)  $B$, (c) $\varepsilon$.}\label{fig:spMean_adr_TS}
  \end{figure}
  
\begin{figure}[!ht]
    \begin{subfigure}{0.17\textwidth}
        \centering
        \captionsetup{justification=centering}
        \includegraphics[width = \linewidth]{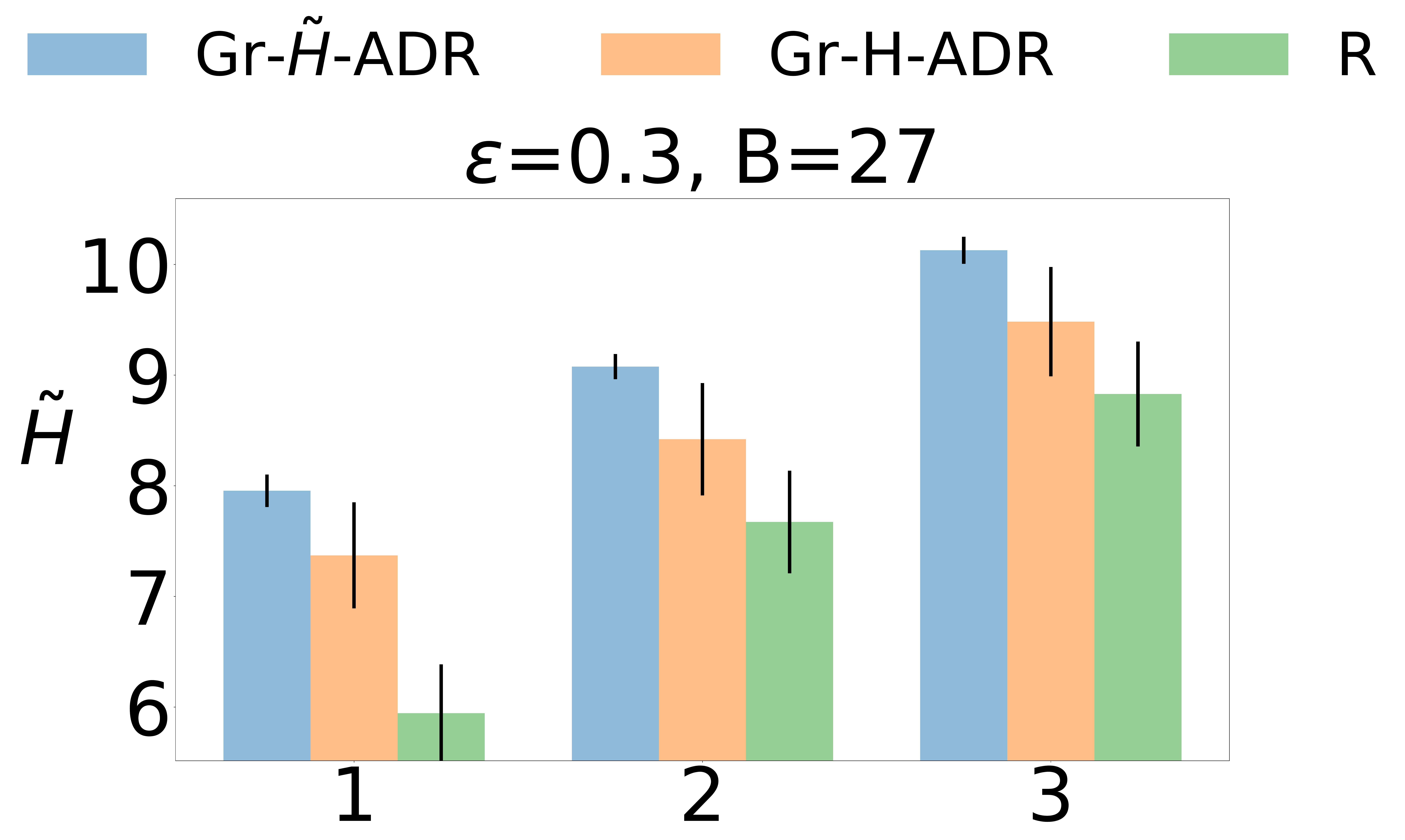}
        \caption{}\label{fig:spMax_adr1_TS_k}
    \end{subfigure}%
    \begin{subfigure}{0.17\textwidth}
        \centering
        \captionsetup{justification=centering}
        \includegraphics[width = \linewidth]{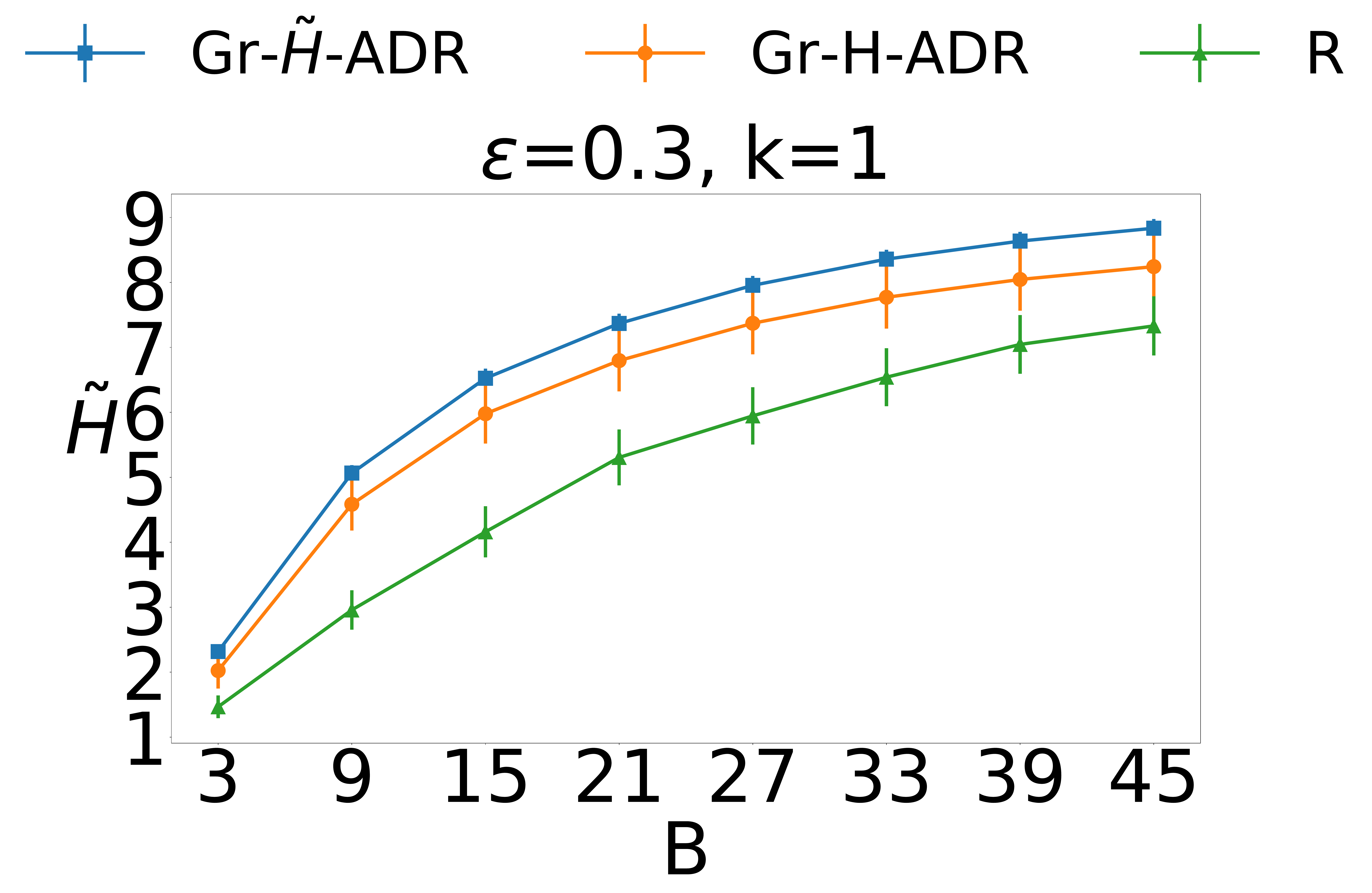}
        \caption{}\label{fig:spMax_adr1_TS_B}
    \end{subfigure}%
    \begin{subfigure}{0.17\textwidth}
        \centering
        \captionsetup{justification=centering}
        \includegraphics[width = \linewidth]{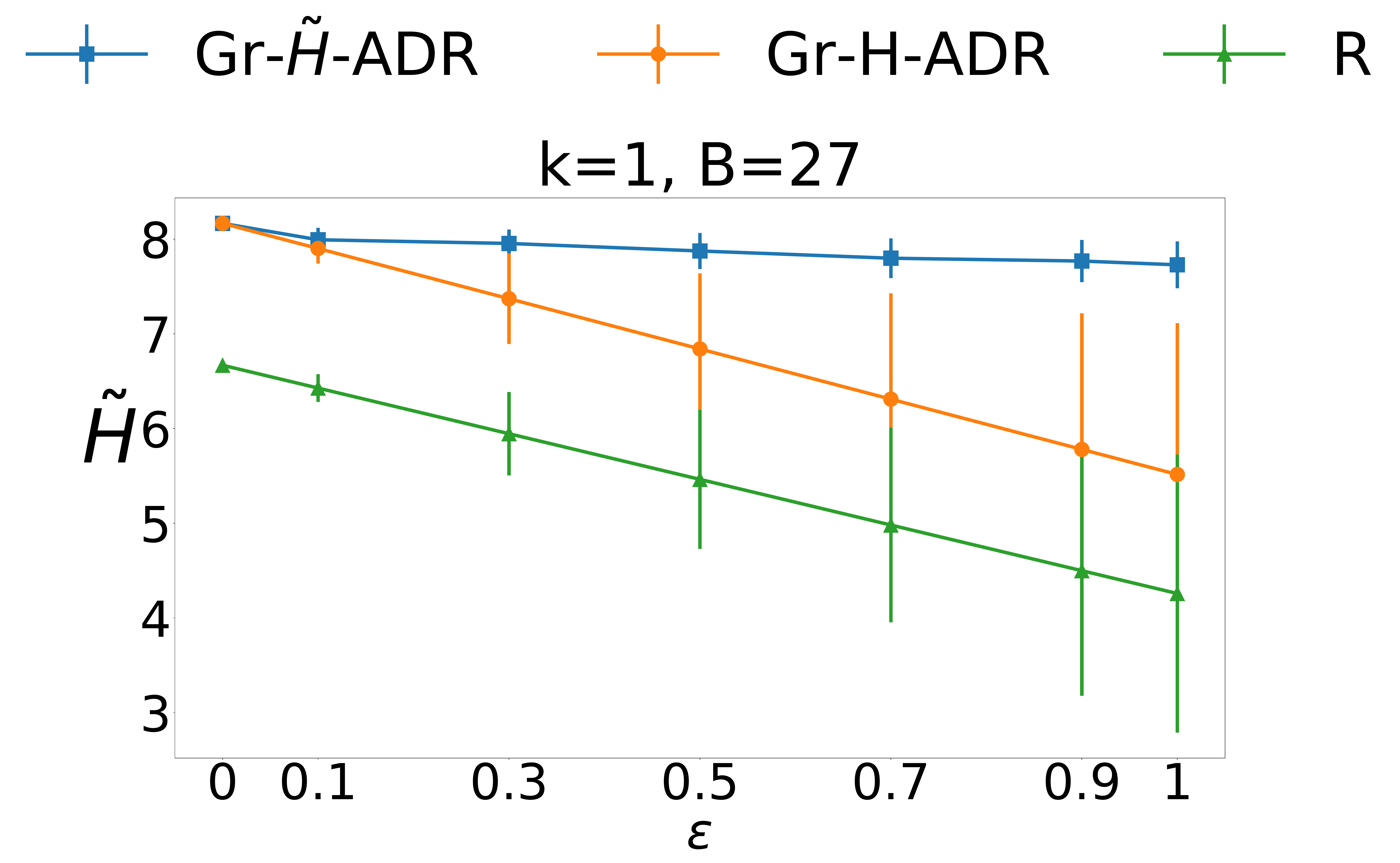}
        \caption{}\label{fig:spMax_adr1_TS_B}
    \end{subfigure}%
    \vspace{-2mm}
\caption{$\tilde{H}$ in MaxG setting vs: (a) $k$, (b)  $B$, (c) $\varepsilon$.}\label{fig:spMax_adr_TS}
  \end{figure}
  

We then considered sensor placement with IS constraints using $k$-Greedy-IS. $k$-Greedy-IS using $H$ as $f$ is denoted with 
$Gr$-$H$-$ADR$, and $k$-Greedy-IS using $\tilde{H}$ as $F$ is denoted with $Gr$-$\tilde{H}$-$ADR$. As expected from the analysis in Section~\ref{sec:paper:improvedgreedy}, the results in Figs.~\ref{fig:spNew_adr_IS}, 
\ref{fig:spMean_adr_IS}, and~\ref{fig:spMax_adr_IS} are similar to those 
in Figs.~\ref{fig:spNew_adr_TS},~\ref{fig:spMean_adr_TS}, and~\ref{fig:spMax_adr_TS}. 

  \begin{figure}[!ht]
    \begin{subfigure}{0.17\textwidth}
        \centering
        \captionsetup{justification=centering}
        \includegraphics[width = \linewidth]{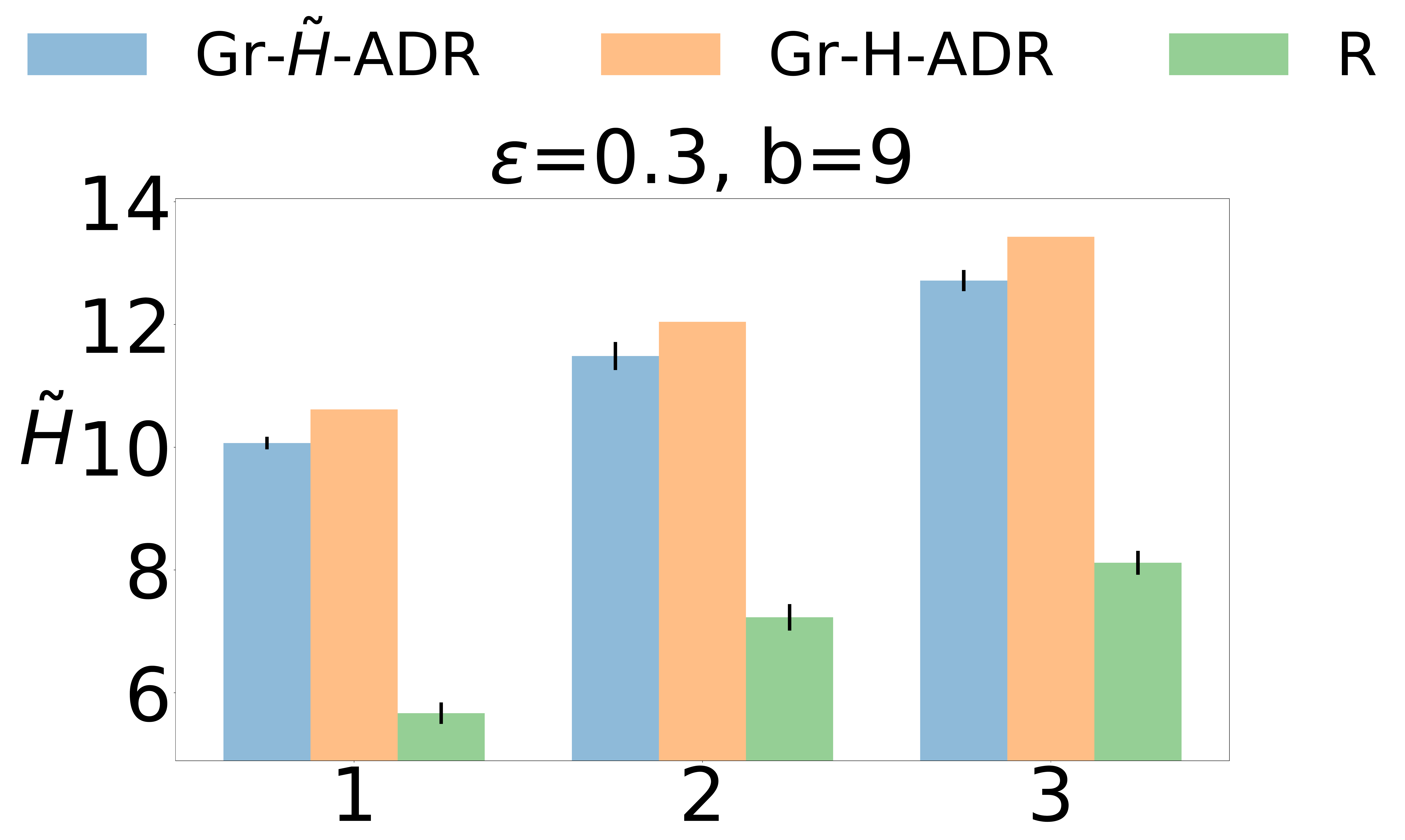}
        \caption{}\label{fig:spNew_adr1_IS_k}
    \end{subfigure}%
    \begin{subfigure}{0.17\textwidth}
        \centering
        \captionsetup{justification=centering}
        \includegraphics[width = \linewidth]{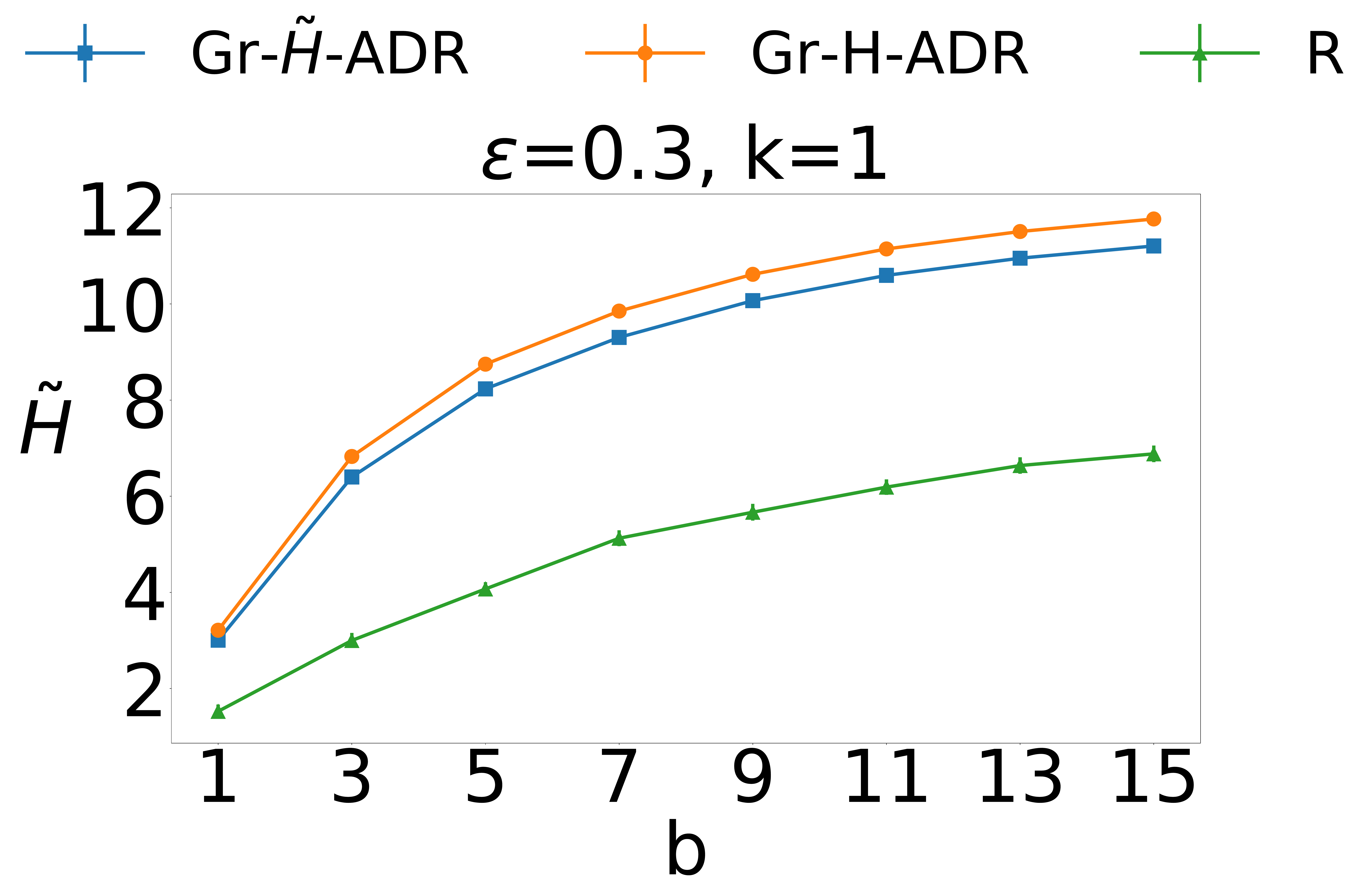}
        \caption{}\label{fig:spNew_adr1_IS_B}
    \end{subfigure}%
    \begin{subfigure}{0.17\textwidth}
        \centering
        \captionsetup{justification=centering}
        \includegraphics[width = \linewidth]{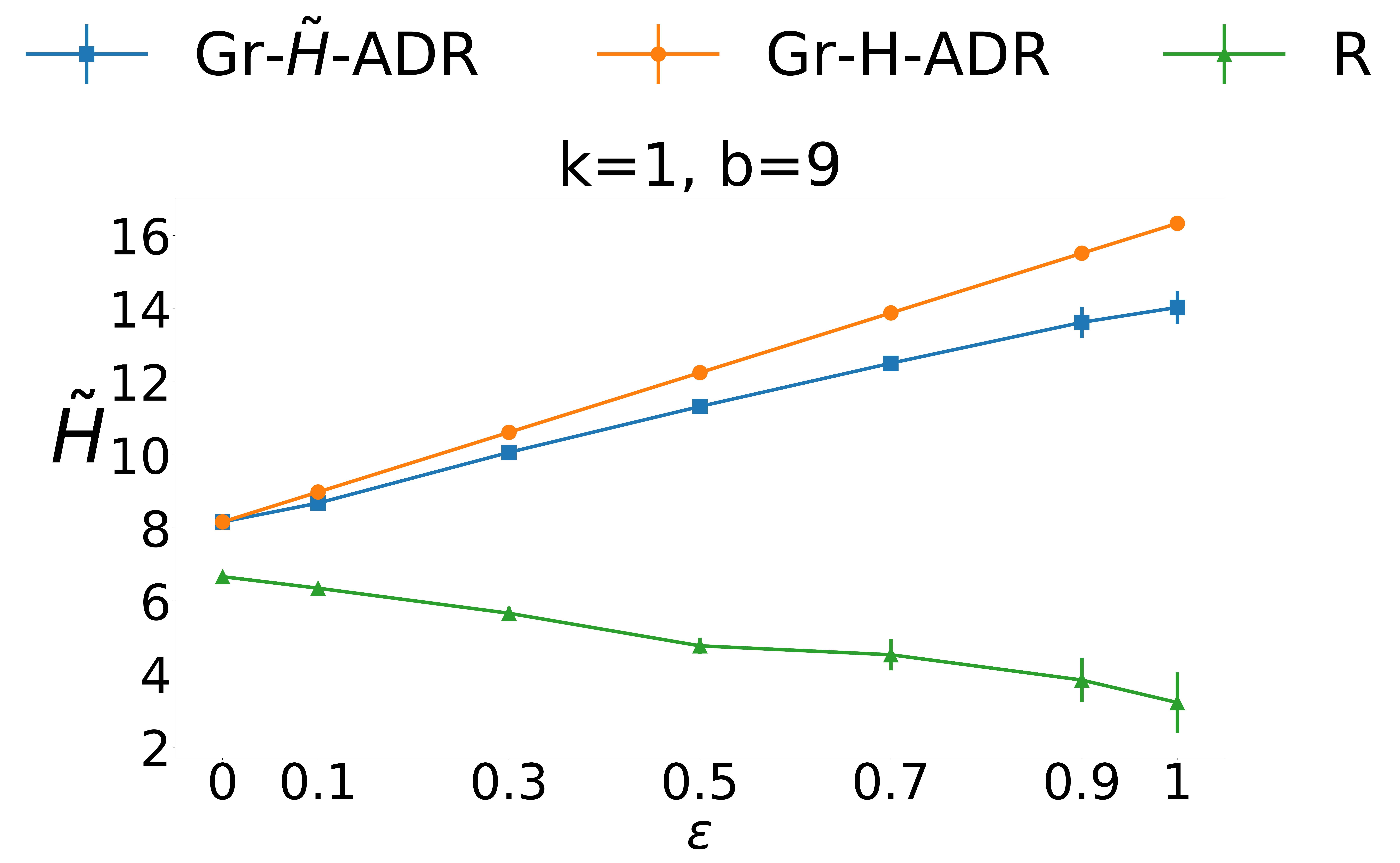}
        \caption{}\label{fig:spNew_adr1_IS_eps}
    \end{subfigure}%
\caption{$\tilde{H}$ in AG setting vs: (a) $k$, (b)  $b$, (c) $\varepsilon$.}\label{fig:spNew_adr_IS}
  \end{figure}
  
  \begin{figure}[!ht]
    \begin{subfigure}{0.17\textwidth}
        \centering
        \captionsetup{justification=centering}
        \includegraphics[width = \linewidth]{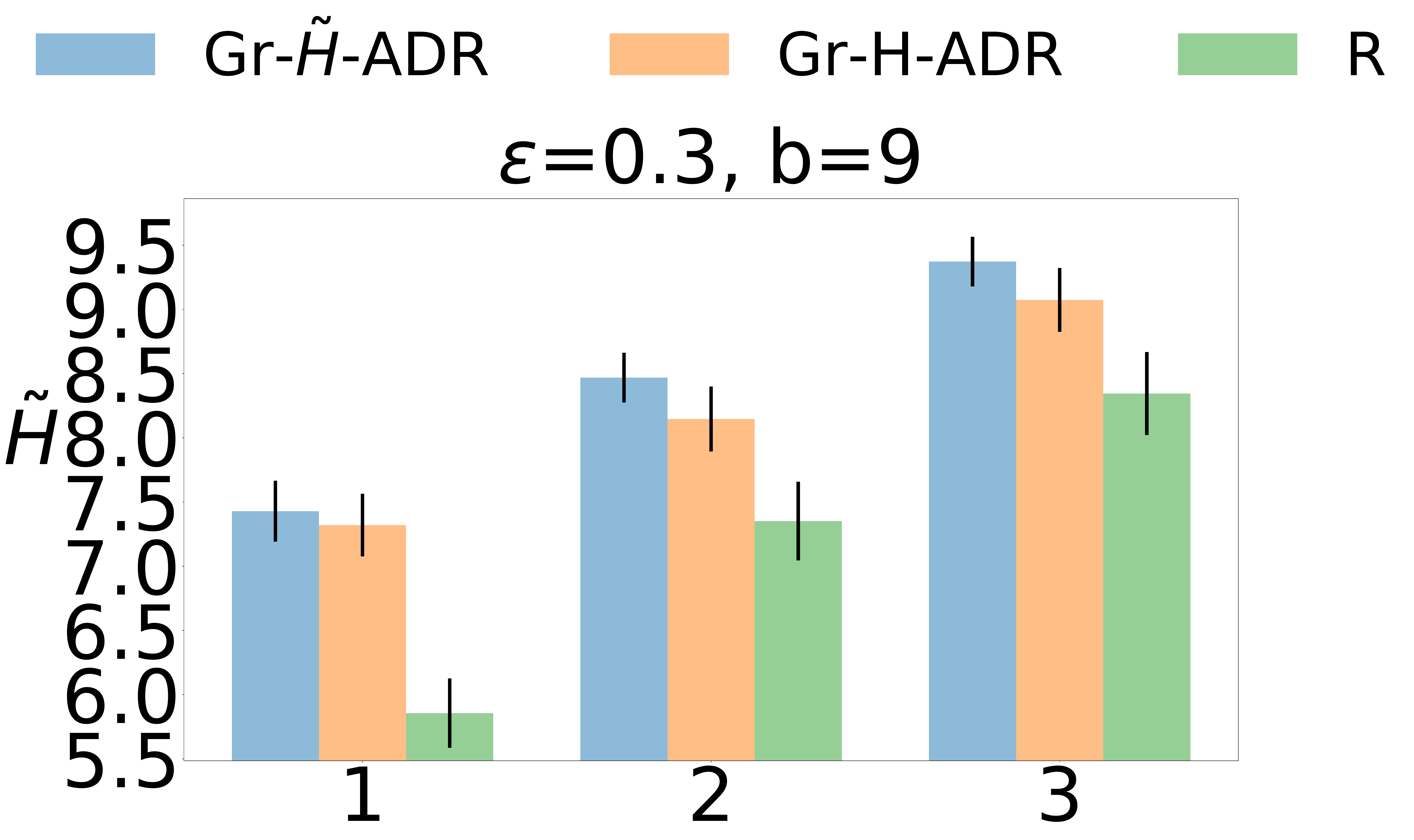}
        \caption{}\label{fig:spMean_adr1_IS_k}
    \end{subfigure}%
    \begin{subfigure}{0.17\textwidth}
        \centering
        \captionsetup{justification=centering}
        \includegraphics[width = \linewidth]{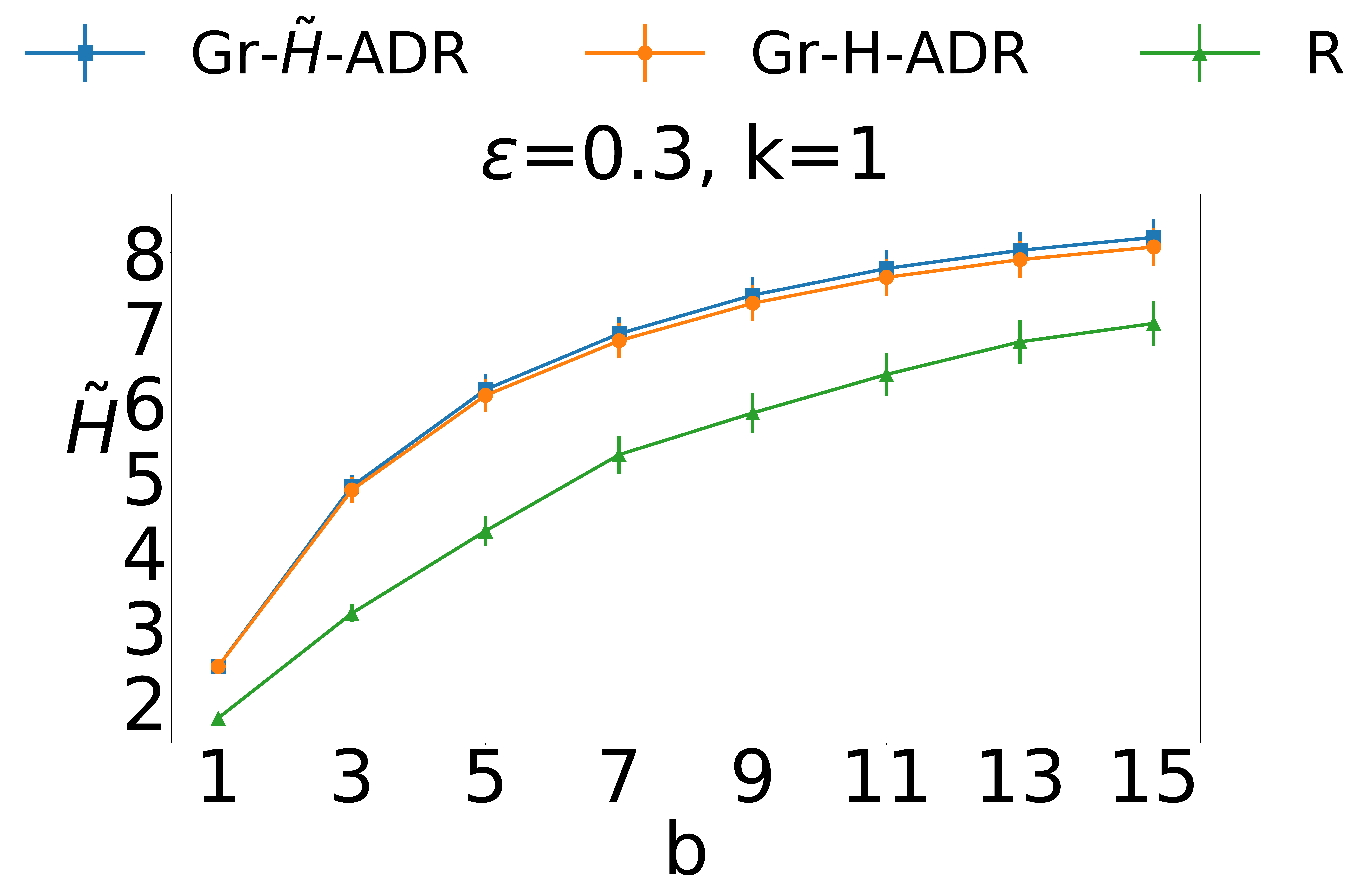}
        \caption{}\label{fig:spMean_adr1_IS_B}
    \end{subfigure}%
    \begin{subfigure}{0.17\textwidth}
        \centering
        \captionsetup{justification=centering}
        \includegraphics[width = \linewidth]{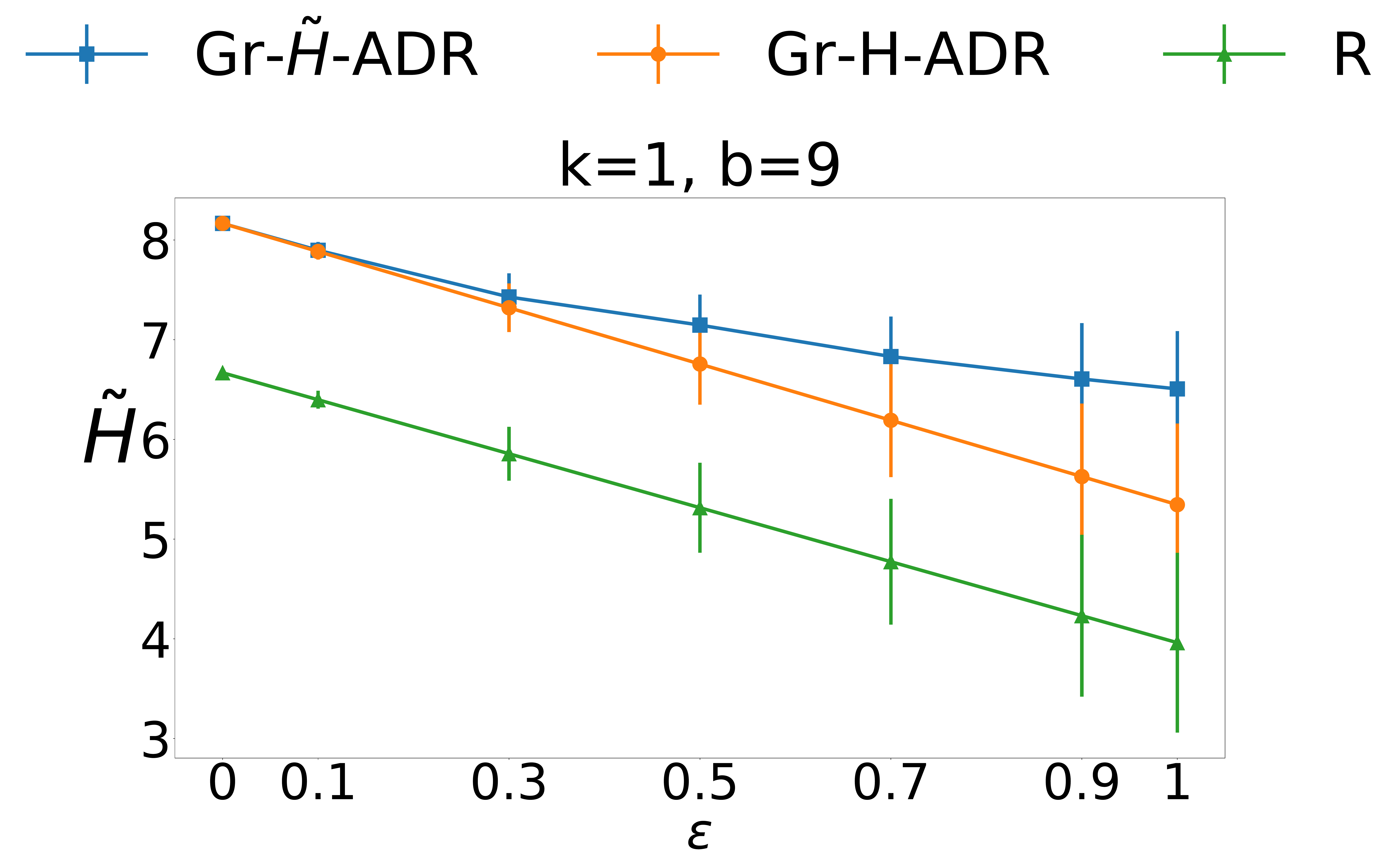}
        \caption{}\label{fig:spMean_adr1_IS_B}
    \end{subfigure}%
\caption{$\tilde{H}$ in MeanG setting vs: (a) $k$, (b) $b$, (c) $\varepsilon$.}\label{fig:spMean_adr_IS}
  \end{figure}

  \begin{figure}[!ht]

    \begin{subfigure}{0.17\textwidth}
        \centering
        \captionsetup{justification=centering}
        \includegraphics[width = \linewidth]{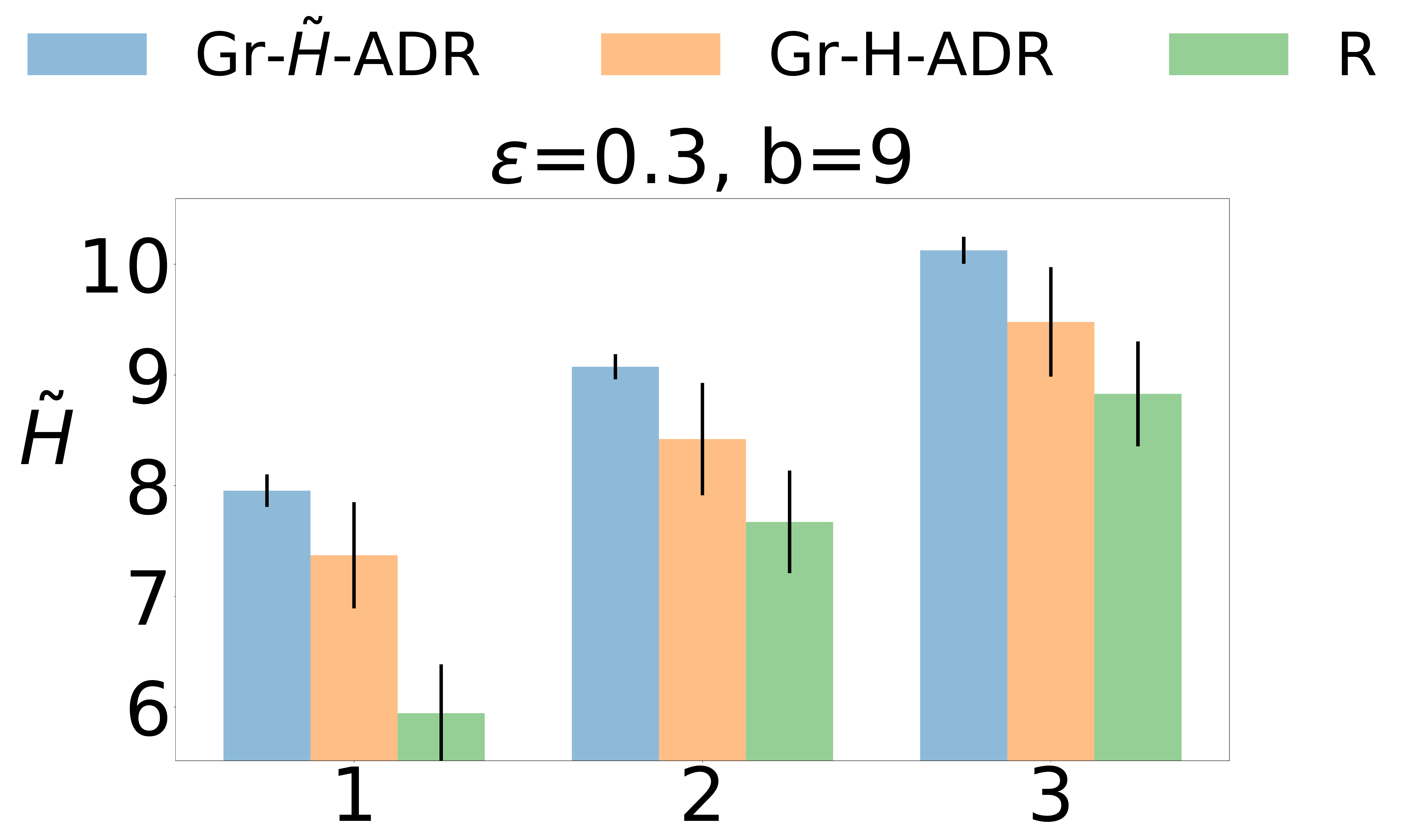}
        \caption{}\label{fig:spMax_adr1_IS_k}
    \end{subfigure}%
    \begin{subfigure}{0.17\textwidth}
        \centering
        \captionsetup{justification=centering}
        \includegraphics[width = \linewidth]{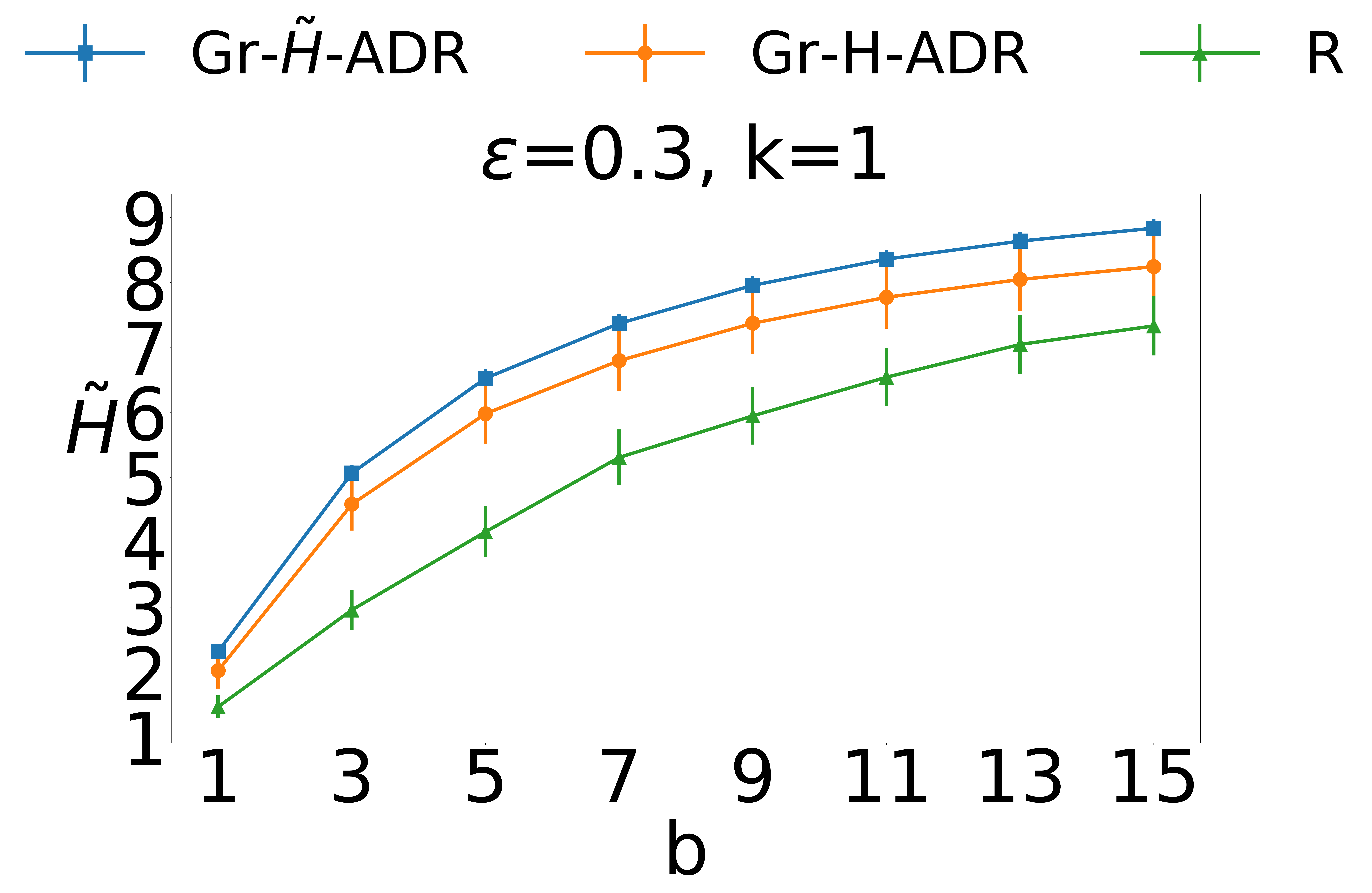}
        \caption{}\label{fig:spMax_adr1_IS_B}
    \end{subfigure}%
    \begin{subfigure}{0.17\textwidth}
        \centering
        \captionsetup{justification=centering}
        \includegraphics[width = \linewidth]{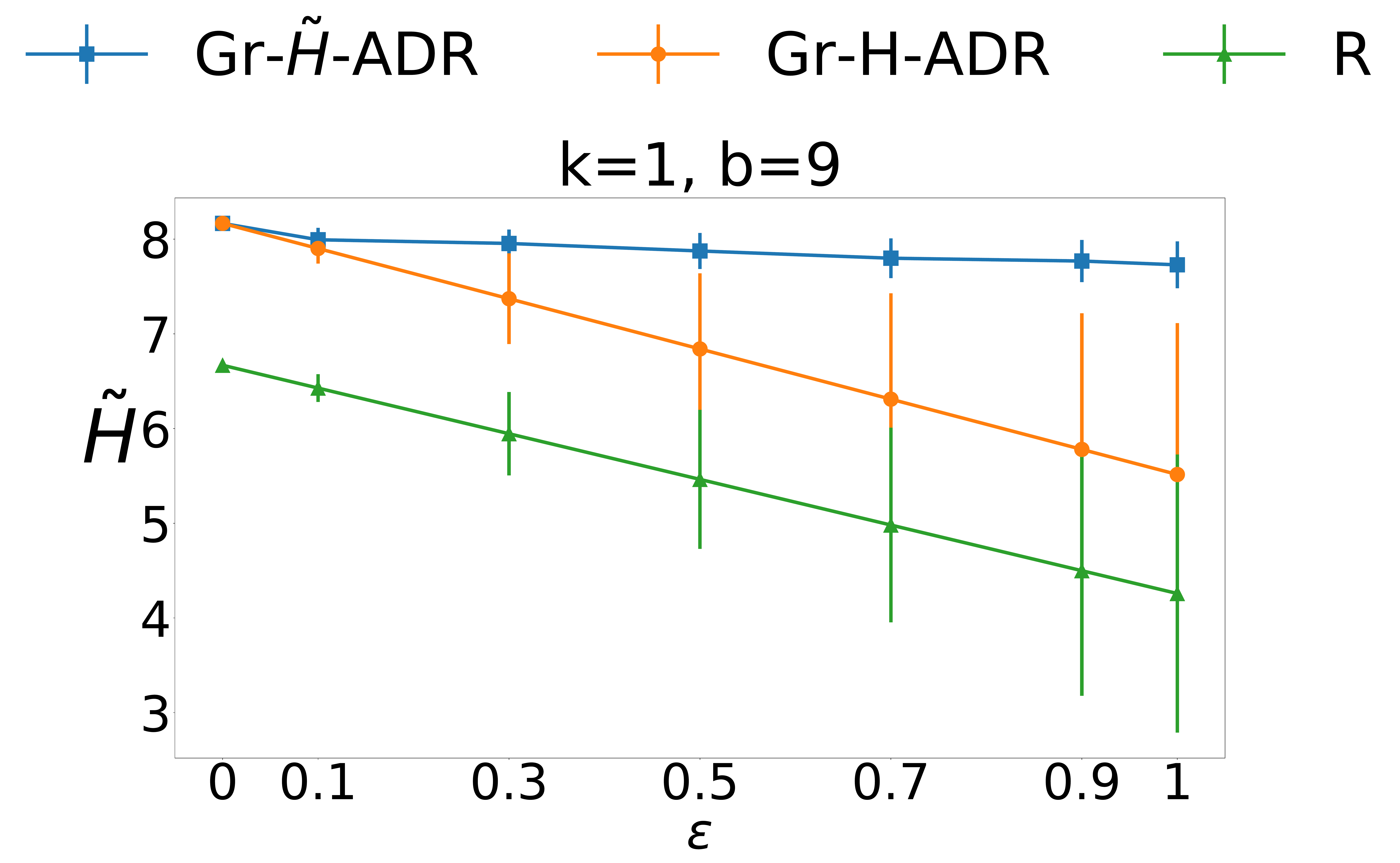}
        \caption{}\label{fig:spMax_adr1_IS_B}
    \end{subfigure}%
    \caption{$\tilde{H}$ in MaxG setting vs: (a) $k$, (b) $b$, (c) $\varepsilon$.}\label{fig:spMax_adr_IS}
  \end{figure}

Last, we considered influence maximization with IS constraints and an  $\varepsilon$-ADR function. 
$k$-Greedy-TS using $I$ as $f$ is denoted with 
$Gr$-$I$-$ADR$, and $k$-Greedy-TS using $\tilde{I}$ as $F$ is denoted with $Gr$-$\tilde{I}$-$ADR$. The obtained 
results are similar to those for the case of $\varepsilon$-AS function (see Section~\ref{exp:IM} and~Appendix~\ref{appendix:E1}).

\begin{figure}[!ht]
    \begin{subfigure}{0.17\textwidth}
        \centering
        \captionsetup{justification=centering}
        \includegraphics[width = \linewidth]{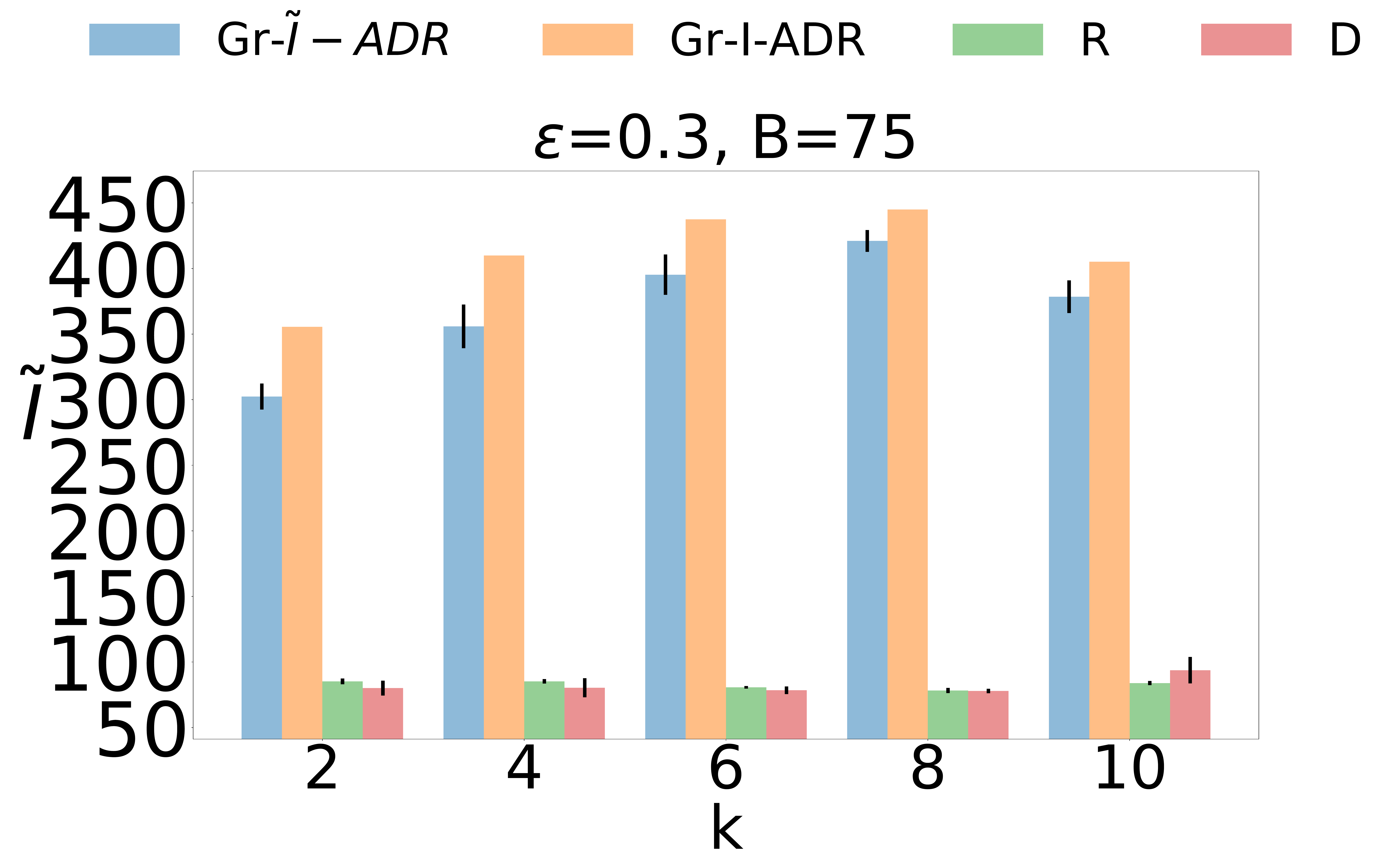}
        \caption{}\label{fig:imAG1}
    \end{subfigure}%
    \begin{subfigure}{0.17\textwidth}
        \centering
        \captionsetup{justification=centering}
        \includegraphics[width = \linewidth]{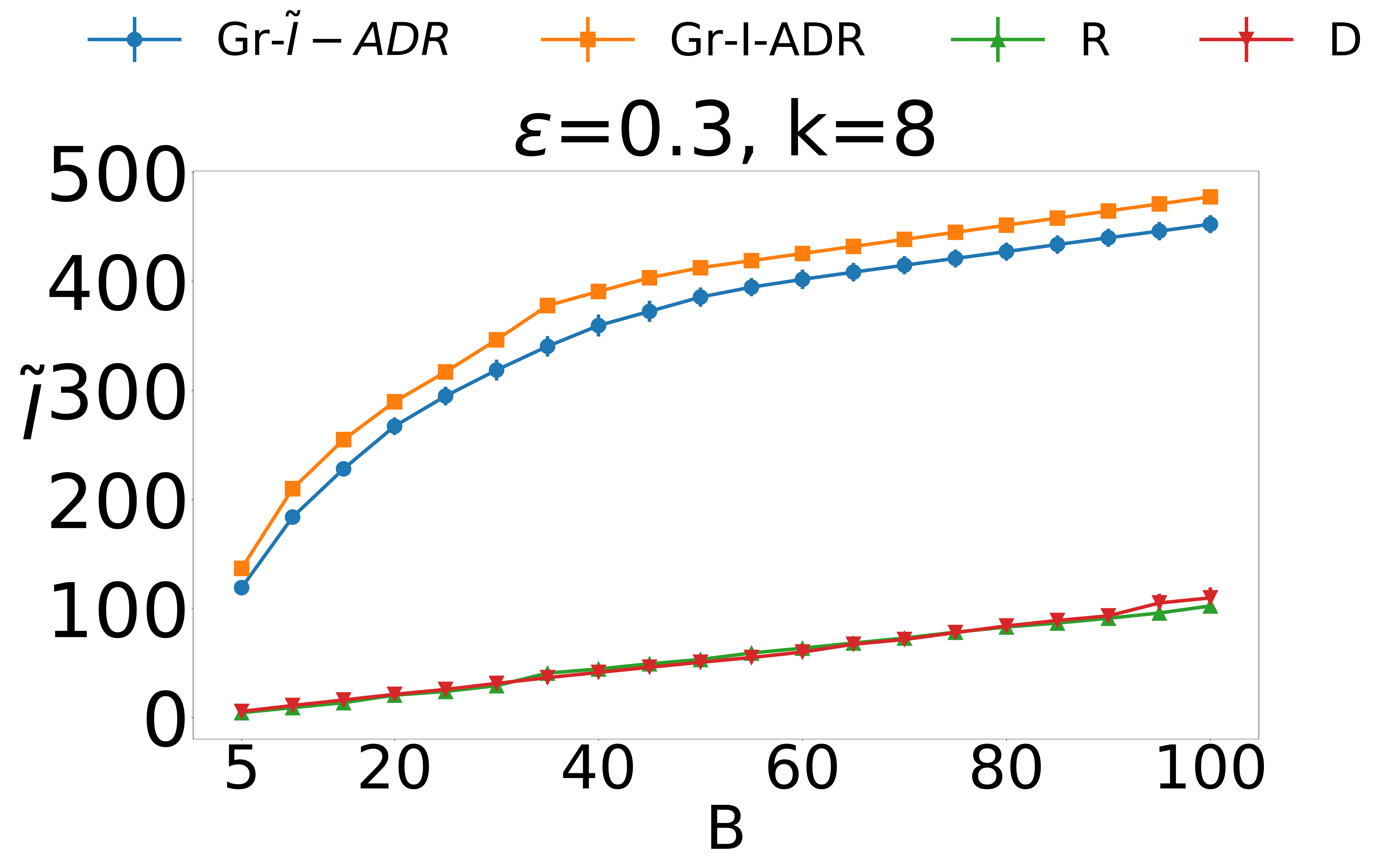}
        \caption{}\label{fig:imAG1}
    \end{subfigure}%
    \begin{subfigure}{0.17\textwidth}
        \centering
        \captionsetup{justification=centering}
        \includegraphics[width = \linewidth]{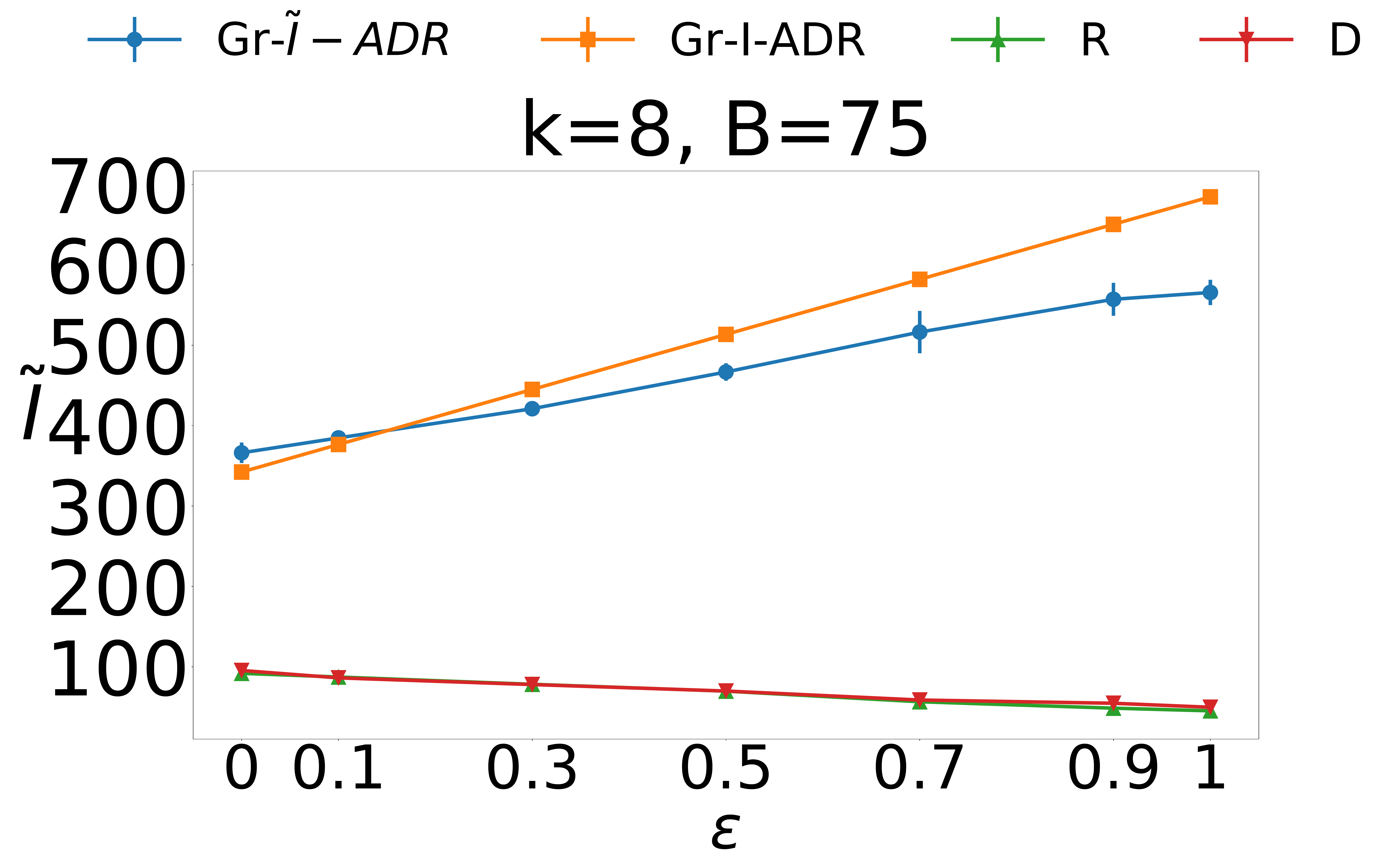}
        \caption{}\label{fig:imAG2}
    \end{subfigure}%
    \caption{$\tilde{I}$ in AG setting vs: (a) k, (b) $B$, (c) $\varepsilon$.}\label{fig:imAGAllADR}
\end{figure}
 
\begin{figure}[!ht]
    \begin{subfigure}{0.17\textwidth}
        \centering
        \captionsetup{justification=centering}
        \includegraphics[width = \linewidth]{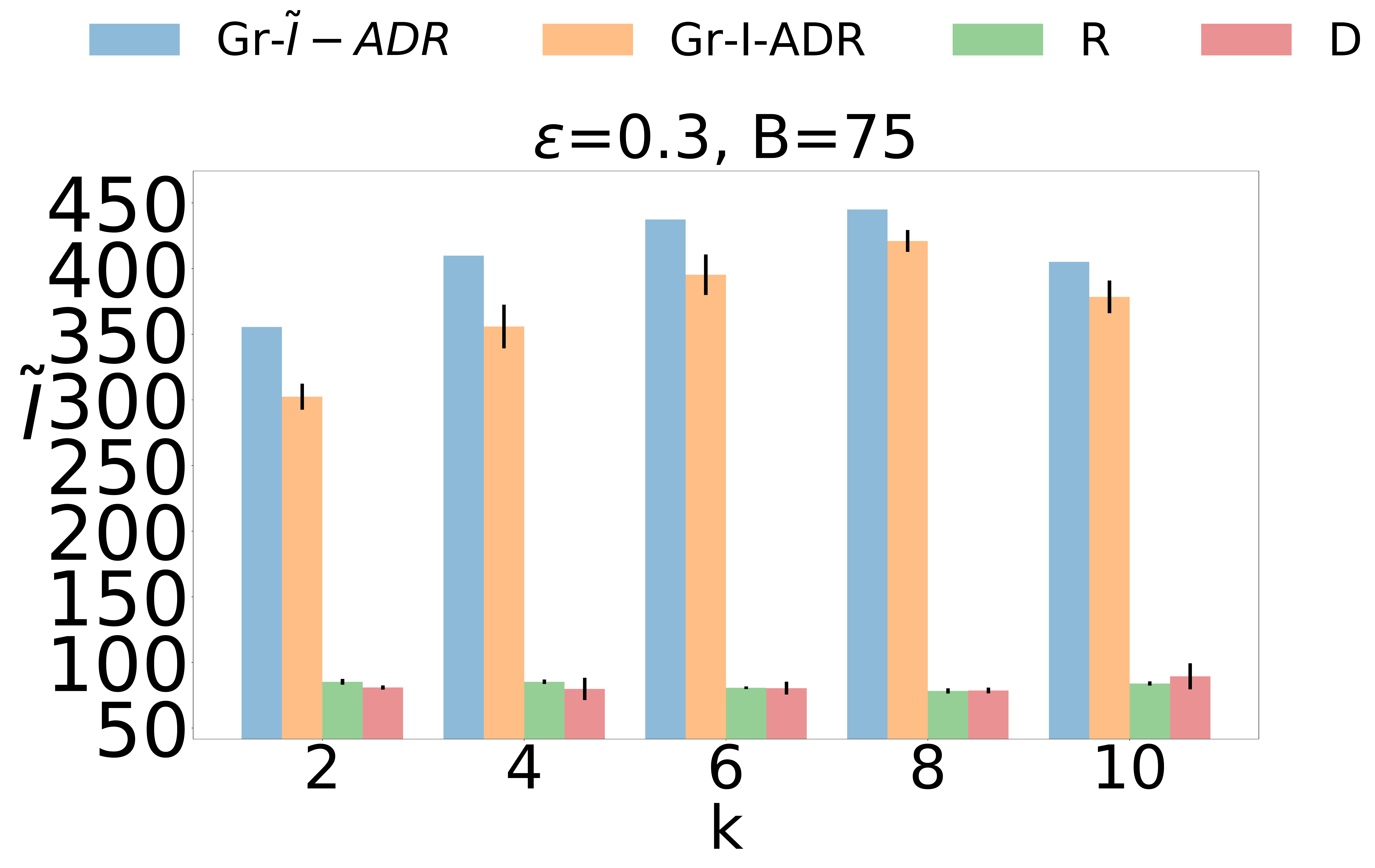}
        \caption{}\label{fig:imAG1}
    \end{subfigure}%
    \begin{subfigure}{0.17\textwidth}
        \centering
        \captionsetup{justification=centering}
        \includegraphics[width = \linewidth]{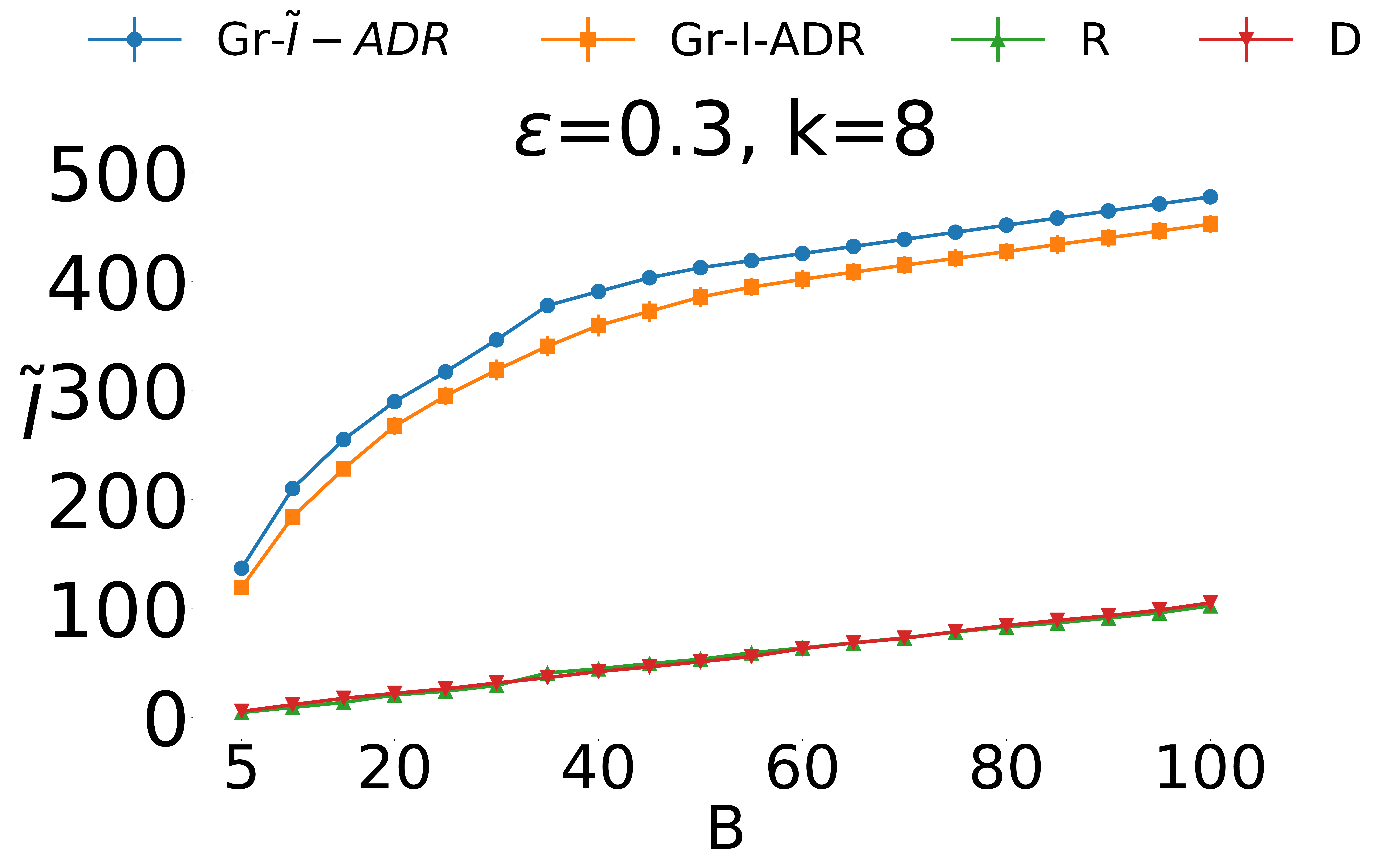}
        \caption{}\label{fig:imMean1}
    \end{subfigure}%
    \begin{subfigure}{0.17\textwidth}
        \centering
        \captionsetup{justification=centering}
        \includegraphics[width = \linewidth]{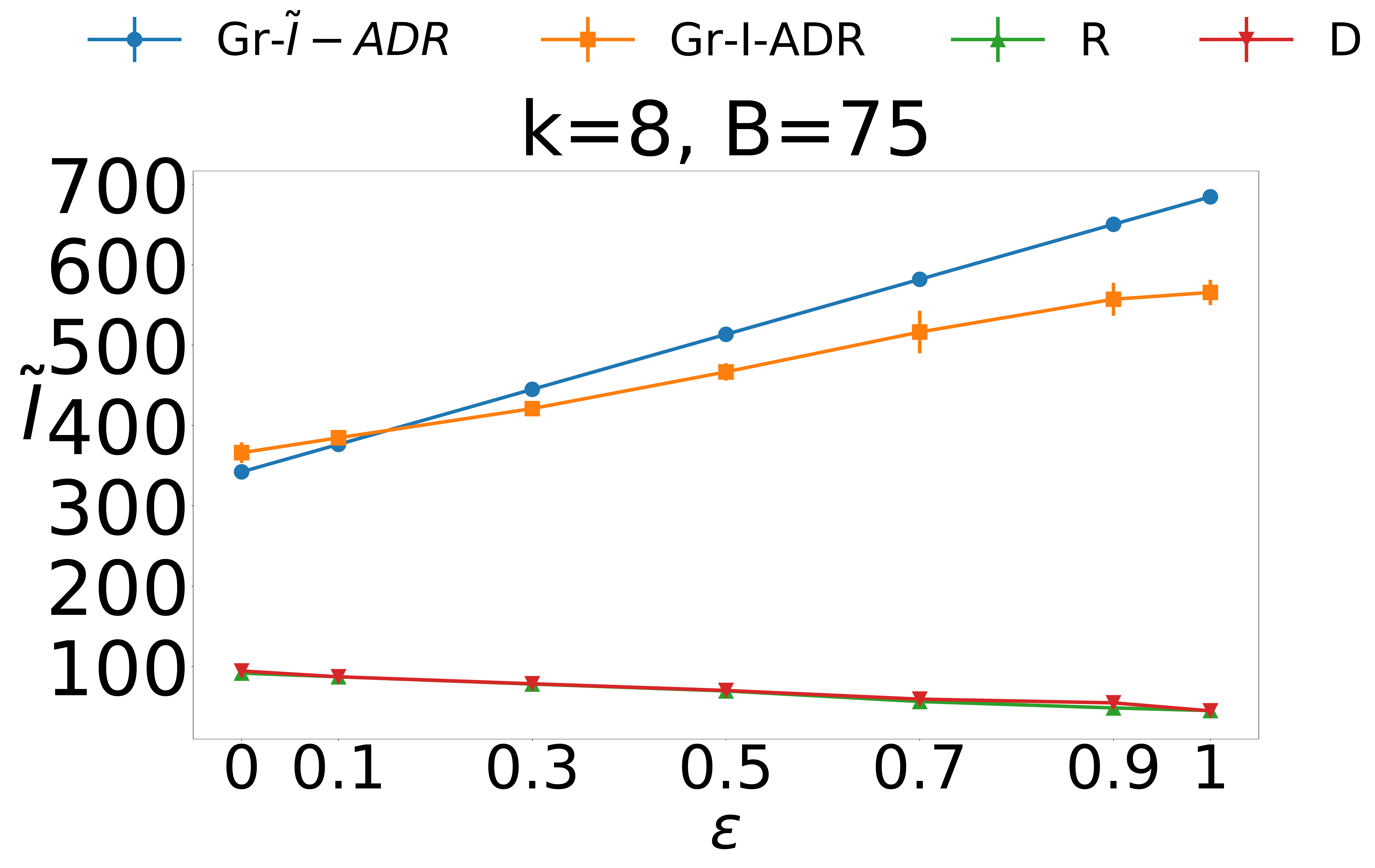}
        \caption{}\label{fig:imMean2}
    \end{subfigure}%
    \caption{$\tilde{I}$ in Mean setting vs:  (a) k, (b) $B$, (c) $\varepsilon$.}\label{fig:imMeanAllADR}
\end{figure}

\begin{figure}[!ht]
    \begin{subfigure}{0.17\textwidth}
        \centering
        \captionsetup{justification=centering}
        \includegraphics[width = \linewidth]{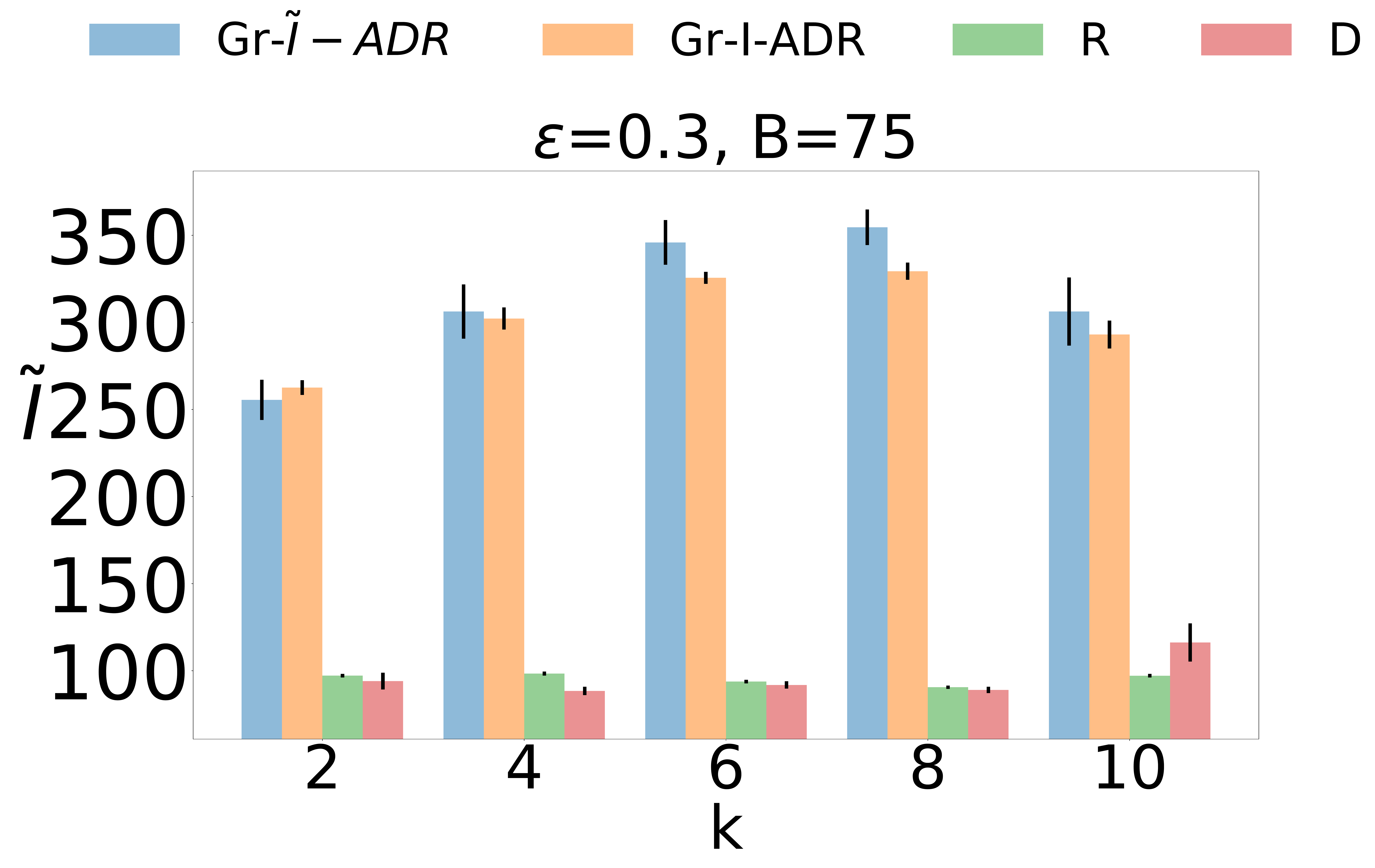}
        \caption{}\label{fig:imAG1}
    \end{subfigure}%
    \begin{subfigure}{0.17\textwidth}
        \centering
        \captionsetup{justification=centering}
        \includegraphics[width = \linewidth]{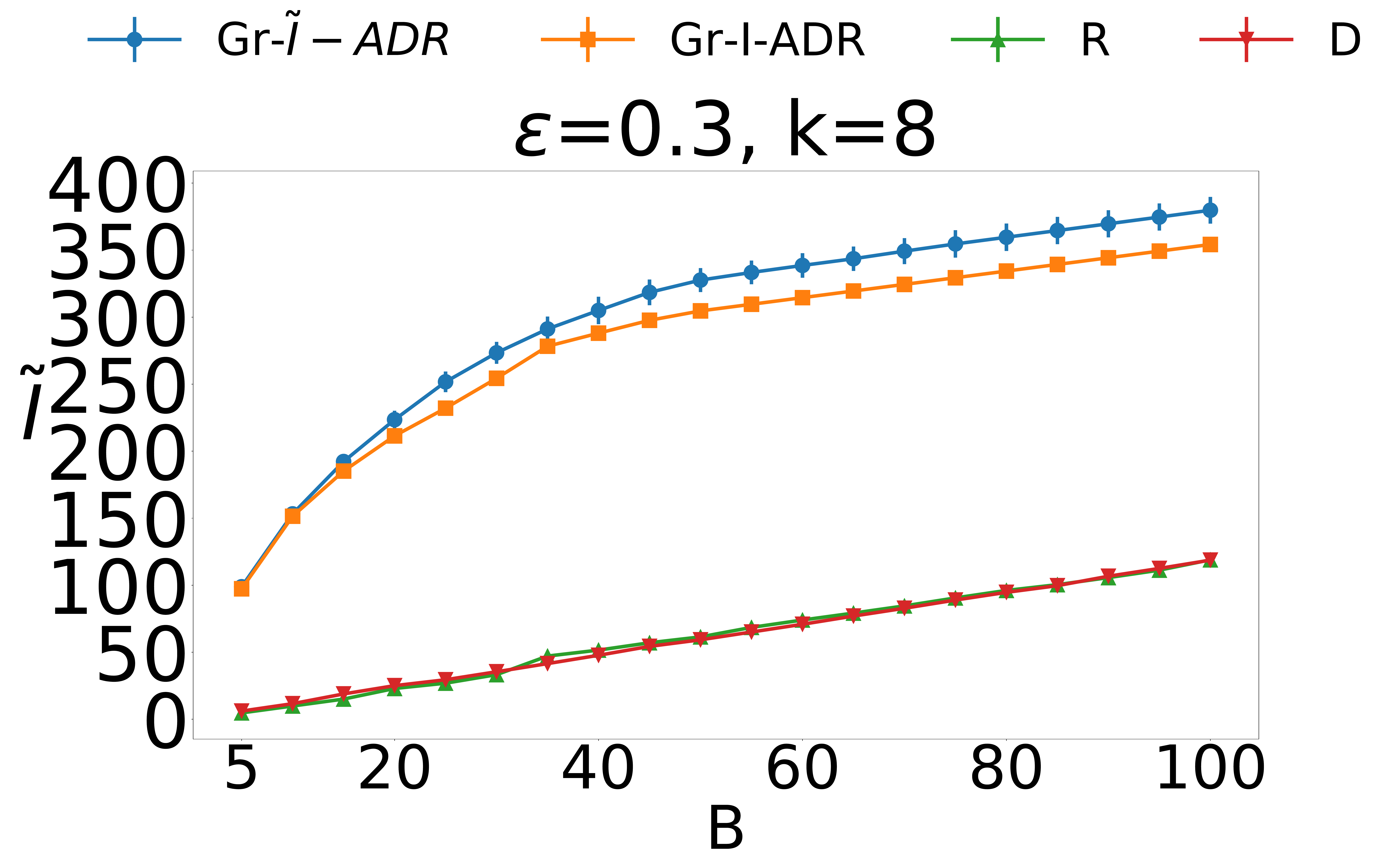}
        \caption{}\label{fig:imMax1}
    \end{subfigure}%
    \begin{subfigure}{0.17\textwidth}
        \centering
        \captionsetup{justification=centering}
        \includegraphics[width = \linewidth]{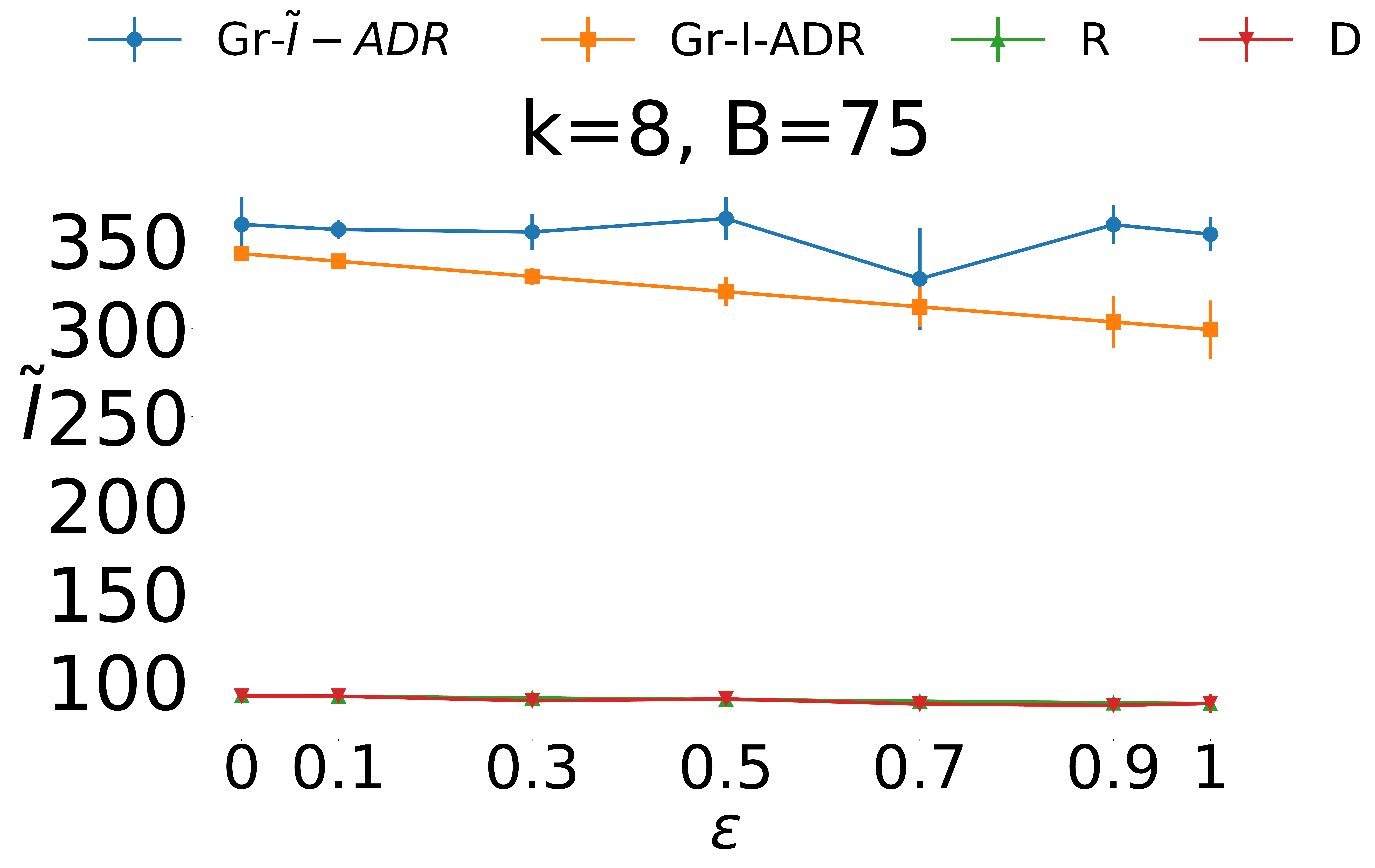}
        \caption{}\label{fig:imMax2}
\end{subfigure}%
    \caption{$\tilde{I}$ in MaxG setting vs:  (a) k, (b) $B$, (c) $\varepsilon$.}\label{fig:imMaxAllADR}
\end{figure}

\end{document}